\documentclass{article}

\usepackage[utf8]{inputenc}
\usepackage[margin=3cm]{geometry}
\usepackage[T1]{fontenc}
\usepackage{lmodern}
\usepackage[english]{babel}
\usepackage{amsmath,amssymb,amsthm}
\usepackage{mathrsfs,dsfont} %fonts
\usepackage[dvipsnames,table]{xcolor}
\usepackage[bibstyle=numeric, sorting=none, citestyle=numeric-comp, giveninits, maxbibnames=4, url=false, isbn=false]{biblatex} % loads etoolbox
\usepackage{mathtools, ifthen, xparse}
\usepackage[shortlabels]{enumitem}
\usepackage{orcidlink}
\usepackage{subcaption}
\usepackage{apptools}\AtAppendix{\counterwithin{theorem}{section}\counterwithin{equation}{section}}
\usepackage{tikz}
\usepackage{csquotes}\MakeOuterQuote"
\usepackage{pgfplots}\pgfplotsset{compat=newest}
\usepackage{todonotes}
\usepackage{hyperref}
\usepackage[capitalize]{cleveref}
\hypersetup{
pdftitle = {energy-limited-quantum-dynamics},
pdfauthor = {van-Luijk}
verbose=false,
colorlinks,
citecolor=MidnightBlue,
linkcolor=MidnightBlue,
urlcolor=OliveGreen
}

\newtheorem{theorem}{Theorem}
\newtheorem{letterthm}{Theorem} %letter-labeled thm env for intro
\newtheorem*{theorem*}{Theorem}

\newtheorem{letterlemma}[letterthm]{Lemma}
\newtheorem{lemma}[theorem]{Lemma}
\newtheorem{corollary}[theorem]{Corollary}
\newtheorem{proposition}[theorem]{Proposition}

\newtheorem*{problem*}{Problem}
\newtheorem{definition}[theorem]{Definition}

\theoremstyle{definition}
\newtheorem{example}[theorem]{Example}
\newtheorem{remark}[theorem]{Remark}
\theoremstyle{remark}
\newtheorem*{example*}{Example}
\newtheorem*{remark*}{Remark}

\crefname{letterthm}{Thm.}{Thms.}
\crefname{letterlemma}{Lem.}{Lems.}
\crefname{theorem}{Thm.}{Thms.}
\crefname{corollary}{Cor.}{Cors.}
\crefname{lemma}{Lem.}{Lems.}
\crefname{proposition}{Prop.}{Props.}
\crefname{definition}{Def.}{Defs.}
\crefname{problem}{Problem}{Problems}
\crefname{remark}{Rem.}{Remarks}
\crefname{example}{Ex.}{Examples}
\crefname{section}{Sec.}{Sections}
\Crefname{Section}{Section}{Sections}

\makeatletter
\def\blfootnote{\xdef\@thefnmark{}\@footnotetext}
\makeatother

\newcommand*{\1}{\text{\usefont{U}{bbold}{m}{n}1}}
\renewcommand{\bar}[1]{\overline{#1}}

\newcommand\CC{\mathbb C}

\newcommand\D{\mathcal D}
\newcommand\eps\varepsilon

\renewcommand\H{\mathcal H}
\newcommand\K{\mathcal K}

\newcommand\NN{\mathbb N}
\def\placeholder{\,\cdot\,} %{{\hspace{1pt}\,{\mathbin{\vcenter{\hbox{\scalebox{0.5}{$\bullet$}}}}}\hspace{1pt}\,}} % placeholder

\newcommand\RR{\mathbb R}

\newcommand\restrictedto\upharpoonright
\newcommand{\states}{{\mathfrak S}}

\DeclareMathOperator\Ad{Ad}
\DeclareMathOperator\ad{ad}

\DeclareMathOperator\tr{tr}
\renewcommand\Re{\operatorname{Re}}
\newcommand\oo\infty
\newcommand\ox\otimes

\DeclareMathOperator{\Ran}{Ran}
\newcommand{\TC}{\mathcal{T}(\mathcal{H})}
\newcommand\boundedoperatorsymbol{{\mathcal B}}
\NewDocumentCommand\BO{mo}{ \IfNoValueTF{#2}{\boundedoperatorsymbol(#1)}{\boundedoperatorsymbol(#1,#2)}}

\let\mc\mathcal
\let\mf\mathfrak 

\def\trp{^{\mathsf T}}
\def\Lmin{\L_{\textrm{min}}}

\newcommand\mat[1]{ \begin{pmatrix} #1 \end{pmatrix} }

\renewcommand\limsup\varlimsup
\renewcommand\liminf\varliminf
\DeclareMathOperator\id{id}

\DeclareMathOperator\lin{span}

\DeclareMathOperator{\Sp}{Sp}
\DeclareMathOperator{\dom}{dom}
\newcommand\op{\textit{op}}

\newcommand\order{\mathbf{o}}
\newcommand\Order{\mathbf{O}}
\let\lim\relax
\NewDocumentCommand\lim{o}{\IfNoValueTF{#1}{\mathop{\textup {lim}}}{\mathop{\textup{{$ #1 $}-lim}}}}

\DeclareFontFamily{U}{matha}{\hyphenchar\font45}
\DeclareFontShape{U}{matha}{m}{n}{ <-6> matha5 <6-7> matha6 <7-8> matha7 <8-9> matha8 <9-10> matha9 <10-12> matha10 <12-> matha12 }{}
\DeclareSymbolFont{matha}{U}{matha}{m}{n}
\DeclareFontFamily{U}{mathx}{\hyphenchar\font45}
\DeclareFontShape{U}{mathx}{m}{n}{ <-6> mathx5 <6-7> mathx6 <7-8> mathx7 <8-9> mathx8 <9-10> mathx9 <10-12> mathx10 <12-> mathx12 }{}
\DeclareSymbolFont{mathx}{U}{mathx}{m}{n}
\DeclareMathDelimiter{\vvvert} {0}{matha}{"7E}{mathx}{"17}%
\DeclarePairedDelimiterX{\normiii}[1]{\vvvert}{\vvvert} {\ifblank{#1}{\:\cdot\:}{#1}}

\DeclarePairedDelimiter\paren\lparen\rparen

\DeclarePairedDelimiterX\norm[1]\lVert\rVert{\ifblank{#1}{\placeholder}{#1}}
\DeclarePairedDelimiter\abs\lvert\rvert

\DeclarePairedDelimiterX\ip[2]{\langle}{\rangle}{#1 , #2}
\DeclarePairedDelimiterX\dual[2]{\langle}{\rangle}{#1 , #2}

\newcommand\qandq{\quad\text{and}\quad}

\DeclarePairedDelimiter\ket\vert\rangle

\DeclarePairedDelimiterX\braket[2]\langle\rangle{ #1 \delimsize\vert #2}
\DeclarePairedDelimiterX\braAket[3]\langle\rangle{ #1 \delimsize\vert #2 \delimsize\vert #3 }
\DeclarePairedDelimiterX\ketbra[2]\vert\vert{ #1 \delimsize\rangle\delimsize\langle #2 }
\DeclarePairedDelimiterX\kettbra[1]\vert\vert{ #1 \delimsize\rangle\delimsize\langle #1 }

\DeclarePairedDelimiterX{\set}[1]\{\}{  #1 }

\newcommand\hide[1]{}

\def\B{{\mathcal B}}

\def\M{{\mathcal M}}

\def\N{{\mathcal N}}
\def\R{{\mathcal R}}

\def\B{{\mathcal B}}

\def\K{{\mathcal K}}
\def\L{{\mathcal L}}
\def\T{{\mathcal T}}
\def\U{{\mathcal U}}
\def\P{{\mathcal P}}

\def\Sch{{\mathcal S}}

\newcommand{\energy}{\mathbf E}
\newcommand\grp{{\Gamma}}
\newcommand\lie{{\mf g}}

\numberwithin{equation}{section}
\numberwithin{theorem}{section}

\title{\vspace{-11pt}Energy-limited quantum dynamics} 

\author{Lauritz van Luijk\,\orcidlink{0000-0003-3153-549X}}

\date{\normalsize Institut für Theoretische Physik, Leibniz Universiät Hannover, \\ Appelstraße 2, 30167 Hannover, Germany\\[2ex]\today \blfootnote{Email: \href{mailto:lauritz.vanluijk@itp.uni-hannover.de}{lauritz.vanluijk@itp.uni-hannover.de}}}

\begin{document}

\maketitle
\vspace{-22pt}
\begin{abstract}
    We consider quantum systems with energy constraints.
    In general, quantum channels and continuous-time dynamics need not satisfy energy conservation. 
    Physically meaningful channels, however, can only introduce a finite amount of energy to the system, and continuous-time dynamics may only increase the energy gradually over time.
    We systematically study such "energy-limited" channels and dynamics. For Markovian dynamics, energy-limitedness is equivalent to a single operator inequality in the Heisenberg picture.
    We observe new submultiplicativity inequalities for the energy-constrained diamond and operator norm.
    Together, our results prove a powerful toolkit for quantitative analyses of dynamical problems in finite and infinite-dimensional systems. As an application, we derive state-dependent bounds for quantum speed limits that outperform the usual diamond/operator norm estimates, which have to account for fluctuations in high-energy states.
\end{abstract}

\tableofcontents

\section{Introduction}\label{sec:intro}

When we want to model a quantum system, we begin by describing the Hilbert space.
In quantum information theory and related areas, models usually use finite-dimensional Hilbert spaces, whereas infinite-dimensional models dominate in quantum optics, statistical mechanics, and quantum field theory.
In models with infinite-dimensional Hilbert space, the Hamiltonian is typically an unbounded operator, which necessarily implies that the system's Hilbert space contains state vectors of infinite energy.
Infinite energy states are often discarded as unphysical.
The reason is that the models we study are only valid in certain regimes, which never contain arbitrarily large energies.
For instance, a laser might be modeled by a single bosonic mode $\H=L^2(\RR)$ with Hamiltonian $H=\Omega a^\dagger a$.
At sufficiently large energies, the lab will catch fire -- a physical effect not accounted for in the model.\footnote{This is assuming poor safety conditions. The more realistic scenario is that a fuse will pop out, causing the laser to turn off. In any case, the model breaks down at large energies.}
The naive solution is to introduce a strict energy cutoff by truncating the model onto the spectral subspace with energy below some threshold energy. This, however, has two problems:
First, the truncation completely butchers algebraic relations between observables, and second, the resulting model is sensitive to the cutoff energy in a discontinuous way.
Instead, it is better to keep the full Hilbert space but to introduce an \emph{energy constraint}, which means that we only consider states whose mean energy does not exceed a fixed threshold.
This way, the observables remain untouched and the resulting theory depends smoothly on the chosen threshold energy.

The role of the Hamiltonian in the above is to determine the energy scale.
We separate this from its role as the generator of dynamics by considering systems equipped with a specified \emph{reference Hamiltonian}, which may or may not be the generator of the system's unitary time evolution.
This is not a mere mathematical generalization but is important in applications.
Take, for instance, a laser coupled to an atom. 
While the dynamics is interacting, we are still interested in the energy of the laser itself, i.e., the mean photon number, which corresponds to the reference Hamiltonian $a^\dagger a$.
In particular, the idea of reference Hamiltonians makes sense in open systems whose dynamics are not generated by a Hamiltonian to begin with.
In the following, we consider open or closed quantum systems with energy constraints relative to reference Hamiltonians. 

A good understanding of a model requires not only the analysis of specific states but also statements concerning all states.
For instance, the Heisenberg uncertainty principle states that the standard deviations of position and momentum measurements satisfy the trade-off inequality
\begin{equation}
    \Delta p\cdot \Delta q \ge \frac\hbar2
\end{equation}
for \emph{all} states of the system. 
While statements for all states on an infinite-dimensional Hilbert space are nice, it suffices to consider states satisfying the energy constraint.
Let us consider another example.
In a qudit system, i.e., $\H=\CC^d$, the Fannes-Audenaert inequality \cite{audenaert_sharp_2007} asserts the continuity bound 
\begin{equation}\label{eq:FA_ineq}
    |S(\rho)-S(\sigma)| \le \eps\log d + h(\eps),
\end{equation}
for the von Neumann entropy of arbitrary states $\rho$ and $\sigma$, where $\eps=\frac12\norm{\rho-\sigma}_1$ is the trace-distance and is $h(\eps)$ is the binary entropy. 
When the dimension $d$ becomes larger, the continuity bound \eqref{eq:FA_ineq} diverges. 
In fact, the von Neumann entropy is discontinuous on the full state space of an infinite-dimensional Hilbert space.
However, if we take seriously the idea of an energy constraint and restrict to states with bounded energy, the von Neumann entropy does become continuous, provided the reference Hamiltonian has a finite partition function $\mc Z = \tr e^{-\beta H}<\oo$.
Indeed, Winter generalized the continuity bound \eqref{eq:FA_ineq} to this setting \cite{winter2016}.
We see that imposing an energy constraint yields a refined understanding of systems described by infinite-dimensional Hilbert spaces.

Energy constraints are widely used in classical and quantum information theory, where they appear in the study of continuous variable systems.
The basic idea is that in communication setups involving continuous signals, only a limited amount of energy is available.
The relevant quantity is then the \emph{energy-constrained capacity} of a channel, i.e., the amount of information that can be communicated through a given channel using input signals with bounded energy.
This idea, developed in Shannon's ground-laying work \cite{shannon_mathematical_1948}, is still used quantum information theory today \cite{eisert2007gaussian,noh_enhanced_2020,davis_energy-constrained_2018,Winter_EnergyConstrDiamond_2017,shirokov2018,pirandola_advances_2020}.

In the presence of energy constraints, quantifying distance in terms of the operator or diamond norm has little significance.
Indeed, these norms are defined by optimizing the norm distance over the full state space and, hence, have to account for errors on infinite-energy states.
By restricting to states with bounded energy expectation, Shirokov and Winter introduced energy-constrained versions of these norms \cite{Winter_EnergyConstrDiamond_2017,shirokov2018} that (a) have an operational interpretation in terms of distinguishability subject to an energy constraint, (b) induce a topology independent of the threshold energy and (c) restore good properties lost in the transition from finite to infinite-dimensions.
Let us give an example for the third aspect: Since Hamiltonians in infinite dimensions are typically unbounded, their unitary dynamics $U(t) = e^{-itH}$ are not operator norm continuous in $t$ but merely strongly continuous.
Norm continuity is, however, restored by the energy-constrained operator norm, which metrizes the strong topology on bounded subsets \cite{shirokov2020Enorm}.
Since these energy-constrained norms were introduced, they have been used to obtain convergence rates and continuity bounds in various physical settings ranging from speed limits to channel capacities \cite{becker2020,becker2021,lie_group_error,Winter_EnergyConstrDiamond_2017,shirokov2018}.

In this work, we further develop the theory of quantum systems energy constraints. 
Building on the works of Shirokov and Winter \cites{Winter_EnergyConstrDiamond_2017,shirokov2018,shirokov2020Enorm,weis_extreme_2021}, we systematically study quantum channels and dynamics that are compatible with the energy scale of the system.
In particular, we provide tools to estimate the maximal output energy at a given input energy constraint as a function of time. This solves an open problem suggested by Becker and Datta in \cite{becker_convergence_2020}.
We observe submultiplicativity inequalities connecting the energy gain of a quantum channel with the energy-constrained norms of Shirokov and Winter, which enable a quantitative analysis of dynamical limit problems such as quantum speed limits or Trotter products in infinite-dimensional systems.

\subsection{Overview of main results}\label{sec:overview}

We consider quantum systems equipped with specified reference Hamiltonians.
Reserving the letter $H$ for the generator of the unitary time-evolution in closed systems, we follow \cite{shirokov2020Enorm,shirokov2020extension,davis_energy-constrained_2018} in denoting the reference Hamiltonian by $G$.
As an absolute quantity, energy is often meaningless.
Instead, the meaningful quantity is the energy relative to the ground state energy.
We fix this arbitrariness by assuming the ground state energy to be zero.
If the system's Hilbert space $\H$ is infinite-dimensional, we assume that the reference Hamiltonian is an unbounded operator of the form
\begin{equation}\label{eq:discrete}
    G=\sum_{n=0}^\oo \epsilon_n \,\kettbra n, \qquad \lim_{n\to\oo} \epsilon_n = \oo
\end{equation}
for a basis $\{\ket n\}_{n=0}^\oo$ of $\H$. This form is guaranteed if the partition function is finite for all temperatures.
A prototypical example is a bosonic system with $n$ canonical degrees of freedom, where $\H = L^2(\RR^n)$ and where the reference Hamiltonian is the number operator $G=\sum_{i=1}^na_i^\dagger a_i$.

Given a threshold energy $E>0$, the \emph{energy-constrained state space} is defined as
\begin{equation}
    \states_E(\H) = \big\{ \rho\in\states(\H) : \energy[\rho]\le E\big\},\quad E>0, 
\end{equation}
where $\energy[\rho]= \tr[G\rho]$ denotes the energy expectation value in the state $\rho$, instead of the full state space $\states(\H)$.
Let us emphasize that states in $\states_E(\H)$ are only constrained in their mean energy. While it is still possible to measure arbitrarily large energies, the probability of doing so decays sufficiently fast.

In general, quantum channels mapping between systems with reference Hamiltonians need not preserve the energy but may pump energy into or extract energy from the system.
However, they must respect the energy scale.
Different ways to define this mathematically turn out to be equivalent:

\begin{letterlemma}\label{lem:introEL}
    Let $T$ be a quantum channel from system $A$ to $B$. The following are equivalent:
    \begin{enumerate}[(a)]
        \item The output energy is linearly bounded by the input energy: There exist $\lambda, E_0\ge0$, s.t.
        \begin{equation}\label{eq:intro_EL1}
            T^*(G_B) \le \lambda G_A+E_0,
        \end{equation}
        where $T^*$ denotes the dual (Heisenberg-picture) channel.
        \item For all finite-energy input states $\rho$, the output energy is finite $\energy_B[T\rho]<\oo$.
        \item\label{it:introEL2} Given any input energy constraint, the output energy is bounded: 
        \begin{equation}\label{eq:intro_f}
            f_T(E) := \sup_{\rho\in\states_E(\H_A)} \energy_B[T\rho] <\oo.\qquad \text{for all $E>0$}.
        \end{equation}
    \end{enumerate}
\end{letterlemma}

A quantum channel $T$ is called \emph{energy-limited} if it satisfies these equivalent properties \cite{Winter_EnergyConstrDiamond_2017}.
The Lemma shows that if a channel is not energy-limited, then infinite output energies exist even at arbitrarily small input energies.
Thus, physically meaningful channels must be necessarily energy-limited.
The expression $T^*(G_B)$ in \eqref{eq:intro_EL1}, where the dual channel acts on an unbounded operator, is defined as a positive self-adjoint operator using the Stinespring dilation of energy-limited channels (see \cref{sec:channels}).
% The operator inequality \eqref{eq:intro_EL1} is equivalent to an affine upper bound $f_T(E)\le \lambda E+E_0$.
By definition, $f_T(E)$ is the maximal output energy of the channel $T$ if the input energy is constrained by $E$.
It is a concave nondecreasing function of the threshold energy $E$, and can equivalently be characterized as:
% Since concave functions are the pointwise minima of their affine upper bounds,  the following dual characterization follows:
\begin{equation}\label{eq:intro_fT_dual}
    f_T(E) = \min \Big\{ \lambda E+E_0 : \lambda,E_0\ge0 \ \ \text{s.t.}\ \ T^*(G_B) \le \lambda G_A+E_0 \Big\},
\end{equation}
where $T^*$ denotes the dual (Heisenberg picture) channel.
Thus, we can estimate the output energy of a channel $T$ by studying operator inequalities in the Heisenberg picture.

Energy-limited quantum channels behave naturally in the context of the energy-constrained operator and diamond norms of Shirokov and Winter \cite{Winter_EnergyConstrDiamond_2017,shirokov2018,shirokov2020extension}.
We observe that the energy-constrained diamond norm $\norm{}_{\diamond,E}$ satisfies the following submultiplicativity-type estimate with respect to energy-limited channels
\begin{equation}\label{eq:intro_submultiplicativity}
    \norm{ST}_{\diamond,E} \le \norm S_{\diamond,f_T(E)}\le \tfrac{f_T(E)}E \, \norm S_{\diamond,E},
\end{equation}
where $S$ is a $^*$-preserving map, e.g., the difference of two channels, and $T$ is an energy-limited channel.
Similarly, the energy-constrained operator norm satisfies
\begin{equation}\label{eq:intro_submultiplicativity2}
    \norm{AU}_{\op,E} \le \norm A_{\op,f_U(E)} \le \sqrt{\frac{f_U(E)}E} \norm{A}_{\op,E},
\end{equation}
where $A$ is an operator on $\H$, $U$ is a unitary and $f_U(E):=f_{T_U}(E)$ with $T_U(\rho)= U\rho U^*$.
These estimates can be used to lift bounds on the distance of quantum dynamics (or products thereof) from the finite-dimensional case to the infinite-dimensional one. 
Indeed, we apply \eqref{eq:intro_submultiplicativity} and \eqref{eq:intro_submultiplicativity2} to obtain error bounds for quantum speed limits and convergence rates Lie-Trotter products.
These bounds scale with the maximal output energy $f_T(E)$.
Similarly, the continuity bounds on energy-constrained channel capacities obtained in \cite{Winter_EnergyConstrDiamond_2017,shirokov2018} require upper bounds on $f_T(E)$. 
Therefore, we can only obtain sharp estimates if we track the output energy carefully.

The main goal of this paper is to develop a theory of energy-limitedness for continuous-time dynamics.
For the reasons indicated above, it is essential to understand the energy increase, in particular, as a function of time.
Let us begin with general open quantum systems.
We say that a quantum time evolution $\rho\to \rho(t)$ is energy-limited if the output energy is bounded linearly for small times:
\begin{equation}\label{eq:intro_first_order}
    \energy[\rho(t)] \le \energy[\rho] + (\omega t+\order(t))(\energy[\rho]+E_0), \qquad 0<t\approx0,
\end{equation}
for all initial states $\rho$, where the "stability constants" $\omega, E_0$ are state-independent.
We show that this first-order bound, in fact, implies
\begin{equation}
    f_{T(t,s)}(E) \le E+ (e^{\omega(t-s)}-1)(E+E_0), \qquad t\ge s\ge0,
\end{equation}
where $T(t,s)$ is the quantum channel taking $\rho(s)$ to $\rho(t)$ for $t\ge s\ge0$. 
We mostly consider Markovian dynamics, fully described by the semigroup $T(t)$ of quantum channels implementing a time $t$ increment, i.e., $T(t)=T(t+t_0,t_0)$ for $t_0\ge0$.
If $\L$ is the infinitesimal generator of the quantum Markov semigroup $T(t)$, a naive expansion in powers of $t$ formally yields the operator inequality
\begin{equation}\label{eq:intro_LG}
    \L^*(G) \le \omega(G+E_0),
\end{equation}
where $\L^*$ denotes the infinitesimal generator of the Heisenberg-picture dynamics.
However, in infinite dimension, the expression $\L^*(G)$, where the (unbounded) dual generator is applied to an unbounded operator, is a priori not defined.
We carefully address these issues to arrive at our main result:

\begin{letterthm}[Informal]\label{thm:intro_main}
    Let $T(t)$ be a quantum Markov semigroup with generator $\L$.
    The following are equivalent:
    \begin{enumerate}[(a)]
        \item The dynamics is energy-limited with stability constants $\omega,E_0$, i.e., \eqref{eq:intro_first_order} holds.
        \item The operator inequality $\L^*(G)\le \omega(G+E_0)$ holds.
    \end{enumerate}
    In this case, the output energy is bounded by
    \begin{equation}\label{eq:intro_fT_bound}
        f_{T(t)}(E) \le E+ (e^{\omega t}-1)(E+ E_0), \qquad t\,,E>0.
    \end{equation}
\end{letterthm}

% For the precise statement, we refer to  \cref{thm:main} in the main text.
% Importantly, in \cref{thm:intro_main} we do not need to presume any regularity between $T$ and $G$ (relative boundedness of $\L$).

In practice, Markovian dynamics are typically given through a Markovian Master Equation.
Lindblad famously showed that generators of uniformly continuous quantum Markov semigroups are of the form
    $\L(\rho) = K\rho +\rho K^* + \sum_\alpha L_\alpha\rho L_\alpha^*$,
where $K$ and $L_\alpha$ are bounded operators and $\sum_\alpha L_\alpha^*L_\alpha=-K^*-K$ \cite{lindblad}.%
\footnote{This is the standard form of Lindblad \cite{lindblad}, which relates to the GKLS form $\L\rho = -i[H,\rho] + \frac12\{ \sum_\alpha L_\alpha^*L_\alpha,\rho\} + \sum_\alpha L_\alpha\rho L_\alpha^*$ with $H=H^*$ via $K=-iH -\frac12\sum_\alpha L_\alpha^*L_\alpha$. We use Lindblad's version here because it is better suited for the infinite-dimensional setting and also covers non-conservative dynamics \cite{siemon_unbounded_2017,fagnola1999,holevo_dissipative_1996}.}
However, in infinite-dimensional systems, quantum dynamics are hardly ever uniformly continuous.
Generators that are formally given by Lindblad's formula -- with potentially unbounded $K$ and $L_\alpha$ -- are called standard generators \cite{siemon_unbounded_2017}. 
For such generators, the formal inequality \eqref{eq:intro_LG} simply reads
\begin{equation}\label{eq:intro_LG_std}
    K^*G+GK+\sum_\alpha L_\alpha^*GL_\alpha \le \omega(G+E_0).
\end{equation}
For finite-dimensional systems, no issues arise, and one can run a semidefinite optimization algorithm to find stability constants $\omega,E_0$ satisfying \eqref{eq:intro_LG}. 
However, in infinite dimensions, standard generators are quite subtle. For example, they might admit escape to infinity in finite time.
Imposing certain regularity assumptions, we show that \eqref{eq:intro_LG_std} indeed implies energy-limitedness of the quantum Markov semigroup generated by the corresponding standard generator $\L$ (see \cref{thm:standard}).
This allows us to obtain stability constants and check energy-limitedness of Markovian dynamics.

Let us now consider unitary dynamics generated by some Hamiltonian $H$.
We say that the unitary group $U(t)=e^{-itH}$ is energy-limited if 
\begin{equation}\label{eq:intro_fU_bound}
    f_{U(t)}(E) \le E+ (e^{\omega \abs t}-1)(E+E_0), \qquad E>0,\ t\in\RR,
\end{equation}
which may be characterized by a first-order condition similar to \eqref{eq:intro_first_order}.
Notice that we require an upper bound on the output energy in both time directions.
Bounding the energy increase of the backward dynamics $U(-t)$ is the same as bounding the energy loss of the forward dynamics $U(t)$.
For unitary dynamics, \cref{thm:intro_main} takes the form:

% Thus, energy-limitedness of unitary dynamics implies continuity bounds on the energy change:
% If $\psi$ is a state vector with energy $E$ and if $E(t)$ is the energy of $U(t)\psi$, then
% \begin{equation}
%    \abs{E(t)-E} \le \omega \abs t(E+E_0)+\Order(t^2).
% \end{equation}
% Considering both the forward and backward dynamics as dynamical semigroups, we get: 

\begin{letterthm}[Informal]\label{thm:intro_unitary}
    Let $H$ be a self-adjoint operator on $\H$.
    The following are equivalent:
    \begin{enumerate}[(a)]
        \item The unitary  is energy-limited with stability constants $\omega, E_0\ge0$, i.e., \eqref{eq:intro_fU_bound} holds.
        \item The operator inequality $\pm i [H,G]\le \omega(G+E_0)$ holds.
    \end{enumerate}
\end{letterthm}

\paragraph{Energy-limitedness in bosonic systems.}
Consider bosonic systems with $n$ modes where the reference Hamiltonian is the number operator.
All Gaussian quantum channels and Gaussian quantum Markov dynamics are energy-limited (see \cref{sec:gaussian}).
The latter have generators of the form
\begin{equation}
    \L(\rho) = \frac12 \sum_{jk} \Big( m_{jk} \big(R_j [\rho,R_k] +[R_j,\rho]R_k \big) + h_{jk} [R_j R_k,\rho] \Big),
\end{equation}
with matrices $0\le m\in M_{2n}(\CC)$, $h=h\trp\in M_{2n}(\RR)$, where $R$ is the vector of canonical operators.
Stability constants can be computed directly from the matrices $m$ and $h$ (see \cref{sec:quantization}).
Using \cref{thm:intro_unitary}, we establish energy-limitedness of the unitary dynamics generated by coherent state quantizations
\begin{equation}
    H = (2\pi)^{-n} \int_{\RR^{2n}} h(\alpha)\, \kettbra \alpha\,d\alpha,
\end{equation}
of functions $h:\RR^{2n}\to\RR$ with uniformly bounded second derivatives, where $\ket\alpha$, $\alpha\in\RR^{2n}$, denotes the family of coherent states in $L^2(\RR^n)$.
This extends to coupled systems: If $h$ is hermitian matrix-valued with uniformly bounded second derivatives, then $H = \int_{\RR^{2n}} h(\alpha)\ox \kettbra \alpha\,d\alpha$ generate energy-limited unitary dynamics on $L^2(\RR^n;\CC^d)$ (see \cref{prop:coherent_state_quant}).
This class of interacting Hamiltonian includes the quantum Rabi model
\begin{equation}
    H=\Omega a^\dagger a+ g\sigma_x(a^\dagger+a) +\nu \sigma_z
\end{equation}
as well as all other such Hamiltonians with interactions linear in $a$ and $a^\dagger$.
This shows energy-limitedness for the dynamics of most closed-system models considered in quantum optics.

\paragraph{Continuity bounds for closed systems.}
To demonstrate how energy-limited dynamics, the submultiplicativity estimate \eqref{eq:intro_submultiplicativity2} and the energy-constrained operator norm can be used for a quantitative analysis of dynamical problems in closed systems, we consider the quantum speed limit.
Given Hamiltonians $H_1$, $H_2$ and a pure state $\psi$, we seek an upper bound on $\norm{e^{-itH_1}\psi-e^{-itH_2}\psi}$ for small times.
Using the operator norm and the usual integration-differentiation trick, we get an upper bound $\abs t \norm{H_1-H_2}$ that works for all states $\psi$. 
In large systems, the operator norm on the right-hand side can be huge. In infinite-dimensional systems, it is typically infinite, making the upper bound useless.
When the state $\psi$ is not known, it is often concluded that this operator norm estimate is optimal since the operator norm bound is always tight in first order on some state $\psi$.
However, if we can bound the energy of the system, we can use this knowledge to get a better bound valid for all states whose energy is in agreement with our estimation:
Indeed, we show
\begin{equation}\label{eq:intro_qsl}
    \norm{e^{-itH_1}\psi-e^{-itH_2}\psi} \le \abs t \norm{H_1-H_2}_{\op,f_t(E)}
\end{equation}
where $E$ is the energy of the state $\psi$, $f_t(E)=E+(e^{\omega t}-1)(E+E_0)$ and $\omega,E_0$ are stability constants for one of the two dynamics.
Note that the right hand side equals $ \abs t \norm{H_1-H_2}_{\op,E}$ up to an error $\Order(t^2)$.
In the infinite-dimensional case, \eqref{eq:intro_qsl} requires mild regularity assumption (see \cref{prop:qsl1} for details).
If $\psi$ is a low-energy state and if the dynamics of $H_1$ and $H_2$ are energy-limited, this bound significantly outperforms the operator norm bound because $\norm{H_1-H_2}_{\op,E}$ is then much smaller than the operator norm. 
This is confirmed by our numerics (see \cref{fig:speed_lim}).
The same techniques are used in \cite{becker_convergence_2024} to derive the following convergence rates for the Trotter product formula
\begin{equation}\label{eq:unitary_trotter_rate}
    \norm{\big(e^{i\frac tnH_1}e^{i\frac tnH_2}\big)^n\psi-e^{it(H_1+H_2)}\psi}
    \le \frac{t^2}{2n} \norm{[H_1,H_2]}_{\op,f_{2t}(E)},
\end{equation}
where $f_t$ is defined as above with $\omega,E_0$ joint stability constants (see \cite[Thm.~3.5]{becker_convergence_2024} for details).

\begin{figure}[htp!]
\usepgfplotslibrary{groupplots}
\centering
\begin{tikzpicture}
\pgfplotsset{width=\textwidth, height=0.8\textwidth,legend cell align={left}, legend columns=3, label style={font=\footnotesize}, legend style={font=\footnotesize}, tick label style={font=\footnotesize},}
\begin{groupplot}[group style={group size= 2 by 1},height=5.1cm,width=6.5cm]
    \nextgroupplot[xlabel= time,ylabel=norm distance, scaled ticks=false, tick label style={/pgf/number format/fixed},xmin=0,xmax=.6,ymin=0,ymax=2,legend to name=zelda]
        \addplot[thick,color=BurntOrange] table[x=time,y=actualError,col sep=comma] {energyL.csv}; \addlegendentry{actual error\phantom a}; 
        \addplot[thick,color=OliveGreen] table[x=time,y=energyBound,col sep=comma] {energyL.csv}; \addlegendentry{our bound\phantom a};
        \addplot[thick,color=MidnightBlue] table[x=time,y=uniformBound,col sep=comma] {energyL.csv}; \addlegendentry{op.\ norm bound};
    \nextgroupplot[xlabel= time, scaled ticks=false, tick label style={/pgf/number format/fixed}, yticklabel=\empty,xmin=0,xmax=.6,ymin=0,ymax=2]
        \addplot[thick,color=BurntOrange] table[x=time,y=actualError,col sep=comma] {energyR.csv}; 
        \addplot[thick,color=OliveGreen] table[x=time,y=energyBound,col sep=comma] {energyR.csv}; 
        \addplot[thick,color=MidnightBlue] table[x=time,y=uniformBound,col sep=comma] {energyR.csv}; 
\end{groupplot}\end{tikzpicture}
\\[5pt]\ref*{zelda}
    \caption{Numerical comparison of $\norm{e^{-itH_1}\psi-e^{-itH_2}\psi}$, the first order of our bound \eqref{eq:intro_qsl}, and the operator norm bound $t\norm{H_1-H_2}$ for the quantum speed limit problem.
    The system is $\H=(\CC^2)^{\otimes 7}$ with reference Hamiltonian $S_x^2+S_y^2+S_z^2$ minus its ground state energy, where $S_j = \sum_{k=1}^7 1^{\ox k-1}\ox \sigma_j\ox1^{\ox N-k}$.
    To obtain a state with relatively small energy, we take a weighted superposition $\psi = c(\Omega +\frac12\phi)$ of the ground state $\Omega$ and a Haar randomly chosen state $\phi$ ($c$ is a normalizing constant).
    On the left, the Hamiltonians are $H_1=S_x+R_1$ and $H_2=S_y+R_2$, where $R_1$ and $R_2$ are random hermitian matrices of operator norm $\norm{R_i}=\tfrac12$. These generate little energy, and we see that our bound \eqref{eq:intro_qsl} is much better than the operator norm bound.
    On the right, we have $H_1=S_x$ and $H_2=R$ is a random hermitian matrix with $\norm{R}=\norm{S_x}=7$. Even though $H_2$ generates a lot of energy, our bound is still better than the operator norm bound, but the benefit is not that large.}
    \label{fig:speed_lim}
\end{figure}

State-dependent continuity bounds such as \eqref{eq:intro_qsl} or \eqref{eq:unitary_trotter_rate} are of interest in both finite-dimensional and infinite-dimensional systems.
To apply them, one needs to identify a reference energy scale so that the two dynamics do not generate too much energy and so that the given state $\psi$ has low energy.
The energy-constrained operator norm appearing on the right-hand side can be estimated through the semidefinite minimization problem (see \cref{thm:dual_prob_Enorm})
\begin{align}\label{eq:intro_sdp}
    \norm{A}_{\op,E}^2 &= \min \big\{\lambda E + E_0 : \lambda,E_0\ge0\ \text{s.t.}\ A^*A \le \lambda G+E_0\big\}.
\end{align}

\paragraph{Continuity bounds for open systems.}

By similar techniques, the results of the previous paragraph can also be obtained for open quantum systems, where we use the energy-constrained diamond norm.
For instance, we derive that if $\L_1$ and $\L_2$ are generators of quantum Markov semigroups, then
\begin{equation}\label{eq:intro_qsl2}
    \norm{e^{t\L_1}-e^{t\L_2}}_{\diamond,E} \le t\norm{\L_1-\L_2}_{\diamond,f_t(E)},
\end{equation}
with $f_t$ as above for stability constants $\omega,E_0$ for one of the two dynamics.
The proof for convergence rates of the Trotter product formula in \cite{becker_convergence_2024} can be adapted to open quantum systems, giving
\begin{align}\label{eq:intro_trotter_bound}
    \norm{\big(T_1(\tfrac tn)T_2(\tfrac tn)\big)^n - T(t)}_{\diamond,E} 
    \le \frac{t^2}{2n}\norm{[\L_1,\L_2]}_{\diamond,f_{2t}(E)} 
\end{align}
where $\omega,E_0$ need to be joint stability constants (see \cref{sec:trotter} for details).

Finally, we mention that the methods developed here solve an open problem posed by Becker and Datta \cite{becker_convergence_2024}, which asks for methods to estimate the maximal output energy at a given energy constraint in open quantum systems.
Together with the results of \cite{becker_convergence_2024} it is then possible to bound the rate at which information can spread in continuous variable systems, as explained in \cite[Sec.~8]{becker_convergence_2024}.
The upper bound on the maximal output energy is necessary to apply the continuity bounds for energy-constrained channel capacities due to Shirokov and Winter \cite{shirokov2018,Winter_EnergyConstrDiamond_2017}.

\paragraph{Acknowledgements.}
The author warmly thanks Simon Becker, Daniel Burgarth, Alexander Hahn, Niklas Galke, Robert Salzmann, Maksim Shirokov, Alexander Stottmeister, Reinhard F.\ Werner, Henrik Wilming, and Andreas Winter for many interesting discussions, spotting errors, and helpful comments.
The author acknowledges funding by the MWK Lower Saxony via the Stay Inspired Program (Grant ID: 15-76251-2-Stay-9/22-16583/2022).

\paragraph{Declarations.}
No data was generated or processed in this work.

\paragraph{Notations and conventions.}
We do not use Dirac notation, with the exception that we write $\ketbra\psi\phi$ for the linear operator $\xi\mapsto \ip\phi\xi\,\psi$ and use $\ket n$ to denote orthonormal bases.
The algebra of bounded operators on a Hilbert space $\H$ is denoted $\B(\H)$, and the trace-class is denoted $\T(\H)$.
The operator and trace norm are denoted $\norm\placeholder$ and $\norm\placeholder_1$, respectively.
We use the convention that the trace, denoted "$\tr$", is defined on all positive operators but may be infinite.
The set of density operators, i.e., positive operators with unit trace, is denoted $\states(\H)$.
The domain of an unbounded operator $A$ will be denoted $\dom A$, and the graph norm on $\dom A$ is denoted $\norm\psi_A = \sqrt{\norm\psi^2 + \norm{A\psi}^2}$.
Positive cones of ordered vector spaces $(X,\le)$ are denoted $X^+$.
The algebraic tensor product of topological spaces is denoted with the symbol "$\odot$" to distinguish it from Banach space tensor products.

\section{Quantum systems with energy reference}\label{sec:kin}

\subsection{setup}\label{sec:setup}

We present and extend the kinematical setup of quantum systems with reference energy scales, which was systematically developed by Winter and, especially, Shirokov \cite{Winter_EnergyConstrDiamond_2017,shirokov2019,shirokov2020Enorm,weis_extreme_2021}. 

A \emph{reference Hamiltonian} for a quantum system described by a Hilbert space $\H$ is a self-adjoint positive operator $G\ge0$ on $\H$ with vanishing ground state energy:
\begin{equation}\label{eq:ground state_energy}
    \inf\Sp(G) = 0.
\end{equation}
We can assume it without loss of generality because we are not interested in absolute energy but rather in the energy relative to the ground state energy.
In general, we do not assume the reference Hamiltonian to be discrete but add this as an extra assumption if needed.
The energy of a state $\rho\in\states(\H)$ is given by
\begin{equation}\label{eq:E_def}
    \energy[\rho] := \lim_n \tr[P_n G\rho] \in \bar\RR^+,
\end{equation}
where $P_n$ is the spectral projection of $G$ onto the interval $[0,n]$.\footnote{The naive definition "$\tr[\rho G]$" via an orthonormal basis cannot be applied because $\rho G$ is an unbounded operator.}
The \emph{energy-constrained state space} is then defined as
\begin{equation}
    \states_E(\H) := \big\{\rho\in\states(\H) : \energy[\rho]\le E\big\},\qquad E>0.
\end{equation}
Note that $\states_E(\H)$ is a convex set,  monotonically increasing in $E$.
The set of all finite-energy states is denoted
\begin{equation}
    \states_{<\oo}(\H):= \big\{\rho\in\states(\H): \energy[\rho]<\oo\big\} = \bigcup_{E>0}\states_E(\H).
\end{equation}
Note that a state $\rho$ has finite energy if and only if $\sqrt G\rho\sqrt G\in\TC$.
It will be useful to extend $\energy$ to a linear functional, the \emph{energy functional}, on the domain 
\begin{equation}\label{eq:domE}
    \dom\energy  := \lin \states_{<\oo}(\H)=(G+1)^{-\frac12} \TC (G+1)^{-\frac12}
\end{equation}
via
\begin{equation}\label{eq:energy_functional}
    \energy[\rho] = \tr\!\big[\sqrt G\rho\sqrt G\big], \qquad \rho\in\dom\energy.
\end{equation}
On positive elements $0\le\rho\in \dom\energy$, this definition agrees with \eqref{eq:E_def}. A positive operator $\rho\in\TC^+$ is in $\dom\energy$ if and only if it is proportional to a finite-energy state.
A rank one operator $\ketbra\psi\phi$ is in $\dom\energy$ if and only if $\psi,\phi\in\dom\sqrt G$.
Abusing notion, we shall write $\energy[\psi]$ for $\energy[\ketbra\psi\psi]$ for vectors $\psi\in\H$.
Note that 
\begin{equation}\label{eq:energy_functional_psi}
    \energy[\psi]=\ \begin{cases} \ \norm{\sqrt G\psi}^2,& \qquad\text{if } \psi\in\dom\sqrt G\\[5pt] \ +\oo,&\qquad\text{else.}\end{cases}
\end{equation}
With eq.~\eqref{eq:E_def} we can define $\energy[\rho]$ in $\bar\RR^+$ for general $\rho\in\TC^+$. 
The situation is similar to that of the integral in Lebesgue theory: The energy functional makes sense either on the cone of general positive elements (cp.\ positive measurable functions) where it may be infinite or on the linear span of the finite-energy states (cp.\ the $L^1$ space).
A convenient fact that we use many times throughout this work is the \emph{lower semicontinuity} of $\energy$ on $\TC^+$:
\begin{equation}\label{eq:lsc}
    \energy[\lim_n\rho] \le \liminf_n \energy[\rho_n]
\end{equation}
for all norm convergent sequences $(\rho_n)$ of positive trace-class operators. 
This follows directly from \eqref{eq:E_def}, which expresses $\energy$ as a pointwise supremum of linear functions.

\begin{lemma}\label{lem:elementary}
\begin{enumerate}[(1)]
    \item\label{it:elementary1}
        For all $E_0>0$, $\dom\sqrt G = \dom\sqrt{G+E_0}$.
        If a subspace $\D\subseteq \dom G$ is a core for $G$, it is also a core for $\sqrt G$.
    \item\label{it:elementary2}
        Let $\rho$ be a state and let $\rho=\sum_\alpha\lambda_\alpha\kettbra{\psi_\alpha}$ be any (countable) decomposition into pure states. Then
        \begin{equation}
            \energy[\rho] = \sum_\alpha \lambda_\alpha \energy[\psi_\alpha]
        \end{equation}
        where both sides may be infinite. In particular, each $\psi_\alpha$ has finite energy if $\rho$ does.
    \item\label{it:elementary3}
        Let $(X,\mu)$ be a measure space. If $\rho:X\to\TC^+$ is a measurable (Bochner) integrable map, then
        \begin{equation}\label{eq:elementy3}
            \energy\bigg[\int_X \rho(x) \,d\mu(x)\bigg] = \int_X \energy[\rho(x)]\,d\mu(x)
        \end{equation}
        where both sides may be infinite.
\end{enumerate}
\end{lemma}
\begin{proof}
    We denote by $G_n$ the truncation of $G$ onto the spectral interval $[0,n]$. Note that $G_1\le G_2\le \ldots$ and that $\energy[\rho]=\lim_n \tr[G_n\rho]$ for $\rho\in\dom\energy$.
    The first item is clear for multiplication operators. Thus, the general case follows from the spectral theorem.
    \cref{it:elementary2} follows from the monotone convergence theorem: $\energy[\rho] = \lim_n \tr[G_n\rho] = \lim_n \sum_\alpha\lambda_\alpha \ip{\psi_\alpha}{G_n\psi_\alpha}= \sum_\alpha \lambda_\alpha \energy[\psi_\alpha]$.

    \ref{it:elementary3}:
    Since $\energy$ is lower semicontinuous, the map $x\mapsto \energy[\rho(x)]$ is measurable.
    The functions $f_n:X\to\RR^+$, $f_n(x) = \tr[G_n\rho(x)]$ are measurable and $f_1\le f_2\le \ldots$ is monotonically increasing.
    By definition of the energy functional, the pointwise limit $f(x):=\lim_{n\to\oo} f_n(x)$ is given by $f(x)=\energy[\rho(x)]$. Therefore the monotone convergence theorem implies $\energy[\int \rho(x)\,d\mu(x)]
        =\lim_n \tr[G_n \int \rho(x)\,d\mu(x)]= \lim_{n\to\oo}\int f_n(x)\,d\mu(x) = \int f(x)\,d\mu(x)=\int\energy[\rho(x)]\,d\mu(x)$.
\end{proof}

 Next, we mention two results on the structure of the energy-constrained state spaces.

\begin{lemma}[Shirokov-Weis {\cite{weis_extreme_2021}}]\label{lem:weis}
    The extremal points of $\states_E(\H)$ are pure states, and $\states_E(\H)$ is the closed convex hull of its extreme points.
    If $f:\states_E(\H)\to \RR$ is a lower semicontinuous convex function, then 
    \begin{equation}
        \sup_{\rho\in\states_E(\H)} f(\rho) = \sup_{\substack{\norm\psi=1\\\energy[\psi]\le E}}f(\kettbra\psi), \qquad E>0.
    \end{equation}
\end{lemma}

% Note that the reference Hamiltonian $G$ has compact resolvent, i.e., $(G+z)^{-1}$ is a compact operator for some/all $z\in \CC$ with $\Re z>0$, if and only if it has discrete spectrum with finite multiplicity and no accumulation points.
% This is equivalent to $G$ being of the form \eqref{eq:discrete}.
%\footnote{It might seem that we do not require finite multiplicity of the eigenvalues in \eqref{eq:discrete}. However, this is contained in $\lim_{n\to\oo}\epsilon_n=\oo$ since a sequence can only converge to infinity if it has no finite cluster points.}

\begin{lemma}[Holevo {\cite{Holevo_2003}}]
    Assume that the reference Hamiltonian $G$ is of the form \eqref{eq:discrete}.\footnote{This is the case if and only if $G$ has compact resolvent if and only if the spectrum is discrete with finite multiplicity.}
    Then, the energy-constrained state space $\states_E(\H)$ is compact in the trace-norm topology for all $E>0$.
\end{lemma}

Both of these Lemmas refer to the trace norm topology. In addition, the energy scale induces two natural norms on $\dom\energy$: the $\energy$-graph norm $\norm\rho_\energy= \norm\rho_1+\abs{\energy[\rho]}$, and the "base" norm\footnote{
    This norm turns $\dom\energy$ into a so-called "base norm space" with base $K=\{\rho \ge0: \tr\rho+\energy[\rho]=1\}$  \cite{nagel1974order}. This follow from the isomorphism in \cref{lem:base} below because $\T(\H)$ is a base norm space.}
\begin{equation}\label{eq:base_norm}
    \normiii\rho_1 = \norm{\sqrt{G+1}\rho\sqrt{G+1}}_1.
\end{equation}
These norms agree on positive elements but differ on general self-adjoint elements, where the relation $\norm\rho_1 \le \norm\rho_\energy \le \normiii\rho_1$, $\rho=\rho^*\in\dom\energy$ holds.\footnote{The first inequality is clear. The second one is seen as follows: $\norm\rho_\energy = \tr V|\rho|V^* \le \tr |V\rho V^*| = \normiii\rho_1$, where $V=\sqrt{G+1}$.}
% \begin{enumerate}
%     \item the $\energy$-graph norm $\norm\rho_\energy= \norm\rho_1+\abs{\energy[\rho]}$, and 
%     \item the "base" norm $\normiii\rho_1 = \norm{\sqrt{G+1}\rho\sqrt{G+1}}_1$.\footnote{
%     This norm turns $\dom\energy$ into a so-called "base norm space" with base $K=\{\rho \ge0: \tr\rho+\energy[\rho]=1\}$  \cite{nagel1974order}. {This follow from the isomorphism in \cref{lem:base} below because $\T(\H)$ is a base norm space.}
% \end{enumerate}
% These norms agree on positive elements but differ on general self-adjoint elements, where the relation
% \begin{equation}
%     \norm\rho_1 \le \norm\rho_\energy \le \normiii\rho_1, \qquad \rho=\rho^*\in\dom\energy,
% \end{equation}
% holds.
% The energy functional is a continuous linear functional with respect to both norms.%
% \footnote{The first inequality is clear. The second one is seen as follows: $\norm\rho_\energy = \tr[\sqrt{G+1}|\rho|\sqrt{G+1}] \le \tr|\sqrt{G+1}\rho\sqrt{G+1}| = \normiii\rho_1$.}
A subspace $\D\subset\dom\energy$ is $\normiii\placeholder_1$-dense if and only if $\sqrt G \D\sqrt G$ is dense in $\T(\H)$ ($\sqrt G$ and $\sqrt{G+1}$ are equal up to multiplication by a bounded operator with bounded inverse).
In this case, $\D$ is also $\norm\placeholder_\energy$-dense and, hence, a core for $\energy$.
In many regards, the topology induced by $\normiii\placeholder_1$ is nicer.
For instance, it turns $\dom\energy$ into a Banach space and the finite-energy state space $\states_{<\oo}(\H)$ into a complete metric space, which is false for the $\energy$-graph topology.%
\footnote{If $G$ is unbounded, $\energy$ is not even a closable: Given a sequence $(\gamma_n)$ of finite-energy states with $E_n:=\energy[\gamma_n]\to\oo$ set $\rho_n = (1-E_n^{-1})\sigma +E_n^{-1}\gamma_n$ for some fixed $\sigma\in\states_{<\oo}(\H)$. Then $\rho_n$ and $\energy[\rho_n]$ converge but $\energy[\lim_n\rho_n]\ne\lim_n\energy[\rho_n]$.}

\begin{lemma}\label{lem:base}
    Consider $\dom \energy$ with the $\normiii{}_1$-norm and the positive cone $(\dom\energy)^+=\dom\energy\cap\T(\H)^+$. 
    Set $Z=\sqrt{G+1}$.
    Then $W : \dom\energy\to\T(\H)$, $W\rho=Z\rho Z$, is an isomorphism of ordered Banach spaces.
    \begin{enumerate}[(1)]
        \item\label{it:base1}
            The energy functional $\energy$ is continuous and the energy-constrained state space $\states_E(\H)$ is closed with respect to the $\normiii\placeholder_1$-norm.
            % $\dom\energy$ is a Banach space in the norm $\normiii\placeholder_1$, $(\dom\energy)^+=\dom\energy\cap\T(\H)^+$ is a $\normiii\placeholder_1$-closed positive cone and $\states_E(\H)$ is a $\normiii\placeholder_1$-closed convex set. 
        \item\label{it:base2}
            An increasing sequence $(\rho_n)\subset(\dom\energy)^+$ with $\sup_n\energy[\rho_n]<\oo$ converges in $\normiii\placeholder_1$-norm.
        \item\label{it:base3}
            A subspace $\D\subset \dom\energy$ is $\normiii\placeholder_1$-dense if and only if $\D^+:=\D\cap\T(\H)^+$ is $\normiii\placeholder_1$-dense in $(\dom\energy)^+$. 
            In this case, $K = \D \cap \states(\H)$ is a convex $\normiii\placeholder_1$-dense subset of finite-energy states.
        \item\label{it:base4}
            If $\D\subset\dom\sqrt G$ is a core, then $\D^{\kettbra{}} := \lin \big\{\ketbra\psi\phi : \psi,\phi\in\D\}\subset\dom\energy $ is $\normiii\placeholder_1$-dense.
        \item\label{it:base5} 
            The dual space of $(\dom\energy,\dom\normiii\placeholder_1)$ can be identified with the space of unbounded operators $A$ such that $Z^{-1}AZ^{-1}\in\B(\H)$ under the norm $\normiii{A}_\oo = \norm{Z^{-1}AZ^{-1}}$.
            The dual pairing is given by $(\rho,A) \mapsto \tr \rho A := \tr\big[\big(Z^{-1}AZ^{-1}\big) \big(Z\rho Z\big)\big] $.
    \end{enumerate}
\end{lemma}
\begin{proof} 
    By definition, the positive cone $(\dom\energy)^+$ corresponds precisely to $\T(\H)^+$ via $W$. 
    \Cref{it:base1} holds because $\energy[\rho]=\tr W\rho+\tr\rho$ and because $\states_E(\H)$ is the intersection of $\normiii\placeholder_1$-closed sets $\states_E(\H)=(\dom\energy)^+\cap \energy^{-1}([0,E])\cap \tr^{-1}(\{1\})$.
    \Cref{it:base2,it:base3,it:base4,it:base5} follow by applying the isomorphism $W$ and using standard properties of the trace class.
    % By definition of $\normiii\placeholder_1$ and $\dom\energy$, the operator $W : \dom\energy\to\T(\H)$, $W\rho=\sqrt{G+1}\rho\sqrt{G+1}$ is an isometric isomorphism. 
    % The energy functional $\energy$ is $\normiii{}_1$-continuous because it corresponds to the trace via $W$.
    % The positive cone $(\dom\energy)^+:=\dom\energy\cap\T(\H)^+$ corresponds precisely to the positive cone $\T(\H)^+$ via $W$. 
    % % Therefore, $W$ is an isomorphism of base norm spaces \cite{nagel1974order}.
    %
    % \ref{it:base1}: 
    % Since $W$ is an isometric isomorphism $(\dom\energy,\normiii{}_1)$ is Banach space with closed positive cone because $\T(\H)$ is.
    % The energy-constrained state space is the intersection of $\normiii\placeholder_1$-closed sets $\states_E(\H)=(\dom\energy)^+\cap \energy^{-1}([0,E])\cap \tr^{-1}(\{1\})$.
    % \ref{it:base2}, \ref{it:base3}, and \ref{it:base4} follow similarly by applying the isomorphism $W$ and using standard properties of the trace class.
\end{proof}

\subsection{Energy-limited quantum channels}\label{sec:channels}

In this section, we consider quantum systems $A,B,\ldots$ with Hilbert spaces $\H_A,\H_B,\ldots$ and reference Hamiltonians $G_A$, $G_B,\ldots$, and we denote by $\energy_A,\energy_B,\ldots$ the respective energy functionals.
A combined system $AB$, described by the  Hilbert space $\H_{AB}=\H_A\ox\H_B$, is equipped with the reference Hamiltonian $G_{AB}=G_A\ox1+1\ox G_B$.
Ancillary quantum systems $R$ will be endowed with trivial reference Hamiltonians $G_R=0$ so that the joint Hamiltonian is simply $G_{AR}=G_A\ox 1$ and the energy of a bipartite state $\rho\in\states(\H_{AB})$ is given by the energy of the partial trace $\energy_{AR}[\rho]=\energy_A[\tr_R\rho]$.

Recall that \emph{quantum channels} between systems $A$ and $B$ are mathematically modeled by trace-preserving completely positive (cp) maps $T:\T(\H_A)\to\T(\H_B)$.
In the case of an open system, an effective description sometimes requires the larger class of trace-nonincreasing cp maps.
% A cp map is trace-nonincreasing if and only if it is contractive with respect to the trace norm.
Let us begin by applying Winter's definition of energy-limited quantum channels from \cite{Winter_EnergyConstrDiamond_2017} to general cp maps:

\begin{definition}\label{def:EL}
    % Let $\H_A,\H_B$ be Hilbert spaces with reference Hamiltonians $G_A,G_B$, respectively.
    A cp map $T:\T(\H_A)\to\T(\H_B)$ is {\bf energy-limited} if
    \begin{equation}\label{eq:EL}
        f_T(E) := \sup\{\energy_{B}[T\rho] : \rho\in\states_E(\H_A)\}<\infty, \qquad E>0.
    \end{equation}
    An operator $V:\H_A\to\H_B$ is energy-limited if $T_V\rho = V\rho V^*$ is energy-limited, and we write $f_V$ for $f_{T_V}$.
\end{definition}

In general, we cannot restrict the supremum in \eqref{eq:EL}, which runs over states with energy bounded by $E$, to a supremum over states with energy \emph{equal} to $E$.%
\footnote{For instance, if the reference Hamiltonian is bounded, this cannot hold because no states with energies $E> \max \Sp G$ exist.}
Shirokov observed that energy-limitedness is equivalent to the statement that the output energy is finite whenever the input energy is \cite{shirokov2020extension}:        

\begin{lemma}\label{lem:EL}
    Let $T:\T(\H_A)\to\T(\H_B)$ be completely positive.
    The following are equivalent:
    \begin{enumerate}[(a)]
        \item\label{it:EL1} $T$ is energy-limited, i.e., $f_T(E)$ is finite for all $E>0$,
        \item\label{it:EL2} $f_T(E)$ is finite for some $E>0$,
        \item\label{it:EL3} for all finite-energy input states $\rho\in\states_{<\oo}(\H_A)$, the output energy is finite $\energy_B[T\rho]<\oo$.
    \end{enumerate}
    In this case, the function $f_T:\RR^+\to\RR^+$ is continuous, nondecreasing, and concave.
    Therefore, it holds
    \begin{equation}\label{eq:concavity_ineq_fT}
        f_T(E) \le f_T(E') \le \frac{E'}{E} \, f_T(E), \qquad E'\ge E\ge 0.
    \end{equation}
\end{lemma}
\begin{proof}
    The last statement and \ref{it:EL1} $\Leftrightarrow$ \ref{it:EL2} were observed by Winter in \cite{Winter_EnergyConstrDiamond_2017}, and the equivalence with \ref{it:EL3} is shown in \cite{shirokov2020extension}.
    Since this Lemma is essential to our work, we recall the proofs here:
It is clear that $f_T$ is a nondecreasing function $\RR^+\to \bar\RR^+$.
To see concavity, let $\eps>0$. Given $E_1,E_2>0$ and $0<p<1$, pick $\rho_i \in\states_{E_i}$ with $\energy[T\rho_i]\ge f_T(E_i)-\eps $ and set $\rho=p\rho_1+(1-p)\rho_2$. Then $\energy[\rho]\le pE_1+(1-p)E_2$ implies
\begin{align*}
    f_T(pE_1+(1-p)E_2)&\ge \energy[T\rho] =p \energy[T\rho_1]+(1-p)\energy[T\rho_2]\ge pf_T(E_1)+(1-p)f_T(E_2)-\eps.
\end{align*}
Thus, $f_T$ is a concave nondecreasing function $\RR^+\to\bar \RR^+$.
This implies \ref{it:EL1} $\Leftrightarrow$ \ref{it:EL2} as well as \eqref{eq:concavity_ineq_fT}; see  \cite[Lem.~1]{shirokov2020Enorm} and \cref{fig:sqrt}. 
\ref{it:EL2} $\Rightarrow$ \ref{it:EL3} is clear. For the converse, assume the contrary and take $\rho_n\in\states_E(\H_A)$ with $\energy_B[T\rho_n]\ge 2^n$ and set $\rho = \sum_{n=0}^\oo 2^{-n}\rho_n\in\states_E(\H)$.
By \cref{lem:elementary}, $T\rho$ has infinite energy $\energy_B[T\rho]= \sum 2^{-n} \energy_B[T\rho_n] \ge \sum 1=\oo$, contradicting \ref{it:EL3}.
\end{proof}

\begin{figure}[ht]
    \centering
    % \resizebox{8cm}{!}{%
% \begin{tikzpicture}
% \begin{axis}[axis lines=middle, clip=false, xticklabels=\empty, yticklabels=\empty,yscale=.8]
%     \addplot[thick,MidnightBlue,samples=100,domain=0.005:.2] {sqrt(x)+.5};
%     \addplot[thick,MidnightBlue,samples=100,domain=.2:4.3] {sqrt(x)+.5};
%     \addplot[gray,samples=2,domain=0:4.25] {.8*x};
%     \addplot[dashed,gray] coordinates {(0,.8*2.666) (2.666,.8*2.666)};
%     \addplot[dashed,gray] coordinates {(0,.8*4) (4,.8*4)};
%     \addplot[dashed,gray] coordinates {(0,2.5) (4,2.5)};
%     \addplot[dashed,gray] coordinates {(4,0) (4,3.2)};
%     \addplot[dashed,gray] coordinates {(2.666,0) (2.666, 2.666*.8)};
%     \draw[fill] (axis cs:2.666,0) circle [radius=1pt] node[below] {\small$E$};
%     \draw[fill] (axis cs:4, 0) circle [radius=1pt] node[below] {\small $E'$};
%     \draw[fill] (axis cs:0, .8*2.666) circle [radius=1pt] node[left] {\small $f(E)$};
%     \draw[fill] (axis cs:0, 2.5) circle [radius=1pt] node[left] {\small $f(E')$};
%     \draw[fill] (axis cs:0, 4/2.666*.8 * 2.666) circle [radius=1pt] node[left] {\small $\tfrac{E'}Ef(E)$};
%     \draw[fill] (axis cs:0, 0) circle [radius=1pt] node[below] {$0$};
% \end{axis}
% \end{tikzpicture}
% }

\resizebox{8cm}{!}{%
\begin{tikzpicture}
\begin{axis}[axis lines=middle, clip=false, xticklabels=\empty, yticklabels=\empty, yscale=.8, ymax=2.99]
    \addplot[thick,samples=100,domain=0.005:.2] {sqrt(x)+.2};
    \addplot[thick,samples=100,domain=.2:4.2] {sqrt(x)+.2};
    \addplot[gray,samples=2,domain=0:4.15] {.687*x};
    \addplot[dashed,gray] coordinates {(0,.687*2.666) (2.666,.687*2.666)};
    \addplot[dashed,gray] coordinates {(0,2.2) (4,2.2)};
    \addplot[dashed,gray] coordinates {(0,4*.687) (4,4*.687)};
    \addplot[dashed,gray] coordinates {(4,0) (4,4*.687)};
    \addplot[dashed,gray] coordinates {(2.666,0) (2.666, 2.666*.687)};
    \draw[fill] (axis cs:2.666,0) circle [radius=1pt] node[below] {\small$E$};
    \draw[fill] (axis cs:4, 0) circle [radius=1pt] node[below] {\small $E'$};
    \draw[fill] (axis cs:0, 1.832) circle [radius=1pt] node[left] {\small $f(E)$};
    \draw[fill] (axis cs:0, 2.2) circle [radius=1pt] node[left] {\small $f(E')$};
    \draw[fill] (axis cs:0, 4*.687) circle [radius=1pt] node[left] {\small $\tfrac{E'}Ef(E)$};
    \draw[fill] (axis cs:0, 0) circle [radius=1pt] node[below] {$0$};
\end{axis}
\end{tikzpicture}
}
    \caption{Visualization of the inequality $f(E)\le f(E')\le \frac{E'}Ef(E)$, valid for concave nondecreasing functions $f:\RR^+\to\RR^+$ and $0<E<E'$. The diagonal has slope $f(E)/E$.}\label{fig:sqrt} 
\end{figure}

Note that \eqref{eq:concavity_ineq_fT} implies $\lambda f_T(E)\le f_T(\lambda E)$ for $0<\lambda<1$. Hence, $\energy_B[T\rho] \le f_T(\energy_A[\rho])$ also holds for subnormalized states.
Another consequence is that $E\mapsto f_T(E)/E$ is monotonically decreasing.
Its limit at $E=0$ is the total energy-amplification factor:
\begin{equation}
    \sup_{\rho\in \states_{<\oo}(\H_A)} \frac{\energy_B[T\rho]}{\energy_A[\rho]}=\lim_{E\to0} \frac{f_T(E)}E,
\end{equation}
which may be infinite, e.g., if the ground state is mapped to a state with nonzero energy.
We collect basic properties of the maximal output energy $f_T(E)$ in a Lemma:

\begin{lemma}\label{lem:energy_lim}
    % Let $\H_A,\H_B,\H_C$ be Hilbert spaces with reference Hamiltonians $G_A,G_B, G_C$, respectively, and l
    Let $T:\T(\H_A)\to\T(\H_B)$ be completely positive.
    \begin{enumerate}[(1)]
        \item\label{it:energy_lim1}
            The maximal output energy is attained on pure states:
            \begin{equation}\label{eq:fT_pure_states}
                f_T(E) = 
                \sup\big\{\energy_B[T\kettbra\psi]: \psi\in\H_A,\ \norm\psi=1,\ \energy_A[\psi]\le E\big\}.
            \end{equation}
            In particular, if $V:\H_A\to\H_B$, we have 
            \begin{equation}
                f_V(E) =\sup\big\{\norm{\sqrt G_{\!B} V\psi}^2 : \psi\in\H_A,\  \norm\psi=1,\  \norm{\sqrt G_{\!A}\psi}^2\le E\big\}.
            \end{equation}
        \item\label{it:energy_lim2}
            To compute $f_T$, we may include subnormalized states, i.e.,
            \begin{equation}\label{eq:fT_subnormalized}
                f_T(E)=\sup\big\{\energy_B[T\rho] : \rho\in\T(\H_A)^+,\ \tr\rho\le 1,\ \energy_A[\rho]\le E\big\}.
            \end{equation}
        \item\label{it:energy_lim3} 
            If  $R$ is an ancillary system, then $T\ox\id_R$ is energy-limited and $f_{T\ox\id_R}=f_T$.
        \item\label{it:energy_lim4} 
            Let $T_n:\T(\H_A)\to\T(\H_B)$ be a sequence of energy-limited cp maps such that $T_n\rho\to T\rho$ for all $\rho\in\T(\H_A)$.
            If there exists a common affine upper bound $f_{T_n}(E)\le\lambda E+E_0$ for all $n\in\NN$, then the limit $T$ is also energy-limited and $f_T(E)\le\lambda E+E_0$.
        \item\label{it:energy_lim5}
            If $T$ is trace-nonincreasing and if $S:\T(\H_B)\to\T(\H_C)$ is cp, then $ST:\T(\H_A)\to\T(\H_C)$ is energy-limited if $S$ and $T$ are, and
            \begin{equation}\label{eq:energy_lim5}
                f_{ST}(E)\le f_{S}\big(f_{T}(E)\big), \qquad E>0.
            \end{equation}
        \item\label{it:energy_lim6}
            If $S:\T(\H_A)\to\T(\H_B)$ is energy-limited, then
            \begin{equation}\label{eq:energy_lim6}
                f_{T+S}(E)\le f_T(E)+f_S(E), \qquad E>0.
            \end{equation}
        \item\label{it:energy_lim7}
            Assume that $T$ maps $\dom\energy_A$ into $\dom\energy_B$.
            Then $T$ is energy-limited if and only if the restriction $T:\dom\energy_A\to\dom\energy_B$ is a bounded operator for the respective base norms $\normiii{}_1$.
    \end{enumerate} 
\end{lemma}
\begin{proof}
    The first item follows from \cref{lem:weis}.
    Item \ref{it:energy_lim6} follows from the definition.
    
    \ref{it:energy_lim2}:
    Let $0\ne\rho\in\TC^+$ with $\lambda=\tr\rho\le 1$, $\energy[\rho]=E$ and set $\sigma=\lambda^{-1}\rho\in\states(\H)$. Then \eqref{eq:concavity_ineq_fT} implies
    \begin{equation*}
        \energy_B[T\rho] =\lambda\cdot \energy_B[T\sigma]
        \le \lambda \cdot f_T(E/\lambda) \le f_T(E). 
    \end{equation*}
    Thus, the supremum in \eqref{eq:fT_subnormalized} is bounded by $f_T(E)$. The other inequality holds trivially.

    \ref{it:energy_lim3}: Let $\psi\in\H_{AR}$ be a unit vector with $\energy_{AR}[\psi]=\norm{(\sqrt G_{\!A}\ox1)\psi}^2\le E$, then $\rho=\tr_R\kettbra\psi\in\states_E(\H_A)$.
    Therefore
    \[
        \energy_{BR}[(T\ox\id)\kettbra\psi]=\energy_B[\tr_R(T\ox\id)\kettbra\psi] = \energy_B[T\rho] \le f_T(E).
    \] 
    By item \ref{it:energy_lim1}, optimizing the left-hand side over such $\psi$ gives us $f_{T\ox\id}(E)$ so that equality is proved.
    
    \ref{it:energy_lim4}:
    Using the lower semicontinuity, we find $\energy_B[T\rho]\le \liminf_n\energy_B[T_n\rho] \le \lambda E+E_0$, $\rho\in\states_E(\H_A)$.

    \ref{it:energy_lim5}: 
    If $\rho\in\states_{<\oo}(\H)$, then $\energy_C[ST\rho] \le f_{S}(\energy_B[T\rho]) \le f_{S}(f_{T}(\energy[\rho]))$ where we used item \ref{it:energy_lim2} and that $f_{S}$ is nondecreasing.

    \ref{it:energy_lim7}: 
    Denote by $W_j:\dom\energy_j\to \T(\H_j)$ the isometric isomorphism from \cref{lem:base}.
    Thus, the restriction is bounded if and only if $S=W_B\circ T\circ W_A^{-1}$ is bounded $\T(\H_A)\to\T(\H_B)$. 
    Since the latter is equivalent to $\exists M>0$ : $\tr S\rho = \energy_B[T\sigma]+\tr T\sigma \le M \tr\rho = M (\energy[\sigma]+\tr\sigma)$ for all $\rho\in\T(\H_A)^+$, where $\sigma=W_A^{-1}\rho\in(\dom\energy_A)^+$, the claim follows.
\end{proof}

\begin{proposition}\label{thm:Stinespring}
    % Let $\H_A,\H_B$ be Hilbert spaces with reference Hamiltonians $G_A, G_B$. 
    Let $T:\T(\H_A)\to\T(\H_B)$ be a cp map and let $V:\H_A\to\H_B\ox\H_R$ be a Stinespring dilation, i.e., $T= \tr_R[V(\placeholder)V^*]$.
    If we set $G_R=0$, then $V:\H_A\to\H_B\ox\H_R$ is energy-limited and $f_T=f_V$.
    In particular, the following are equivalent:
    \begin{enumerate}[(a)]
        \item\label{it:StinespringEL} $T$ is energy-limited.
        \item\label{it:Stinespring1} $T$ admits a Stinespring dilation $T= \tr_R[V(\placeholder)V^*]$ with $V:\H_A\to \H_B\ox\H_R$ being energy-limited.
        \item\label{it:Stinespring2} For every Stinespring dilation $T= \tr_R[V(\placeholder) V^*]$, the operator $V:\H_A\to\H_B\ox\H_R$ is energy-limited. 
    \end{enumerate} 
\end{proposition}
\begin{proof}
    $G_R=0$ implies $\energy_{BR}[\phi]= \energy_B[\tr_R\kettbra\phi]$ for $\phi\in\H_B\ox\H_R$.
    Optimizing $\energy_{BR}[V\psi]= \energy_B[T\kettbra\psi]$ over unit vectors $\psi\in\dom\sqrt G$ with $\energy_{AR}[\psi]\le E$ shows $f_T(E)=f_V(E)$.
\end{proof}

This result is implicitly also contained in \cite{shirokov2020extension}, where the Stinespring dilation is used to extend the action of the dual operation $T^*$ to $\sqrt G$-bounded operators. 
As a consequence, we get:
\begin{corollary}
    % Let $\H_A,\H_B$ be Hilbert spaces with reference Hamiltonians $G_A, G_B$. 
    Let $T:\T(\H_A)\to\T(\H_B)$ be an energy-limited cp map. 
    Then there are energy-limited operators $K_\alpha:\H_A\to\H_B$ such that
    \begin{equation}\label{eq:Kraus}
        T\rho = \sum_\alpha K_\alpha \rho K_\alpha^*,\qquad\rho\in\T(\H_A).
    \end{equation}
    % In particular, each $K_\alpha$ maps $\dom\sqrt G_{\!A}$ into $\dom\sqrt G_{\!B}$.
\end{corollary}

In \cite{Winter_EnergyConstrDiamond_2017}, after introducing energy-limited quantum channels, Winter noted that an affine upper bound $f_T(E)\le \lambda E+E_0$ formally corresponds to the operator inequality
\begin{equation}\label{eq:formal}
    T^*(G_B) \le \lambda G_A +E_0,
\end{equation}
where $T^*$ is the dual (Heisenberg-picture) channel.
However, since $T^*$ cannot be applied to unbounded operators, a rigorous version of this statement requires, first of all, a proper definition of $T^*(G_B)$.
Indeed, we will show that there is a canonical way to turn $T^*(G_B)$ into a positive self-adjoint operator, which then lets us prove \eqref{eq:formal} rigorously.
Before we proceed, we note that the assumption of energy-limitedness is necessary for $T^*(G_B)$ to make sense as an operator:
Consider the quantum channel $T\rho=(\tr \rho) \,\kettbra\psi$ with $\psi\notin\dom\sqrt G_{\!B}$ then $T^*(G_B)$ formally evaluates to multiplication by the "scalar" $\ip\psi{G_B\psi}=\oo$.

\begin{lemma}\label{thm:def_TG}
    Let $T:\T(\H_A)\to\T(\H_B)$ be an energy-limited cp map. Then: 
    \begin{enumerate}[(1)]
        \item\label{it:def_TG1}
            For every Stinespring dilation $(V,\H_R)$ of $T$, the operator $L:=(\sqrt G_{\!B}\ox1)V:\dom \sqrt G_{\!A}\allowbreak\to \H_B\ox\H_R$ is closable and well-defined.
        \item\label{it:def_TG2}
            The self-adjoint operator $T^*(G_B):=L^*\bar L$, with $L$ as in \ref{it:def_TG1}, does not depend on the chosen dilation.
        \item\label{it:def_TG3} 
            The positive quadratic form $a_0(\psi,\phi)=\energy_B[T\ketbra\phi\psi]$ defined on the form domain $Q(a_0)=\dom\sqrt G_{\!A}$ is closable.
            $T^*(G_B)$, as defined in item \ref{it:def_TG2}, is the unique self-adjoint positive operator inducing the closure of the positive quadratic form $a_0$.
    \end{enumerate} 
\end{lemma}

\begin{proof}
    \ref{it:def_TG1}: Let $(\psi_n)$ be an $L$-graph norm Cauchy sequence in $\dom\sqrt G_{\!A}$ such that $\psi_n\to 0$ in $\H_A$ and set $\psi=\lim_n L\psi_n$.
    Then $\psi=0$ because for every vector $\varphi$ from the dense subspace $\dom\sqrt G_{\!B} \odot\H_R\subseteq \H_B\ox\H_R$ it holds that $\ip\varphi{\psi} = \lim_n \ip\varphi{(\sqrt G_{\!B}\ox1)V\psi_n} = \lim_n \ip{(\sqrt G_{\!B}\ox1)\varphi}{V\psi_n}=0$.

    \ref{it:def_TG2}: Recall that if $B$ is a closed operator then $B^*B$ is self-adjoint \cite[Thm.~X.25]{ReedSimon2}.
    If $(V_i,\H_{R,i})$ is a Stinespring dilation of $T$, then $L^*_i\bar L_i$ is the unique positive self-adjoint operator inducing the closed quadratic form $a_i(\psi,\phi)=\ip{\bar L_i\psi}{\bar L_i\phi}$ with $Q(a_i)=\dom\bar L_i$, where $L_i=(\sqrt G_{\!B}\ox1)V_i$ with $\dom L_i =\dom \sqrt G_{\!A}$, $i=1,2$.
    $L_1$ and $L_2$ induce the same graph norm because $\norm{\psi}_{L_1}^2 = \norm\psi^2+\energy[T\kettbra\psi] = \norm\psi_{L_2}^2$, $\psi\in\dom\sqrt G_{\!A}$.
    Thus, they induce the same closed quadratic forms $a_1$ and $a_2$, which implies $L_1^*\bar L_1=L_2^*\bar L_2$ \cite[Sec.~X.4]{ReedSimon1}.
    
    \ref{it:def_TG3}: 
    Since we have $a_0(\psi,\psi)=\energy[T\ketbra\psi\psi] = \ip\psi{T^*(G_B)\psi}$, the polarization identity implies that $a_0$ is the form corresponding to $T^*(G_B)$ restricted to the form core $\dom \sqrt G_{\!A}$. This implies the claim.
    % To show that $a_0$ is closable we construct a closed extension. 
    % The proof of \ref{it:def_TG2} shows that for every Stinespring dilation $(V,\H_R)$, the positive self-adjoint operator $L^*\bar L$, with $L$ defined as in item \ref{it:def_TG1}, induces the same closed quadratic form $a$ extending $a_0$.
% 
    % It remains to be shown that $a$ is the closure of $a_0$. 
    % This means that the form domain $Q(a_0)=\dom L=\dom\sqrt G_{\!A}$ is dense in $Q(a)=\dom\bar L$ with respect to the norm $\psi\mapsto\sqrt{\norm\psi^2+a(\psi,\psi)}=\sqrt{\norm\psi^2+\norm{\bar L\psi}^2}=\norm\psi_{\bar L}$.
    % Density holds because $\dom L$ is graph norm-dense in $\dom\bar L$.
\end{proof}

Recall that the operator ordering $A\le B$ is defined for positive self-adjoint unbounded operators $A$ and $B$ by
\begin{equation}\label{eq:order}
    A \le B :\!\iff \dom\sqrt A\supseteq\dom\sqrt B \qandq \norm{\sqrt A\psi}\le\norm{\sqrt B\psi},\quad \psi\in\dom\sqrt B.
\end{equation}
Equipped with this definition and the above definition of $T^*(G_B)$ as a positive self-adjoint operator, we make Winter's statement precise:

\begin{proposition}\label{thm:EL2}
    % Let $\H_A,\H_B$ be Hilbert spaces with reference Hamiltonians $G_A,G_B$, respectively, 
    Let $T:\T(\H_A)\to\T(\H_B)$ be a cp map and let $\lambda,E_0\ge0$.
    The following are equivalent:
    \begin{enumerate}[(a)]
        \item $T$ is energy-limited and $f_T(E)\le\lambda E + E_0$ for all $E>0$.
        \item The positive quadratic form $(\psi,\phi)\mapsto \energy_B[T\ketbra\phi\psi]$ with form domain $\{\psi\in\dom\sqrt G_{\!A} : \energy_B[T\kettbra\psi]<\oo\}$ is densely defined and closable, and the operator $T^*(G_B)$ inducing its closure satisfies 
        \begin{equation}\label{eq:non_formal}
            T^*(G_B)\le \lambda G_A+E_0.
        \end{equation}
    \end{enumerate}
\end{proposition}

\begin{proof}
    Let $T$ be energy-limited with $\lambda,E_0\ge0$ such that $f_T(E)\le \lambda E+E_0$. 
    The form domain is simply $\dom\sqrt G_{\!A}$ and, hence, dense. Closability is proved in \cref{thm:def_TG}.  
    \Cref{eq:non_formal} follows from
     $$\norm{\sqrt{T^*(G_B)}\psi}^2 =\energy_B[T(\kettbra\psi)] \le \lambda\energy_A[\kettbra\psi]+E_0\tr\kettbra\psi=\lambda\norm{\sqrt G_{\!A}\psi}^2+E_0\norm\psi^2$$ for all $\psi\in\dom\sqrt G_{\!A}$.
     The converse is clear.    
\end{proof}

% \begin{remark}
%     Alternatively, one can define $T^*(G_B)$ and prove that $T^*(G_B)\le \lambda G+E_0 \iff f_T(E)\le \lambda E+E_0$ rigorously by using the extended positive cone of $\B(\H_A)$ (see \cref{sec:epc}).
%     This approach gives a short, elegant proof, b
%     we found the operator theoretic approach via the Stinespring dilation of $T$ to be more pleasing.
% \end{remark}

\begin{corollary}\label{thm:dual_prob}
    Let $T:\T(\H_A)\to\T(\H_B)$ be an energy-limited cp map.
    Then 
    \begin{equation}\label{eq:dual_prob}
        f_T(E) = \min \Big\{ \lambda E+E_0  : \lambda,E_0\ge0\ \text{s.t.}\  T^*(G_B) \le \lambda G_A +E_0 \Big\}.
    \end{equation}
\end{corollary}
\begin{proof}
    Since $f_T$ is a concave function, it is the pointwise minimum of affine functions dominating it, and, since $f_T$ is nondecreasing, we can restrict to affine functions with positive slope.
    Therefore:
    \begin{equation*}
        f_T(E)= \min\Big\{ \lambda E +E_0 : \lambda,E_0\in\RR\ \text{s.t.}\ f_T(E')\le\lambda E'+E_0, \forall E'>0\Big\}.
    \end{equation*}
    By \cref{thm:EL2}, $\lambda E +E_0$ dominates $f_T$ if and only if $T^*(G_B)\le \lambda G_A+E_0$, proving the claim.
\end{proof}

For Hilbert space operators, we can connect energy-limitedness with graph norm-boundedness:

\begin{corollary}\label{thm:EL_contraction}
    % Let $\H_A,\H_B$ be Hilbert spaces with reference Hamiltonians $G_A,G_B$, respectively, and l
    Let $V:\H_A\to\H_B$ be a bounded operator.
    The following are equivalent:
    \begin{enumerate}[(a)]
        \item\label{it:EL_contraction1} 
            $V$ is energy-limited
        \item\label{it:EL_contraction2}
            $V$ maps $\dom\sqrt G_{\!A}$ into $\dom\sqrt G_{\!B}$ and there exist constants $\lambda,E_0\ge0$ such that the operator inequality  $V^*G_BV\le \lambda G_A+E_0$ holds.
        \item\label{it:EL_contraction3}
             $V$ maps $\dom\sqrt G_{\!A}$ into $\dom\sqrt G_{\!B}$ and $\sqrt G_{\!B}V$ is $\sqrt G_{\!A}$-bounded, i.e., $V$ restricts to a bounded operator $\dom\sqrt G_{\!A}\to\dom\sqrt G_{\!B}$ (both of which are equipped with the graph norms).
     \end{enumerate}
     In this case $\norm{\sqrt G_{\!B}V\psi}^2 \le \lambda\norm{\sqrt G_{\!A}\psi}^2+E_0\norm\psi^2$ holds for all $\psi\in\dom\sqrt G_{\!A}$ and a given pair of constants $\lambda, E_0\ge0$ if and only if the operator inequality in \ref{it:EL_contraction2} holds.
\end{corollary}
\begin{proof}
    \ref{it:EL_contraction1} $\Leftrightarrow$ \ref{it:EL_contraction2} is clear from the proof of \cref{thm:EL2}.
    \ref{it:EL_contraction1} $\Leftrightarrow$ \ref{it:EL_contraction3} follows from $\energy_B[V\psi]=\|\sqrt G_{\!B}V\psi\|^2 =\ip\psi{V^*G_BV\psi}$ and $\lambda \|\sqrt G_{\!A}\psi\|^2+E_0\norm{\psi}^2 = \ip\psi{(\lambda G_A+E_0)\psi}$.
\end{proof}

Our next result is concerned with altering the reference Hamiltonians:

\begin{theorem}
    % Let $\H_A,\H_B$ be Hilbert spaces with reference Hamiltonians $G_A,G_B$, respectively.
    Let $\varphi:\RR^+\to\RR^+$ be a continuous operator-monotone function such that $\varphi(0)=0$ and $\lim_{t\to\oo}\frac{\varphi(\lambda t)}{\varphi(t)}<\oo$ for all $\lambda>0$.
    Then every trace-nonincreasing cp map $T:\T(\H_A)\to\T(\H_B)$ which is energy-limited with respect to $G_A$ and $G_B$ is also energy-limited with respect to $\varphi(G_A)$ and $\varphi(G_B)$.
\end{theorem}

In particular, the result holds for the square root $\varphi(t)=t^{1/2}$.
The main tools used in the proof are operator inequalities involving cp maps and operator monotonicity, which we extend to unbounded operators.
These extensions are best formulated using the "extended positive cone" of $\B(\H)$, a concept from the theory of von Neumann algebras (see \cref{sec:epc}).

\begin{proof}
    Let $\lambda,E_0>0$ be such that $T^*(G_B)\le \lambda G_A+E_0$.
    Then \cref{thm:op_monotone,thm:op_monotone_quantum_op} show
    \begin{equation*}\label{eq:help1}
        T^*(\varphi(G_B)) \le \varphi(T^*(G_B)) \le \varphi(\lambda G_A+E_0).
    \end{equation*}
    We want to show that the right-hand side is bounded by $\lambda'\varphi(G_A)+E_0'$ for constants $\lambda', E_0'\ge0$.
    It suffices to show $\varphi(\lambda t+E_0)\le \lambda'\varphi(t)+E_0'$ for all $t\ge0$.
    For uniformly bounded $\varphi$ we can just set $E_0'=\sup_t \varphi(t)$ and $\lambda'=1$.
    This leaves us with the case where $\varphi(t)\to\oo$ as $t\to\oo$. 
    Since $\varphi$ is continuous we can pick $E_0'>0$ such that $\varphi(\lambda t+E_0)\le E_0'$ for all $t\in[0,1]$.
    A constant $\lambda'>0$ such that the inequality holds for all $t\ge0$ exists if and only if
    \[
        \sup_{t\ge1} \frac{\varphi(\lambda t + E_0) -E_0'}{\varphi(t)} <\oo. 
    \] 
    Since $\varphi(t)$ is bounded away from zero for  $t\ge1$, this holds if and only if $\lim_{t\to\oo}\frac{\varphi(\lambda t+E_0)}{\varphi(t)}<\oo$ which is guaranteed by the growth assumption and the monotonicity of $\varphi$.
    Combining this with \eqref{eq:help1}, we get $\tr T^*(\varphi(G_B))\rho \le \tr\varphi(\lambda G_A+E_0)\rho \le \lambda'\tr\varphi(G_A)\rho + E_0$ for all $\rho\in\TC^+$ (where we adopted the trace notation instead of using the energy functionals of $\varphi(G_i)$ to make things clearer).
\end{proof}

\begin{remark}
    The proof actually shows more: If $\varphi_A:\RR^+\to\RR^+$ is an operator monotone function with $\varphi_A(0)=0$ and $\varphi_B:\Sp(G_B)\to\RR^+$ a Borel function with $\varphi_B(0)=0$ such that for all $\lambda,E_0>0$ there exist $\lambda',E_0'>0$ such that $\varphi_A(\lambda t+E_0)\le \lambda'\varphi_B(t)+E_0'$ for all $t\ge0$.
    Then energy-limitedness with respect to $G_A$ and $G_B$ implies energy-limitedness with respect to the reference Hamiltonians $\varphi_A(G_A)$ and $\varphi_B(G_B)$.
\end{remark}

\subsection{Energy-constrained norms}\label{sec:norms}

We collect properties and definitions of energy-constrained norms and show how they relate to energy-limited quantum channels.
The energy-constrained diamond norm was introduced by Shirokov and Winter in \cite{shirokov2018,Winter_EnergyConstrDiamond_2017}, and the energy-constrained operator norm was introduced by Shirokov \cite{shirokov2020Enorm,shirokov2020strong}.\footnote{A norm similar to the energy-constrained diamond norm was also introduced in \cite{pirandola_fundamental_2017} for bosonic systems. The energy-constrained operator norm is called \emph{operator $E$-norm} by Shirokov \cite{shirokov2020Enorm,shirokov2022form}. We choose "energy-constrained operator norm" to highlight the analogy with the energy-constrained diamond norm.}

\begin{definition}[Shirokov, Winter]\label{def:norms}
    Let $\H_A$ and $\H_B$ be Hilbert spaces with reference Hamiltonians $G_A$, $G_B$ respectively and let $E>0$.
    For operators $V:\H_A\supseteq \dom V\to \H_B$ with $\dom V\supseteq \dom \sqrt G_{\!A}$, the {\bf energy-constrained operator (ECO) norm} is defined as
    \begin{equation}\label{eq:norm1}
        \norm V_{\op,E} = \sup_{\substack{\norm\psi=1\\ \energy[\psi]\le E}} \norm{V\psi} 
    \end{equation}
    For $^*$-preserving linear maps $T:\T(\H_A)\supseteq \dom T\to \T(\H_B)$ with $\dom T \supset \states_{<\oo}(\H_A)$, the {\bf energy-constrained diamond (ECD) norm} is defined as
    \begin{equation}\label{eq:norm2}
        \norm{T}_{\diamond,E} = \sup_{\rho\in\states_E(\H_{AR})} \ \norm{(T\ox\id)\rho}_1,
    \end{equation}
    where $R$ is an ancillary system with infinite-dimensional Hilbert space $\H_R$ and $G_R=0$.
\end{definition}

For a fixed $^*$-preserving map $T$, the ECD norm $\norm T_{\diamond,E}$ is a concave nondecreasing function of the energy and, hence, satisfies
\begin{equation}\label{eq:concavity_ineq_ECD}
    \norm T_{\diamond,E}\le \norm T_{\diamond,E'}\le \tfrac{E'}E\, \norm T_{\diamond,E},\qquad E'\ge E>0.
\end{equation}
In the case of the ECO norm, one has
\begin{equation}\label{eq:concavity_ineq_opE}
    \norm V_{\op,E}\le \norm V_{\op,E'}\le \sqrt{\tfrac{E'}E}\,\norm V_{\op,E},\qquad E'\ge E>0.
\end{equation}
Sometimes, it is useful to express the ECO norm as
\begin{equation}\label{eq:Enorm'}
    \norm V_{\op,E} = \sup_{\substack{\norm\phi=\norm\psi=1\\\energy[\psi]\le E}} \abs{\ip\phi{V\psi}}
\end{equation}
where $\phi$ may be drawn from some dense subspace.\footnote{Similarly, it suffices to optimize over vectors $\psi$ in some core $\D$ of $\sqrt G$; see \cref{lem:ECD} below.}
% For unitaries, the ECO and ECD norm distances are related through (see \cite[Thm.~1]{becker2021} and \cite[Lem.~13]{lie_group_error}):
% \begin{equation}\label{eq:ECD_unitary}
%   \norm{U(\placeholder)U^* - V(\placeholder)V^*}_{\diamond,E} \le 2 \norm{U-V}_{\op,E}.
% \end{equation}
% 
% \begin{lemma}\label{lem:ECD_unitary}
%     Let $U,V$ be unitary on a Hilbert space $\H$ with reference Hamiltonian $G$.
%     Then
%     \begin{equation}\label{eq:ECD_unitary}
%         \sqrt2\min_{\phi\in\RR}\,\norm{e^{i\phi}U-V}_{\op,E} \le \norm{U(\placeholder)U^* - V(\placeholder)V^*}_{\diamond,E} \le 2 \norm{U-V}_{\op,E}.
%     \end{equation}
% \end{lemma}

By definition of the ECD norm, it holds that $\norm{T\rho}_1 \le \norm T_{\diamond,\energy[\rho]}$ for all $\rho\in\states(\H)$.
Therefore, convergence in ECD norm always implies pointwise convergence on the state space.

\begin{lemma}[Shirokov {\cite{shirokov2018,shirokov2020Enorm}}]
    Assume that the reference Hamiltonian is of the form \eqref{eq:discrete}.
    Then the ECD norm metrizes the strong operator topology on bounded sets of $^*$-preserving linear maps $\T(\H_A)\to\T(\H_B)$, and the ECO norm metrizes the strong operator topology on bounded subsets of $\B(\H_A,\H_B)$.
\end{lemma}

The ECO norm can be characterized via a semidefinite minimization problem:

\begin{lemma}\label{thm:dual_prob_Enorm}
    Let $V:\dom \sqrt G_{\!A}\to\H_B$ be $\sqrt G_{\!A}$-bounded.
    Then 
    \begin{equation}\label{eq:dual_prob_Enorm}
        \norm V_{\op,E}^2 = \min \Big\{ \lambda E+ E_0 : \lambda,E_0\ge0\ \mathrm{s.t.}\ VV^* \le \lambda G_A +E_0 \Big\}.
    \end{equation}
\end{lemma}
\begin{proof}
    By \cite{shirokov2020Enorm} $E\mapsto \norm V_{\op,E}^2$ is concave and nondecreasing.
    Hence, it is the pointwise minimum of affine functions $\lambda E+E_0$ with $\lambda,E_0\ge0$ such that $\norm V_{\op,E}^2\le \lambda E+E_0$  $\forall E>0$.
    The latter is equivalent to $\norm{V\psi}^2 \le \lambda \norm{\sqrt G_{\!A}\psi}^2 + E_0\norm \psi^2$ for all $\psi\in \dom\sqrt G_{\!A}$, which is equivalent to $V^*V\le \lambda G_A+E_0$ (see \eqref{eq:order}).
\end{proof}

The dual of the semidefinite minimization problem \eqref{eq:dual_prob_Enorm} is precisely to maximize the energy of $V\rho V^*$ under the energy constraint $E$, i.e., \eqref{eq:Enorm_ECD_norm}.
Thus, the primary and dual problems have the same solution.

We now collect some useful properties of these energy-constrained norms, most of which are taken from Shirokov's works \cite{shirokov2020Enorm,shirokov2019,shirokov2020strong}:

\begin{lemma}\label{lem:ECD}
    Let $\H_A$ and $\H_B$ be Hilbert spaces and let $G_A$ be a reference Hamiltonian on $\H_A$. Let $\H_R$ be a separable Hilbert space with $G_R=0$.
    Let $T:\T(\H_A)\supseteq\dom T\to\T(\H_B)$ be a $^*$-preserving map with $\dom T\supseteq \dom\energy_A$, and let $V:\H_A\supseteq \dom V\to\H_B$ be an operator with $\dom V\supseteq \dom\sqrt G_{\!A}$.
    \begin{enumerate}[(1)]
        \item\label{it:ECD_subnormalized} 
            To compute the ECD norm, one may include subnormalized energy-constrained states, i.e.,
            \begin{equation}\label{eq:ECD_subnormalized}
                \norm{T}_{\diamond,E} = \sup\big\{\norm{(T\ox\id)\rho}_1 : \rho\in\T(\H_{AR})^+\!,\ \tr\rho\le1,\ \energy_{AR}[\rho]\le E\big\}.
            \end{equation}
        \item\label{it:Enorm_subnormalized} 
            To compute the ECO norm, one may include subnormalized pure states, i.e.,
            \begin{equation}\label{eq:Enorm_subnormalized}
                \norm{V}_{\op,E} = \sup\big\{\norm{V\psi} : \psi\in\dom\sqrt G_{\!A},\ \norm\psi\le1,\ \energy_A[\psi]\le E\big\}.
            \end{equation}
        \item\label{it:ECD_cp}
            If $T$ is cp, the ECD norm is given by $\norm{T}_{\diamond,E} = \sup_{\rho\in\states_E}\norm{T\rho}_1 = \sup_{\rho\in\states_E}\tr[T\rho]$.
        \item\label{it:Enorm_ECD_norm}
            Assume $V$ has finite ECO norm. If $\rho\in\states_{<\oo}(\H_A)$ is a finite-energy state and $\rho=\sum_\alpha \lambda_\alpha\kettbra{\psi_\alpha}$ with $\lambda_\alpha\ge0$, then $V\rho V^*:=\sum_{\alpha}\ketbra{V\psi_\alpha}{V\psi_\alpha}$ converges in trace-norm.
            Furthermore,
            \begin{equation}\label{eq:Enorm_ECD_norm}
                \norm V_{\op,E} = \sup_{\rho\in\states_E} \sqrt{\tr{V\rho V^*}} = \sqrt{\norm{V(\placeholder)V^*}_{\diamond,E}}<\oo,
            \end{equation}
            where we extend $V(\placeholder)V^*$ linearly to a map $\dom\energy_A\to\T(\H_B)$.
        \item\label{it:norm_dense}
            $V$ has finite ECO norm if and only if $V$ is $\sqrt G_{\!A}$-bounded.
            If $\D$ is a core for $\sqrt G_{\!A}$ and $V$ is $\sqrt G_{\!A}$-bounded, the supremum in \eqref{eq:norm1} can be restricted to vectors in $\D$, i.e.,
            \begin{equation}\label{eq:norm_dense}
                \norm{V}_{\op,E} = \sup\{ \norm{V\psi} : \psi\in\D,\ \norm\psi=1,\ \energy[\psi]\le E\}.
            \end{equation}
            If $V':\D\to\H_B$ is an operator such that the right-hand side of \eqref{eq:norm_dense} is finite for some $E>0$, then $V'$ is $\sqrt G_{\!A}$-bounded on $\D$ and its $\sqrt G_{\!A}$-graph norm continuous extension $V:\dom\sqrt G_{\!A}\to\H_B$ has ECO norm given by the right-hand side of \eqref{eq:norm_dense}.
        \item\label{it:ECDnorm_dense}
            $T$ has finite ECD norm if and only if $T\ox\id:\dom\energy_{AR}\to\T(\H_{BR})$ is bounded, where $\dom\energy_{AR}$ is equipped with the norm $\normiii{}_1$ and $G_{AR}=G_A\ox 1$.
            In this case, the supremum in \eqref{eq:norm2} can be restricted to any $\normiii{}_1$-dense subspace $\D\subset\dom\energy_A$, i.e.,
            \begin{equation}\label{eq:ECDnorm_dense}
                \norm T_{\diamond,E} = \sup \{\norm{T\ox\id \rho}_1 : \rho \in\states_E(\H_{AR})\cap(\D\odot\T(\H_R))\}.
            \end{equation}
            If $T':\D\to\T(\H_B)$ is a $^*$-preserving map such that the right-hand side of \eqref{eq:ECDnorm_dense} is finite for some $E>0$, then $T'$ is bounded for the $\normiii{}_1$ norm on $\D$ and the $\normiii{}_1$-continuous extension to a map $\dom\energy_A\to\T(\H_B)$ has finite ECD norm.
    \end{enumerate} 
\end{lemma}
\begin{proof}
    \ref{it:ECD_subnormalized} is proved in \cite[Lem.~1]{shirokov2019}.
    \ref{it:Enorm_subnormalized} is proved in \cite[Prop.~3]{shirokov2020Enorm} and
    \ref{it:ECD_cp} is straightforward. % Let $\rho \in \states_E(\H_{AR})$. Then $\tr_R\rho\in\states_E(\H_A) $ and thus $$\norm{(T\ox\id)\rho}_1=\tr[(T\ox\id)\rho] = (\tr_B\ox\tr_R)(T\ox\id)\rho = \tr_B[(T\circ\tr_R) \rho] =\norm{T\tr_R\rho}_1.$$ Taking the supremum over all $\rho\in\states_E(\H_{AR})$, we see that the left-hand side is bounded by the right-hand side of \eqref{eq:ECD_norm_cp}. Since the converse inequality holds trivially, the claim follows.
    \ref{it:Enorm_ECD_norm}
    was shown by Shirokov \cite{shirokov2020Enorm}. The last equality in \eqref{eq:Enorm_ECD_norm} follows from \ref{it:ECD_cp}.

    \ref{it:norm_dense}:
    The equivalence of $\sqrt G_{\!A}$-boundedness and finite ECO norm was proved by Shirokov in \cite{shirokov2020Enorm}.
    Since we assume that $V$ is $\sqrt G_{\!A}$-bounded, \eqref{eq:norm_dense} is clear.
    Now let $V':\D\to\H_B$ be such that the right-hand side of \eqref{eq:norm_dense} is finite.
    Following Shirokov \cite{shirokov2020Enorm,shirokov2020strong}, the right-hand side equals the square root of the supremum of $\sum p_\alpha \norm{V'\psi_\alpha}^2$ where we optimize over probability distributions $(p_\alpha)$ on $\NN$ and sequences of unit vectors $\psi_\alpha\in\D$ such that $\sum_\alpha p_\alpha\energy[\psi_\alpha]\le E$.
    Therefore the right-hand side of \eqref{eq:norm_dense} is the square root of a concave nondecreasing function of $E$ and hence bounded by $\sqrt{aE+b}$ for some $a,b\ge0$.
    This immediately gives $\norm{V'\psi}^2\le a\norm{\sqrt G\psi}^2+b\norm\psi^2$ for $\psi\in\D$ and, thus, shows that $V'$ is $\sqrt G$-bounded. The rest follows from considering $V'$ as the restriction of its graph norm continuous extension $V'\subset V:\dom\sqrt G_{\!A}\to\H_B$.

    \ref{it:ECDnorm_dense}:
    Let $W_{AR}:\dom \energy_{AR}\to\T(\H_{AR})$ be the isometric isomorphism from \cref{lem:base}.
    Assume $T\ox \id$ is bounded with $M>0$ such that $\norm{(T\ox\id)\rho}_1\le M\normiii\rho_1$, $\rho\in\dom\energy_{AR}$.
    If $\rho\in\states_E(\H_{AR})$ then $\norm{(T\ox\id)\rho}_1 \le M\normiii\rho_1 = M\norm{W_{AR}\rho}_1 \le M(E+1)$ implies $\norm{T}_{\diamond,E}\le M(E+1)<\oo$.
    Conversely, assume that $\norm T_{\diamond,E}<\oo$.
    If $\rho\in\states(\H_{AR})$ then $W_{AR}^{-1}\rho$ is a subnormalized state with energy bounded by $1$.
    By \eqref{eq:ECD_subnormalized}, we have $\norm{(T\ox\id)W_{AR}^{-1}\rho}_1 \le \norm{T}_{\diamond,1}$.
    Therefore $(T\ox\id)W_{AR}^{-1}:\T(\H_{AR})\to\T(\H_{BR})$ is bounded which is equivalent to boundedness of $(T\ox\id):\dom\energy_{AR}\to\T(\H_{BR})$.
    In this case, \eqref{eq:ECDnorm_dense} follows for every $\normiii{}_1$-dense subspace $\D\subset\dom\energy_A$.
    Now let $T':\D\to\T(\H_B)$ be as described. The right-hand side of \eqref{eq:ECDnorm_dense} is a concave function of $E$. Hence, it is finite for all $E>0$.
    Therefore, boundedness of $T'\ox\id:\D\odot\T(\H_R)\to\T(\H_{AR})$ with respect to the $\normiii{}_1$ norm follows as before. 
    The rest follows from considering $T'$ as the restriction of its $\normiii{}_1$-continuous extension $T:\dom\energy_{A}\to\T(\H_{B})$.
\end{proof}

As a consequence of item \ref{it:ECDnorm_dense}, we can partially answer a conjecture of Shirokov \cite{shirokov2019completion}: A $^*$-preserving map $T:\dom\energy_A \to \T(\H_B)$ has finite ECD norm $\norm T_{\diamond,E}<\oo$ if and only if $T=T_+-T_-$ is the difference of two completely positive maps with finite ECD norm $\norm{T_\pm}_{\diamond,E}<\oo$.\footnote{Indeed, decomposing the completely bounded $^*$-preserving map $S=T\circ W_A^{-1}:\T(\H_A)\to\T(\H_B)$ as $S=S_+-S_-$ yields a decomposition of $T$ via $T_\pm = S_\pm\circ W_A$, where $W_A$ is the isomorphism from \cref{lem:base}}

\begin{proposition}[Submultiplicativity]\label{thm:submultiplicativity}
    Let $\H_A,\H_B,\H_C$ be Hilbert spaces and let $G_A$, $G_B$ be reference Hamiltonians on $\H_A$ and $\H_B$, respectively.
    \begin{enumerate}[(1)]
        \item\label{it:submultiplicativity1} 
            Let $V:\H_A\to\H_B$ be an energy-limited contraction and let $W:\H_B\supseteq \dom W \to\H_C$ be an operator with $\dom W\supseteq \dom\sqrt G_{\!B}$. 
            Then
            \begin{equation}\label{eq:submultiplicativity1}
                \norm{WV}_{\op,E} \le \norm{W}_{\op,f_V(E)} \le \sqrt{\frac{f_{V}(E)}E} \norm W_{\op,E},\qquad E>0.
            \end{equation}
        \item\label{it:submultiplicativity2} 
            Let $T:\T(\H_A)\to\T(\H_B)$ be an energy-limited trace-nonincreasing cp map and let $S:\T(\H_B)\supseteq\dom S\to\T(\H_C)$ be a $^*$-preserving linear map such that $\states_{<\oo}(\H_B)\subset\dom S$. Then
            \begin{equation}\label{eq:submultiplicativity2}
                \norm{ST}_{\diamond,E} \le \norm{S}_{\diamond,f_T(E)}\le \frac{f_{S}(E)}E \norm S_{\diamond,E},\qquad E>0.
            \end{equation}
    \end{enumerate} 
\end{proposition}

\begin{proof}
    By \eqref{eq:Enorm_ECD_norm}, the second item implies the first one. 
    \ref{it:submultiplicativity2}:
    Let $\rho\in\states_E(\H_{AR})$. Then $\sigma=(T\ox\id)\rho\in\T(\H_{BR})^+$ with $\tr\sigma\le 1$ and $\energy_{BR}[\sigma]\le f_{T\ox\id}(E)=f_T(E)$ (see item \ref{it:energy_lim3} of \cref{lem:energy_lim}).
    By item \ref{it:ECD_subnormalized} of \cref{lem:ECD}, it holds that 
    $
    \norm{(ST\ox\id)\rho}_1 = \norm{(S\ox\id)\sigma}_1 \le \norm{S}_{\diamond,f_T(E)}.
    $
    The result now follows from \eqref{eq:concavity_ineq_fT}.
\end{proof}

\begin{remark}[Nonzero ground state energy]\label{rem:nonzero_gse0}
    Most of the statements presented in \cref{sec:channels,sec:norms} do not require the assumption that the reference Hamiltonians have vanishing ground state energy.
    In particular, \cref{thm:dual_prob,thm:dual_prob_Enorm} do not depend on the ground state energy being zero.
    However, items \ref{it:energy_lim2} and \ref{it:energy_lim5} of \cref{lem:energy_lim} and items \ref{it:ECD_subnormalized} and \ref{it:Enorm_subnormalized} of \cref{lem:ECD} need the ground state energy to be nonzero.
    These statements have in common that they (or their proofs) involve subnormalized states.
    For instance, item \ref{it:energy_lim5} of \cref{lem:energy_lim} will be true even for nonzero ground state energy if the cp map $S$ is trace-preserving.
\end{remark}

\section{Energy-limited dynamics}\label{sec:dynamics}

In this chapter, we develop the theory of energy-limited dynamics.
We mostly focus on the case of Markovian dynamics.
Nonetheless, we begin by properly defining energy-limitedness in the general case and establishing its basic properties in full generality.
We fix a Hilbert space $\H$ with a reference Hamiltonian $G$.

A \emph{quantum evolution system} $\{T(t,s)\}_{t\ge s}$ is a collection of completely positive trace-nonincreasing maps $T(t,s)$ on $\TC$, defined for times $t\ge s$ in some interval, such that
\begin{enumerate}[(i)]
    \item $T(t,s) T(s,u) = T(t,u)$ and $T(t,t)=\id$ for all $t\ge s\ge u$,
    \item $T(t,s)\rho \to \rho$ as $t\to s^+$ for all $\rho \in\TC$ and $s$.
\end{enumerate}
Physically, the maps $T(t,s)$ model the change from time $s$ to time $t$.
In general, we do not assume the time evolution to be trace-preserving, accounting for cases where particles are lost (e.g., in arrival time measurements \cite{werner_arrival_1987}). Additionally, we do not assume the evolution to be time-homogeneous.
However, we say that an evolution system is \emph{conservative}, if $T(t,s)$ is trace-preserving for all $t\ge s$, and \emph{Markovian}, if $T(t,s)$ only depends on the time increment $t-s$.
If $\{T(t,s)\}_{t\ge s}$ is a Markovian evolution system, we set $T(t):= T(t,0)$. The properties of evolution systems imply that $\{T(t)\}_{t\ge0}$ is a \emph{quantum dynamical semigroup}, i.e., a strongly continuous one-parameter semigroup of completely positive trace-nonincreasing maps on $\TC$ \cite{davies1974,davies1975,davies1976}, from which the evolution system can be recovered via $T(t,s)=T(t-s)$.

As our definition of energy-limited dynamics, we take that for small time-increments, the output energy should be linearly bounded by the input energy:

\begin{definition}\label{def:energy_lim_dyn}
    % Let $\H$ be a Hilbert space and let $G$ be a reference Hamiltonian.
    A quantum evolution system $\{T(t,s)\}_{t\ge s}$ on $\H$ is {\bf energy-limited} if there exist constants $\omega,E_0\in\RR$ such that for each finite-energy state $\rho$, it holds
    \begin{equation}\label{eq:first_order}
        \energy[T(t+\Delta t,t)\rho] \le \energy[\rho] + (\omega \Delta t+\order(\Delta t)) (\energy[\rho]+E_0),\qquad t,\Delta t\ge0.
    \end{equation}
    Such constants $\omega,E_0$ are called {\bf stability constants}.
    A quantum dynamical semigroup is energy-limited if the corresponding Markovian evolution system is.
\end{definition}

Since the right-hand side of \eqref{eq:first_order} must be larger than the ground state energy, which is zero by convention, any pair of stability constants $\omega,E_0$ must satisfy $\omega \cdot E_0 \ge 0$.

\begin{lemma}\label{thm:first_order}
    A quantum evolution system $\{T(t,s)\}_{t\ge s}$ is energy-limited with stability constants $\omega,E_0$ if and only if 
    \begin{equation}\label{eq:affine_semigroup_bound}
        f_{T(t,s)} \le E + (e^{\omega (t-s)}-1)(E+E_0), \qquad t\ge s\ge0.
    \end{equation}
    In this case, $t\mapsto\energy[T(t,s)\rho]$ is right-continuous and lower semicon\-tinuous in $t\ge s$ for all $s$ and all finite-energy states $\rho\in\states_{<\oo}$.
\end{lemma}

The functions $f_t(E) = E+(e^{\omega t}-1)(E+E_0)-$ form groups of affine functions, i.e., $f_t\circ f_s = f_{t+s}$ holds for all $t,s\in\RR$. Sometimes the form $f_t(E)=Ee^{\omega t}+(e^{\omega t}-1)E_0$ is more convenient.
% To be precise, every continuous one-parameter semigroup $f_t$, $t\ge0$, of affine functions on $\RR$ is either of the form $f_t(E)=E+tE_0$ or of the form $f_t(E) = e^{\omega t}(E+E_0)-E_0$, where $E_0,\omega\in\RR$. In fact, all of these are actually groups.}

\begin{proof}
    The "if" part is clear. For the converse, we start by showing, for $n\in\NN$, the estimate
    \begin{align}\label{eq:cuss}
        f_{T(nt,0)}(E) \le \big(1+\omega t + \order(t)\big)^n E + E_0\big(\omega t+\order(t)\big) \sum_{k=0}^{n-1} \big(1+\omega t + \order(t)\big)^k
    \end{align}
    with induction. The case $n=1$ is follows from \eqref{eq:first_order} and the induction step goes as follows:
    \begin{align*}
        f_{T((n+1)t,0)}(E)
        &\le f_{T((n+1)t,nt)}\circ f_{T(nt,0)}(E)\\
        &\le \big(1+\omega t+\order(t)\big)f_{T(nt,0)}(E) + E_0(t + \order(t)) \\
        % &\le \big(1+\omega t + \order(t)\big)^{n+1} E + E_0(\omega t+\order(t)) \sum_{k=0}^{n-1} \big(1+\omega t + \order(t)\big)^{k+1} + E_0(\omega t + \order(t))\\
        &\le \big(1+\omega t + \order(t)\big)^{n+1} E + E_0(\omega t+\order(t)) \sum_{k=0}^{n} \big(1+\omega t + \order(t)\big)^k.
    \end{align*}
    Evaluating the geometric sum in \eqref{eq:cuss} and replacing $t$ by $\frac tn$, gives
    \begin{align}
        f_{T(t,0)}(E) 
        % &\le (1+\tfrac{\omega t}n + \order(\tfrac tn))^n E + E_0(\tfrac{\omega t}n + \order(\tfrac tn)) \frac{1- (1+\tfrac{\omega t}n+\order(\tfrac tn))^n}{\tfrac{\omega t}n + \order(\tfrac tn)} \nonumber\\
        &\le (1+\tfrac{\omega t}n + \order(\tfrac tn))^n E - E_0(1- (1+\tfrac{\omega t}n+\order(\tfrac tn))^n)
    \end{align}
    By Eulers Formula, the right-hand side converges to $e^{\omega t}E+(e^{\omega t}-1)E_0$ as $n\to\oo$.
    Lower semicontinuity follows from lower semicontinuity of $\energy$. Right-continuity follows from lower semicontinuity: $\energy[\rho]\le \liminf_{t\to s^+}\energy[T(t,s)\rho]
        \le \limsup_{t\to s^+} \energy[T(t,s)\rho] \le \lim_{t\to s^+} (e^{\omega (t-s)}(\energy[\rho]+E_0)-E_0)
        =\energy[\rho]$.
\end{proof}

We say that $\omega, E_0$ are \emph{joint stability constants} for a collection of quantum evolution systems $\{T_i(t,s)\}_{t\ge s}$, $i\in I$, if they are stability constants for each of the dynamics.
A collection is \emph{jointly energy-limited} if it admits joint stability constants.

\begin{lemma}\label{lem:joint_constants}
    Every finite collection $\{T_i(t,s)\}_{t\ge s}$, $i\in I$, of energy-limited quantum evolution systems is jointly energy-limited.
\end{lemma}
\begin{proof}
    % Since $e^{\omega t} E + (e^{\omega t}-1)E_0$ is monotone in $\omega,E_0$ for $t\ge0$, we can just pick the largest stability constants and use them for all of these evolution systems.
    Let $\omega_i,E_{0,i}$ be stability constants for the respective dynamics and set $\omega=\max_i \omega_i$, $E_0=\max E_{0,i}$. Then $f_{T_i(t,s)}(E)\le E+(e^{\omega_i(t-s)}-1)(E+E_{0,i})\le E+(e^{\omega t}-1)(E+E_0)$ for all $i$.
\end{proof}

In the rest of this chapter, we restrict to Markovian dynamics.
In \cref{sec:unitary_dynamics}, we start with unitary dynamics.
In \cref{sec:general_open}, we deal with open quantum systems in full generality.
Afterward, we consider standard generators in \cref{sec:standard}.
Examples of energy-limited dynamics can be found in \cref{sec:examples}.

\subsection{Unitary dynamics}\label{sec:unitary_dynamics}

Unitary one-parameter groups $\{U(t)\}_{t\in\RR}$ describe invertible Markovian quantum dynamics.
We distinguish between forward and backward energy-limitedness:

\begin{definition}
    A unitary one-parameter group $\{U(t)\}_{t\in\RR}$ is called {\bf forward} (resp.\ {\bf backward energy-limited}) if the forward dynamical semigroup $\{T_+(t)\}_{t\ge0}$ (resp.\ the backward dynamical semigroup $\{T_-(t)\}_{t\ge0}$) is energy-limited, where $T_\pm(t):=U(\pm t)(\placeholder) U(\pm t)^*$.
    We say that $\{U(t)\}_{t\in\RR}$ is {\bf energy-limited} if it is both forward and backward energy-limited. 
\end{definition}

According to \cref{lem:joint_constants}, a unitary one-parameter group is energy-limited if and only if there are stability constants $\omega,E_0\ge0$ such that
\begin{equation}\label{eq:fU_bound}
    f_{U(t)}(E) \le e^{\omega\abs t}(E+E_0)-E_0, \qquad t\in\RR.
\end{equation}
Backward energy-limitedness is equivalent to a lower bound on the energy loss of the forward dynamics.
This also lets us prove:

\begin{lemma}
    Let $\{U(t)\}_{t\in\RR}$ be an energy-limited unitary group with stability constants $\omega,E_0\ge0$.
    Let $\psi\in\H$ be a unit vector. Then the energy change of $\psi$ is bounded as
    \begin{equation}
        - w(-\abs t)\le  \energy[U(t)\psi]-\energy[\psi]  \le w(\abs t), \qquad t\in\RR,
    \end{equation}
    where $w(t) = (e^{\omega t}-1)(\energy[\psi]+E_0) = \omega t(\energy[\psi]+E_0)+\Order(t^2)$.
    In particular, $t\mapsto \energy[U(t)\psi]$ is continuous in $t$ for all $\psi \in \dom \sqrt G$.
\end{lemma}
\begin{proof}
    Assume without loss of generality that $t>0$.
    The upper bound is immediate from forward energy-limitedness.
    The lower bound follows from backward energy-limitedness: Since $\psi = U(-t) (U(t)\psi)$, we have
    $\energy[\psi] \le e^{-\omega t}(\energy[U(t)\psi]+E_0)-E_0$, which is equivalent to the lower bound.
\end{proof}

\begin{example}[Forward but not backward energy-limited]
    Let $\H = L^2(\RR)$ and let $Q,P$ be the canonical position and momentum operators. Consider the multiplication operator $G = \exp Q^3$.
    Then $U(t)=e^{-itP}$ is forward energy-limited (because $U(t)^*G U(t)= e^{(Q-t)^3} \le G$, $t>0$) but not backward energy-limited because one can never find $\omega, E_0$ such that $e^{(x+t)^3}$ is bounded by $e^{\omega t} (e^{x^3}+E_0)$ for all $x>0$ and all $t>0$.
\end{example}

We start by stating our main result on energy-limited unitary dynamics.
To do this, we define
\begin{equation}\label{eq:restricted_dom}
    \dom(H\upharpoonright\dom\sqrt G) := \big\{\psi\in\dom\sqrt G\cap \dom H : H\psi\in\dom\sqrt G\big\},
\end{equation}
where $H$ is some densely defined operator on $\H$.

\begin{theorem}\label{thm:unitary}
    Let $H$ be a self-adjoint operator on $\H$ and set $U(t)=e^{-itH}$.
    Then $\{U(t)\}_{t\in\RR}$ is energy-limited with stability constants $\omega,E_0\ge0$ if and only if both of the following properties hold:
    \begin{enumerate}[(i)]
        \item For all $t\in\RR$, $U(t)$ leaves $\dom\sqrt G$ invariant, the restrictions $U_0(t):=U(t)|_{\dom\sqrt G}$ are $\sqrt G$-graph norm bounded and form a $\sqrt G$-graph norm-strongly continuous one-parameter group.
        \item The operator inequality $\pm i [H,G]\le \omega(G+E_0)$ holds in the sense that
        \begin{equation}\label{eq:unitary2}
            \pm i (\ip{\sqrt G\psi}{\sqrt GH\psi}-\ip{\sqrt GH\psi}{\sqrt G\psi}) \le \omega(\norm{\sqrt G\psi}^2+E_0\norm\psi^2)
        \end{equation}
        for all $\psi\in\dom(H\upharpoonright\dom\sqrt G)$.
    \end{enumerate}
\end{theorem}

The downside to \cref{thm:unitary} is that verifying the strong continuity of $U_0(t)$ for a given self-adjoint operator $H$ is hard in practice.
To address this, we adapt an idea due to Fröhlich \cite{frohlich} to obtain sufficient conditions that make energy-limitedness explicitly checkable in concrete examples:

\begin{theorem}\label{thm:unitary2}
    % Let $\H$ be a Hilbert space with reference Hamiltonian $G$.
    Let $\omega,E_0\ge0$ and let $H$ be self-adjoint operator.
    Let $\D\subset \dom H$ be a core got $G$ on which $H$ is $G$-bounded and satisfies $\pm i[H,G]\le \omega(G+E_0)$ for $\omega,E_0\ge0$, in the sense that
    \begin{equation}\label{eq:commutator_bound}
        \abs{\ip{H\psi}{G\psi}-\ip{G\psi}{H\psi}} \le \ip\psi{\omega(G+E_0)\psi}, \qquad \psi\in\D.
    \end{equation}
    Then, the unitary group generated by $H$ is energy-limited with stability constants $\omega,E_0$.
\end{theorem}

The condition of self-adjointness is redundant: By Nelson's commutator theorem, a symmetric $G$-bounded operator $H_0$ which satisfies \eqref{eq:commutator_bound} on a core $\D$ of $G$ is essentially self-adjoint \cite{nelson_time-ordered_1972,frohlich}.

These theorems will follow from more general results about contraction semigroups on $\H$.
This has two benefits: (1) it allows us to prove forward/backward energy-limitedness for unitary dynamics even in cases where energy-limitedness in both time directions might fail, and (2) considering the energy increase of proper contraction semigroups will be useful for our study of open quantum systems later on (see \cref{sec:standard}).

We briefly recall the basics:
A contraction semigroup $\{C(t)\}_{t\ge0}$ on $\H$ is a strongly continuous contraction-valued map $C:\RR^+\to\B(\H)$ such that $C(t)C(s)=C(t+s)$ and $C(0)=1$. 
The generator $K$ of a contraction semigroup is the operator $K\psi = (d/dt) C(t)\psi|_{t=0}$ whose domain consists of all vectors $\psi\in\H$ such that $t\mapsto C(t)\psi$ is $C^1$.
Since the dynamics is contractive, the generator is \emph{dissipative}, i.e., satisfies $K+K^* \le 0$ in the sense that
\begin{equation*}
    \Re\ip\psi{K\psi} = \frac12 \frac{d}{dt}\Big|_{t=0} \norm{C(t)\psi}^2 \le 0 , \qquad \psi\in \dom K.
\end{equation*}
Among all dissipative operators, the generators of contraction semigroups are precisely the \emph{maximally} dissipative ones, those that admit no proper dissipative extensions \cite{arendt_extensions_2023}.
Thus, maximally dissipative operators are for contraction semigroups what self-adjoint operators are for unitary groups.
In fact, an operator $H$ is self-adjoint if and only if $-iH$ and $iH$ are both maximally dissipative.

We say that a contraction semigroup $\{C(t)\}_{t\ge0}$ is energy-limited with stability constants $\omega,E_0$ if
\begin{equation}\label{eq:fC_bound}
    f_{C(t)}(E)\le E + (e^{\omega t}-1)(E+E_0), \qquad E>0,\ t>0.
\end{equation}
The technical backbone of this section is the following Lemma, which reformulates energy-limitedness with respect to the $\sqrt G$-graph norm:
\begin{lemma}\label{lem:contraction_core}
    Let $\{C(t)\}_{t\ge0}$ be a contraction semigroup with generator $K$ and let $T(t)\rho=C(t)\rho C(t)^*$ be the corresponding quantum dynamical semigroup. 
    The following are equivalent:
    \begin{enumerate}[(a)]
        \item\label{it:contraction_core1} For all $t>0$, $C(t)$ leaves $\dom\sqrt G$ invariant and the restrictions $C_0(t)=C(t)|_{\dom\sqrt G}$ form a $\sqrt G$-graph norm-strongly continuous semigroup of bounded operators on $\dom\sqrt G$ (with the graph norm).
        \item\label{it:contraction_core2} For all $t>0$, the cp map $T(t)$ is energy-limited and $t\mapsto\energy[T(t)\rho]$ is continuous for all finite-energy states $\rho$. 
            Furthermore, $\sup_{0\le t\le \delta}f_{T(t)}(E)<\oo$ for some/all $E,\delta>0$.
        \item\label{it:contraction_core3} For all $t>0$, the contraction $C(t)$ is energy-limited and $\energy[C(t)\psi]\to\energy[\psi]$ as $t\to0^+$ for all $\psi\in\dom\sqrt G$.
            Furthermore, $\sup_{0\le t\le \delta}f_{C(t)}(E)<\oo$ for some/all $E,\delta>0$.
    \end{enumerate} 
    If these equivalent properties hold, then $\dom(K\upharpoonright\dom\sqrt G)$ is a common core for $K$ and $\sqrt G$, and 
    $t\mapsto\energy[C(t)\psi]$ is differentiable for all $\psi\in\dom(K\upharpoonright\dom\sqrt G)$ with derivative 
    \begin{equation}\label{eq:contraction_core1}
        \frac d{dt} \,\energy[C(t)\psi]= 2\Re \ip{\sqrt GKC(t)\psi}{\sqrt GC(t)\psi}.
    \end{equation}
\end{lemma}

\begin{proof}
    Equivalence of "some" and "all" in  \ref{it:EL_contraction2} and \ref{it:EL_contraction3} follows from \eqref{eq:concavity_ineq_fT} and \eqref{eq:energy_lim5}.
    \ref{it:contraction_core2} $\Rightarrow$ \ref{it:contraction_core3} is clear.
    
    \ref{it:contraction_core1} $\Rightarrow$ \ref{it:contraction_core2}:
    Let $\rho\in \states_{<\oo}$ with spectral decomposition $\rho=\sum_i \lambda_i\ketbra{\psi_i}{\psi_i}$. Then $\psi_i\in \dom\sqrt G$ and, of course, $(\lambda_i)\in\ell^1$.
    Using dominated convergence, we find
    \begin{equation*}
        \abs[\big]{\energy[T(t)\rho]-\energy[T(s)\rho]} 
        \le\sum_i \lambda_i \abs[\big]{\|\sqrt GC_0(t)\psi_i\|^2 - \|\sqrt GC_0(s)\psi_i\|^2}\xrightarrow{\abs{t-s}\to0}0.
    \end{equation*}
    Thus, $t\mapsto\energy[T(t)\rho]$ is continuous for all $\rho\in\states_{<\oo}$. 
    By \cref{thm:EL_contraction}, the general fact that strongly continuous semigroups of bounded operators are uniformly norm bounded for small times (see \cite[Prop.~I.5]{EngelNagel}) implies that $f_{T(t)}(E)=f_{C(t)}(E)$ is uniformly bounded for small times.

    \ref{it:contraction_core3} $\Rightarrow$ \ref{it:contraction_core1}:
    $C(t)$ leaves $\dom\sqrt G$ invariant and restricts to a $\sqrt G$-graph norm bounded operator because it is energy-limited.
    It suffices to show strong continuity at $t=0$ \cite[Prop.~I.5.3]{EngelNagel}. 
    Since $C(t)$ is already known to be strongly continuous on $\H$, we only need to show $\sqrt G C(t)\psi\to\sqrt G\psi$ as $t\to0^+$.
    For vectors $\psi\in \dom G$, we have  
    \begin{align*}
        \|\sqrt G(C(t)-1)\psi\|^2
        =\underbrace{\norm{\sqrt G C(t)\psi}^2}_{=\energy[C(t)\psi]\to\energy[\psi]}+ \norm{\sqrt G\psi}^2 - 2\underbrace{\Re\ip{G\psi}{C(t)\psi}}_{\to \ip\psi{G\psi}=\energy[\psi]}\to0.
    \end{align*}
    The first term converges by assumption, and the last term converges by strong continuity of $C(t)$ on $\H$.
    Strong convergence extends from the core $\dom G$ to all of $\dom\sqrt G$ since, by assumption, $\norm{\sqrt G C(t)\psi}\le M \norm{\sqrt G\psi}$ for some $M>0$ sufficiently small $t$.

    We now assume the equivalent properties to hold.
    The generator of the strongly continuous semigroup $C_0(t)$ on $\dom\sqrt G$ is the restriction of $K$ to $\dom(K\upharpoonright\dom\sqrt G)$, \cite[Sec.~II.2.3]{EngelNagel}.
    Consequently, $\dom(K\upharpoonright\dom\energy)$ is a core for $K$ because it is dense and $C(t)$-invariant, and a core for $\sqrt G$ because $\dom K_0$ is dense in $\dom\sqrt G$.
    Let $\psi\in\dom(K\upharpoonright\dom\sqrt G)$. Then $t\mapsto C_0(t)\psi$ is $C^1$ with respect to the $\sqrt G$-graph norm or, what is equivalent, $t\mapsto \sqrt GC(t)\psi$ is $C^1$ in $\H$. The derivative is $(d/dt) \sqrt GC(t)\psi = \sqrt GKC(t)\psi$.
    Thus, $t\mapsto \energy[C(t)\psi] = \ip{\sqrt G C(t)\psi}{\sqrt GC(t)\psi}$ is $C^1$ with derivative given by \eqref{eq:contraction_core1}.
\end{proof}

From this, we can deduce a contraction semigroup-version of \cref{thm:unitary}:

\begin{proposition}\label{thm:contraction}
    A contraction semigroup $\{C(t)\}_{t\ge0}$ with generator $K$ is energy-limited with stability constants $\omega,E_0$ if and only if both of the following properties hold:
    \begin{enumerate}[(i)]
        \item\label{it:contr1} 
            For all $t>0$, $C(t)$ leaves $\dom\sqrt G$ invariant and the restrictions $C_0(t)$ to $\dom\sqrt G$ are $\sqrt G$-graph norm bounded and form a $\sqrt G$-graph norm-strongly continuous one-parameter semigroup.
        \item\label{it:contr2}
            The operator inequality $GK+K^*G\le \omega(G+E_0)$ holds in the sense of quadratic forms:
            \begin{equation}\label{eq:contr2}
                2\Re \ip{\sqrt GK\psi}{\sqrt G\psi}\le \omega (\norm{\sqrt G\psi}^2 + E_0\norm\psi^2), \qquad \psi\in\dom(K\upharpoonright\dom\sqrt G)
            \end{equation}
    \end{enumerate}
\end{proposition}
\begin{proof}
    Note that \ref{it:contr1} is one of the equivalent properties of \cref{lem:contraction_core}. Assume that $\{C(t)\}_{t\ge0}$ is energy-limited with stability constants $\omega,E_0$.
    Then \eqref{eq:fC_bound} and lower semicontinuity imply
    \begin{equation*}
        \limsup_{t\to0^+} \energy[C(t)\psi] \le \limsup_{t\to0} e^{\omega t}(\energy[\psi]-E_0\norm\psi^2)-E_0\norm\psi^2 = \energy[\psi] \le \liminf_{t\to0^+}\energy[C(t)\psi]
    \end{equation*}
    for $\psi\in\dom\sqrt G$. Since this implies $\lim_{t\to0^+}\energy[C(t)\psi]=\energy[\psi]$, the equivalent properties of \cref{lem:contraction_core} hold (in particular, \ref{it:contr1} holds). 
    Let $\psi\in\dom(K\upharpoonright\dom\sqrt G)$ be a unit vector with energy $\energy[\psi]=E$. Then we have
    \[
        \energy[C(t)\psi]-E \le e^{\omega t}(E+E_0)-(E+E_0).
    \] 
    If we divide both sides by $t$ and take the limit $t\to 0^+$, \cref{lem:contraction_core} shows
    \[
        2\Re \ip{\sqrt GK\psi}{\sqrt G\psi} = \frac d{dt}\energy[C(t)\psi] |_{t=0} \le \omega(E+E_0)= \omega(\norm{\sqrt G\psi}^2+E_0\norm\psi^2).
    \] 

    Conversely, \ref{it:contr1} implies that $t\mapsto \energy[C(t)]$ is $C^1$ with $(d/dt)\energy[C(t)\psi]=2\Re \ip{\sqrt GKC(t)\psi}{\sqrt GC(t)\psi}$, and \ref{it:contr2} implies that the right-hand side is bounded by $\omega(\norm{\sqrt GC(t)\psi}^2 + E_0 \norm{C(t)\psi}^2) \le \omega(\energy[C(t)\psi] + E_0\norm\psi^2)$.
    Thus, the $C^1$ function $F(t)=\energy[C(t)\psi]+E_0\norm\psi^2$ satisfies $F'(t) \le \omega F(t)$ and
    Gronwall's Lemma \cite[App.~B.2.j]{evans2022partial} gives $F(t) \le e^{\omega t} F(0)$. 
    Thus, we have $\energy[C(t)\psi] \le e^{\omega t}(\energy[\psi]+E_0\norm\psi^2)-E_0\norm \psi^2$ for all $\psi\in \dom(K \restriction \sqrt G)$.
    Since $\dom(K\restriction\sqrt G)$ is a core for $\sqrt G$, the same holds for all vectors $\psi\in\dom\sqrt G$, i.e., $C(t)$ is energy-limited with stability constants $\omega,E_0$.
\end{proof}

\cref{thm:unitary} is immediate from \cref{thm:contraction} because a unitary group $\{U(t)\}_{t\in\RR}$ is energy-limited if and only if it is forward and backward energy-limited.
Similarly, \cref{thm:unitary2} follows from:

\begin{proposition}\label{thm:dissipative}
    Let $K$ be a maximally dissipative operator such that $\dom K$ contains a core $\D$ of $G$ on which $K$ is $G$-bounded and satisfies $K^*G+GK\le \omega(G+E_0)$ for some $\omega,E_0\ge0$, in the sense that 
    \begin{equation}\label{eq:dissipative}
        2\Re\ip{K\psi}{G\psi}\le \omega\ip\psi{(G+E_0)\psi}, \qquad \psi\in\D.
    \end{equation}
    Then, the contraction semigroup generated by $K$ is energy-limited with stability constants $\omega,E_0\ge0$.
\end{proposition}

As with \cref{thm:unitary2}, the assumption that $K$ is a generator is redundant: If $K:\D\to\H$ is a dissipative $G$-operator satisfying \eqref{eq:dissipative}, then $\bar K$ is maximally dissipative.
This follows from the generalization of Nelson's commutator theorem in \cref{sec:generation}.
Another consequence of this, and an important step in the proof, is that $\dom G$ is a core for $K$.

\begin{proof}
    Since $K$ is $G$-bounded on a core for $G$, we know that $K$ is $G$-bounded on $\dom G\subset \dom K$ as well.
    By taking $G$-graph norm limits, \eqref{eq:dissipative} extends to all $\psi\in \dom G$. Thus, we may simply assume $\D=\dom G$ in the following.
    The following is inspired by the proof of \cite[Lem.~2]{frohlich}.
    % Furthermore, we may assume $E_0>0$. 
    % Indeed, the case $E_0=0$ then follows from considering $E_0'=\eps>0$, which then yields energy-limitedness with stability constants $\omega,\eps$ for all $\eps>0$. Since the limit $\eps\to0$ of $Ee^{\omega t} +(e^{\omega t}-1)E_0$ is simply $Ee^{\omega t}$, it follows that $\omega,0$ are stability constants.

    \emph{Step 1.} We start by showing the claim for a bounded approximation of $K$. 
    We use the resolvent-type operators $R_\eps = (1+\eps G)^{-1}$, $\eps>0$, to define a regularized generator 
    \begin{equation*}
        K_\eps = R_\eps K R_\eps.
    \end{equation*}
    Since $K$ is $G$-bounded, $K_\eps$ is a bounded dissipative operator and $e^{tK_\eps} = \sum_{n=0}^\oo (t^n/n!) K_\eps^n$ is a contraction semigroup.
    From spectral theory and $G$-boundedness of $K$, it follows that $\sqrt{G}K_\eps$ is bounded as well.
    For $\psi\in \dom \sqrt{G}$, the estimate
    \begin{equation}
        \norm{\sqrt G e^{tK_\eps}\psi} \le \sum_{n=0}^\oo \frac{t^n}{n!}\norm{\sqrt GK_\eps^n\psi} \le \norm{\sqrt G\psi} + \norm\psi\sum_{n=1}^\oo \frac{t^n}{n!}\norm{\sqrt GK_\eps} \norm{K_\eps}^{n-1} < \oo
    \end{equation}
    shows that $e^{tK_\eps}$ leaves $\dom\sqrt G$ invariant (cp.~\cite{frohlich}) and the estimate
    \begin{equation*}\label{eq:diff}
        \big\|\sqrt{G}\big(\tfrac1t(e^{tK_\eps}\psi-\psi)-K_\eps\psi\big)\big\|
        % = \bigg\|\sum_{n=2}^\oo \frac{t^{n-1}}{n!} \sqrt{G}K_\eps^n\psi \bigg\|
        \le \sum_{n=2}^\oo \frac{t^{n-1}}{n!} \norm{\sqrt{G}K_\eps} \norm{K_\eps}^{n-1}\norm\psi <\oo,
    \end{equation*}
    shows that $\RR^+\ni t\mapsto \sqrt{G} e^{tK_\eps}\psi\in\H$ is a continuously differentiable map with derivative $\sqrt{G}K_\eps e^{tK_\eps}\psi$ (cp.~\cite{frohlich}). 
    Therefore, $\{e^{tK_\eps}\}_{t\ge0}$ satisfies condition \ref{it:contr1} of \cref{thm:contraction}.
    Next we check condition \ref{it:contr2}:
    If $\psi\in\dom \sqrt G= \dom(K_\eps\restriction\dom\sqrt G)$ then
    \begin{align*}
        2\Re \ip{\sqrt{G}K_\eps\psi}{\sqrt G\psi} = 2\Re \ip{GK_\eps R_\eps\psi}{\sqrt GR_\eps\psi}\le \omega \ip{R_\eps\psi}{(G+E_0)R_\eps\psi}\le \omega\ip\psi{(G+E_0)\psi}
    \end{align*}
    where we applied \eqref{eq:dissipative}, which is allowed since $R_\eps\psi \in \dom G=\D$. 
    Therefore, \cref{thm:contraction} shows that $\{e^{tK_\eps}\}_{t\ge0}$ is energy-limited with stability constants $\omega,E_0$.
    
    \emph{Step 2.} In this step, we take the limit $\eps\to 0$.
    Since $K$ is $G$-bounded, $X = K(1+G)^{-1}$ is a bounded operator on $\H$ and $K = X(1+G)$ on $\dom G$.
    Since $R_\eps$ converges strongly to the identity as $\eps\to0$, it follows that $K_\eps = R_\eps XR_\eps (1+G)$ converges strongly to $K$ strongly on $\dom G$.
    By \cref{thm:generation} $\dom G$ is a core for the generator $K$.
    The Trotter-Kato approximation theorem \cite[Thm.~III.4.8]{EngelNagel} now shows that $e^{tK_\eps}$ converges strongly to $e^{tK}$ as $\eps\to0$.
    By item \ref{it:energy_lim3} of \cref{lem:energy_lim}, this implies that $\{e^{tK}\}_{t\ge0}$ is also energy-limited with stability constants $\omega,E_0$.
\end{proof}

\subsection{General open systems}\label{sec:general_open}

An open quantum system is a quantum system with irreversible dynamics, i.e., non-unitary, dynamics.
In the Markovian case, these dynamics are described by quantum dynamical semigroups, which are strongly continuous one-parameter semigroups of trace-nonincreasing cp maps.
% We can understand the behavior of open systems by studying the infinitesimal generator of the dynamics.
The generator $\L$ of a quantum dynamical semigroup  $\{T(t)\}_{t\ge0}$, is defined as
\begin{equation}
    \L\rho = \lim_{t\to0^+} t^{-1}(T(t)\rho-\rho)
\end{equation}
on the domain $ \dom\L = \{\rho\in\TC : [t\mapsto T(t)\rho] \in C^1(\RR^+,\TC)\}$.
Importantly, $\{T(t)\}_{t\ge0}$ is conservative, i.e., each $T(t)$ is trace-preserving, if and only if 
\begin{equation}
    \tr\L\rho = 0,\qquad \rho\in\dom\L.
\end{equation}
In general, we only have $\tr\L\rho\le 0$ for $0\le \rho\in\dom\L$.
The semigroup can be recovered from its generator because $T(t)\rho$ is the unique solution to the initial value problem $\dot\rho(t)=\L\rho$, $\rho(0)=\rho$ for $\rho\in\dom\L$.
This is summarized by writing $T(t)=e^{t\L}$.
We denote the generator of the dual (Heisenberg-picture) semigroup $\{T^*(t)\}_{t\ge0}$ on $\B(\H)$, which is strongly continuous for the $\sigma$-weak operator topology, by $\L^*$.
Our goal is to understand energy-limitedness in terms of the generator.
We define
\begin{equation}\label{eq:domL_E}
    \dom(\L\upharpoonright \dom\energy) = \big\{\rho\in\dom\L\cap\dom\energy : \L\rho\in\dom\energy\big\}.
\end{equation}
Our main result is the following:

\begin{theorem}\label{thm:main}
    Let $\{T(t)\}_{t\ge0}$ be a quantum dynamical semigroup with generator $\L$ and resolvents $\R(\lambda)=(\lambda-\L)^{-1}$.
    The following are equivalent:
    \begin{enumerate}[(a)]
        \item\label{it:main1}
            $\{T(t)\}_{t\ge0}$ is energy-limited with stability constants $\omega,E_0$.
        \item\label{it:main2}
            For all $\lambda>\omega$, the output energy of the resolvents is bounded by 
            \begin{equation}\label{eq:main2}
                \energy[\R(\lambda)\rho]\le \frac1{\lambda-\omega}\Big(\energy[\rho] + \frac \omega\lambda E_0\tr\rho\Big),\qquad\rho\in\TC^+.
            \end{equation}
        \item\label{it:main3}
            $\dom\energy = (\lambda-\L)\dom(\L\upharpoonright\dom\energy) $ and the operator inequality $\L^*(G)\le \omega(G+E_0)$ holds in the sense that 
            \begin{equation}\label{eq:main3}
                \energy[\L\rho]\le \omega(\energy[\rho]+E_0\tr\rho),\qquad  0\le\rho\in\dom(\L\upharpoonright\dom\energy).
            \end{equation}
    \end{enumerate}
    In this case, $\dom(\L\upharpoonright\dom\energy)$ is a $T(t)$-invariant core for $\energy$ and for $\L$.
\end{theorem}

Before we give the proof, we recall a few properties of resolvents. 
Let $T(t)$ be a quantum dynamical semigroup with generator $\L$.
It can be shown that $(\lambda-\L)$ is surjective for all $\lambda>0$, and that the inverse $\R(\lambda)=(\lambda-\L)^{-1}$, the resolvent, is a bounded operator on $\TC$. 
Equivalently, the resolvents can be expressed as the Laplace transforms of the semigroup:
\begin{equation}\label{eq:laplace_trafo}
    \R(\lambda)\rho=\int_0^\oo e^{-\lambda t} T(t)\rho\,dt,\qquad\rho\in\TC.
\end{equation}
This implies $\tr\R(\lambda)\rho \le \lambda^{-1}\tr\rho$ for $\rho\in\TC^+$.
The range of the resolvents is exactly the domain of the generator $\Ran\R(\lambda)=\dom\L$ and $\L$ can be recovered from the resolvents via
\begin{equation}\label{eq:resolvent_LR_eq}
    \L\R(\lambda)\rho = \lambda\R(\lambda)\rho-\rho, \qquad\rho\in\TC.
\end{equation}
The dynamics can be recovered directly from the resolvents via 
\begin{equation}\label{eq:T_from_R}
    T(t)\rho = \lim_{n\to\oo} \big(\tfrac nt \R(\tfrac nt)\big)^n\rho,\qquad\rho\in\TC.
\end{equation}
By \eqref{eq:laplace_trafo} and \eqref{eq:T_from_R} the resolvents are cp if and only if the dynamics is.
We also note the formulae
\begin{equation}\label{eq:lambda_R_convergence}
    \lim_{\lambda\to\oo}\lambda\R(\lambda)\rho=\rho,\qandq \lim_{\lambda\to\oo}\lambda\L\R(\lambda)\rho = \L\rho,\qquad \rho\in\dom\L.
\end{equation}
The first limit even holds for all $\rho\in\TC$.
All of these statements (and many more) can be found in \cite{EngelNagel}.
We need the following immediate consequence of item \ref{it:elementary3} of \cref{lem:elementary}:

\begin{lemma}\label{thm:energy_integral}
    Let $T(t)$ be a quantum dynamical semigroup and let $0\le p\in L^1(\RR^+)$.
    Then $t\mapsto \energy[T(t)\rho]$ is a Borel measurable $\bar\RR^+$-valued map for all $\rho\in\TC^+$ and
    \begin{equation}\label{eq:energy_integral}
        \energy\bigg[\int p(s) T(s)\rho\,ds\bigg] =\int p(s) \energy[T(s)\rho]\,ds, \qquad\rho\in\TC^+,
    \end{equation}
    where both sides may be infinite.
    If $p(t)\energy[T(t)\rho]$ is integrable for all finite-energy states, equation \eqref{eq:energy_integral} extends linearly to all $\rho\in\dom\energy$.
\end{lemma}
\hide{\begin{proof}
    Borel measurability holds for $\rho\in\TC^+$ because $t\mapsto\energy[T(t)\rho]$ is lower semicontinuous. By linearity, measurability extends to $\dom\energy$.
    Denote by $P_n$ the spectral projection of $G$ of the interval $[0,n]$ and set $G_n=P_nG$.
    For $\rho\in\TC^+$ define $h_n(t) = \tr[G_nT(t)\rho]$ and observe that $h_1\le h_2\le \ldots$ is monotonically increasing.
    By definition of the energy functional, the pointwise limit  $h(t):=\lim_{n\to\oo} h_n(t)$ is equal to $g(t)=\energy[T(t)\rho]$. Therefore the monotone convergence (Beppo-Levi) theorem implies $\lim_{n\to\oo}\int h_n(t)\,d\mu(t) = \int h(t)d\mu(t)$.
    Since $\int T(s)\rho\,d\mu(s)$ exists as a Bochner integral in $\TC$, we conclude
    \begin{equation*}
        \energy\bigg[\int T(s)\rho\,d\mu(s)\bigg]
        =\lim_n \tr\bigg[{ G_n\! \int T(s)\rho\,d\mu(s)}
        =\lim_n \int \underbrace{\tr[G_n T(s)\rho]}_{=h_n(s)}\,d\mu(s)
        =\int \underbrace{\energy[T(s)\rho]}_{=h(s)}\,d\mu(s).
    \end{equation*}
\end{proof}}

\begin{lemma}\label{thm:open_core}
    Let $\L$ be the generator of an energy-limited quantum dynamical semigroup with stability constants $\omega,E_0$.
    If $\lambda>\omega$, then $\lambda\R(\lambda)$ is an energy-limited trace-nonincreasing cp map and
    \begin{equation}\label{eq:open_core}
        \dom(\L\upharpoonright\dom\energy) = \R(\lambda)\dom\energy   
    \end{equation}
    is a $T(t)$-invariant core for $\L$. Furthermore, $t\mapsto e^{-\lambda t}\,\energy[T(t)\rho]$ is $L^1$ and
    \begin{equation}\label{eq:open_core2}
        \energy[\R(\lambda)\rho] =\int_0^\oo e^{-\lambda t}\energy[T(t)\rho]\,dt,\qquad\rho\in\dom\energy.
    \end{equation}
\end{lemma}

\begin{proof}
    Eq.~\eqref{eq:open_core2} follows from \cref{thm:energy_integral} and the observation that the bound \eqref{eq:affine_semigroup_bound} implies that $e^{-\lambda t}\,\energy[T(t)\rho]$ is $L^1$ for all $0\le\rho\in\dom\energy$. Indeed, 
    \begin{align*}
        \energy[\R(\lambda)\rho]=\int_0^\oo e^{-\lambda t}\energy[T(t)\rho]\,dt &\le \int_0^\oo \big[e^{(\omega-\lambda)} (\energy[\rho]+E_0\tr\rho)- e^{-\lambda t}E_0\tr\rho\big]\,dt
        \\ &= \frac{1}{\lambda-\omega}(\energy[\rho]+E_0\tr\rho)-\frac1\lambda E_0\tr\rho<\oo.
    \end{align*}
    Another consequence of this is that $\R(\lambda)$ maps $\dom\energy$ into itself, which shows $\R(\lambda)\dom\energy \subseteq \dom(\L\upharpoonright\dom\energy)$.
    Conversely, let $\rho\in\dom(\L\upharpoonright\dom\energy)$. By definition of $\dom(\L\upharpoonright\dom\energy)$, we have $\sigma = (\lambda-\L)\rho\in\dom\energy$ and, hence, $\rho=\R(\lambda)\sigma \in\R(\lambda)\dom\energy$.
    We conclude that \eqref{eq:open_core} holds which, in particular, implies that $\R(\lambda)\dom\energy$ does not depend on $\lambda>\omega$.
    Applying \cref{thm:energy_integral} and \eqref{eq:laplace_trafo}, we see that \eqref{eq:open_core2} holds.

    Eq.~\eqref{eq:open_core} makes it evident that $\dom(\L\upharpoonright\dom\energy)$ is a $T(t)$-invariant core for $\L$.
    Indeed, density in trace-norm holds because $\rho=\lim_{\lambda\to\oo} \lambda\R(\lambda)\rho$ for all $\rho\in\dom\energy$ shows that we can approximate a dense set of elements (namely, $\dom\energy$) with elements of $\R(\lambda)\dom\energy$.
    Furthermore, $\R(\lambda)\dom\energy$ is $T(t)$-invariant because each $T(t)$ is energy-limited and because $T(t)$ commutes with the resolvents. Thus, it follows from the core theorem (see \cite[Prop.~II.1.7]{EngelNagel}) that $\R(\lambda)\dom\energy$ is a core.    
\end{proof}

\begin{proof}[Proof of \cref{thm:main}]
    \ref{it:main1} $\Rightarrow$ \ref{it:main2}:
    Let $\rho$ be a state with energy $E$. Then \eqref{eq:laplace_trafo} and \eqref{eq:open_core2} imply 
    \begin{equation*}
        \energy[\R(\lambda)\rho] 
        = \int_0^\oo \!e^{-\lambda t} \energy[T(t)\rho]\,dt
        \le \int_0^\oo \big[e^{(\omega-\lambda)t} (E+E_0)-e^{-\lambda t}E_0\big]\,dt
        = \frac{E+E_0}{\lambda-\omega} - \frac{E_0}\lambda,
    \end{equation*}
    where we used \eqref{eq:affine_semigroup_bound}.
    The right-hand side can be rearranged to give \eqref{eq:main2}.
    %\[
    %   \energy[\R(\lambda)\rho]\le\frac1{\lambda-\omega}(E+E_0)-\frac1\lambda E_0 = \frac1{\lambda-\omega}(E+E_0-(1-\frac\omega\lambda)E_0 ) = \frac1{\lambda-\omega}(E+ \frac\omega\lambda E_0).
    %\] 

    \ref{it:main2} $\Rightarrow$ \ref{it:main1}:
    Let $\rho\in\states_{<\oo}$. Iterating \eqref{eq:main2} $n$ times, we find
    \begin{align*}
        \energy[\lambda^n \R(\lambda)^n\rho]
        &\le \paren[\Big]{1-\frac\omega\lambda}^{-n}\energy[\rho] + \frac{\omega E_0}{\lambda} \sum_{k=1}^n \paren[\Big]{1-\frac\omega\lambda}^{-k}\\
        &=\paren[\Big]{1-\frac\omega\lambda}^{-n}\energy[\rho] + \frac{\omega E_0}\lambda (1-\frac\omega\lambda)^{-1} 
        \frac{1-\paren[\Big]{1-\frac\omega\lambda}^{-n}}{1-\paren[\Big]{1-\frac\omega\lambda}^{-1}}\\
        &=\paren[\Big]{1-\frac\omega\lambda}^{-n}\energy[\rho] - E_0 \paren*{1-\paren[\Big]{1-\frac\omega\lambda}^{-n}}.
    \end{align*}
    If we put $\lambda = \frac nt$, the right-hand side converges to $e^{\omega t}(\energy[\rho]+E_0)-E_0$ as $n\to \oo$.
    We conclude from lower semicontinuity and \eqref{eq:T_from_R} that
    \begin{align*}
        \energy[T(t)\rho] \le \liminf_n \energy\Big[\Big(\frac nt \R\Big(\frac nt\Big)\!\Big)^n\rho\Big] \le e^{\omega t}(\energy[\rho]+E_0)-E_0
    \end{align*}

    \ref{it:main1} and \ref{it:main2} $\Rightarrow$ \ref{it:main3}:
    \Cref{thm:open_core} implies $(\lambda-\L)\dom(\lambda\restriction\dom\energy)=\dom\energy$.
    We start by showing $\energy[\lambda\R(\lambda)\rho]\to \energy[\rho]$ as $\lambda\to\oo$ for all $\rho\in\dom\energy$.
    Indeed, this follows from \eqref{eq:lambda_R_convergence} and lower-semicontinuity
    \begin{equation}\label{eq:core}
        \energy[\rho]\le \liminf_{\lambda\to\oo} \energy[\lambda\R(\lambda)\rho] \le \limsup_{\lambda\to\oo}\energy[\lambda\R(\lambda)\rho] \le \lim_{\lambda\to\oo} \frac\lambda{\lambda-\omega}(\energy[\rho] + \frac\omega\lambda E_0\tr\rho) = \energy[\rho]
    \end{equation}
    for $\rho\in\TC^+$.
    Extending this linearly, we get $\energy[\lambda\R(\lambda)\rho]\to \energy[\rho]$ for all $\rho\in\dom\energy$.
    To show \eqref{eq:main3}, we combine the inequality \eqref{eq:main2} with eq.~\eqref{eq:resolvent_LR_eq} and obtain
    \begin{align*}
        \energy[\L\R(\lambda)\rho]
        =\energy[\lambda\R(\lambda)\rho] - \energy[\rho]
        \le \frac{\lambda\energy[\rho]+\omega E_0\tr\rho}{\lambda-\omega} - \energy[\rho]
        =\omega\frac{{\energy[\rho]+E_0\tr\rho}}{\lambda-\omega} 
    \end{align*}
    for all $0\le\rho\in\dom\energy$.
    Using $\lim_{\lambda\to\oo}\energy[\lambda\R(\lambda)\sigma]=\energy[\sigma]$ for all $\sigma\in\dom\energy$, we find
    \begin{align*}
        \energy[\L\rho]
        =\lim_{\lambda\to\oo} \energy[\lambda\R(\lambda)\L\rho] 
        &= \lim_{\lambda\to\oo}\lambda\energy[\L\R(\lambda)\rho]\\
        &\le \omega \liminf_{\lambda\to\oo} \frac\lambda{\lambda-\omega} (\energy[\rho]+E_0\tr\rho)
        =\omega(\energy[\rho]+E_0\tr\rho)
    \end{align*}
    for $0\le \rho\in \dom(\L\upharpoonright\dom\energy)$.

    \ref{it:main3} $\Rightarrow$ \ref{it:main2}:
    We can reformulate the first assumption as $\R(\lambda)\dom\energy=\dom(\L\upharpoonright\dom\energy)$.
    Let $0\le \rho\in\dom\energy$, then $0\le \R(\lambda)\rho\in \dom(\L\upharpoonright\dom\energy)$.
    Thus, \eqref{eq:main3} and \eqref{eq:resolvent_LR_eq} imply
    \begin{align*}
        \lambda\energy[ \R(\lambda)\rho] 
        =\energy[\L\R(\lambda)\rho] + \energy[\rho]
        \le \omega(\energy[\R(\lambda)\rho]+E_0\tr\R(\lambda)\rho)+ \energy[\rho].
    \end{align*}
    Rearranging and using the estimate $\tr\R(\lambda)\rho\le \lambda^{-1}\tr\rho$, 
    \begin{equation*}
        (\lambda-\omega)\energy[\R(\lambda)\rho]\le \energy[\rho]+\frac\omega\lambda E_0\tr\rho,
    \end{equation*}
    which shows that \eqref{eq:main2} holds for all $\rho\in\dom \energy$ and hence for all $\rho\in\TC^+$.
\end{proof}

\begin{remark}\label{rem:nonzero_gse}
    The proof of \cref{thm:main} does not require the ground state energy of the reference Hamiltonian $G$ to be zero if $f_{T(t)}(E)$ is defined by \eqref{eq:EL} for an arbitrary self-adjoint operator $G\ge0$.
\end{remark}

\begin{remark}
    Consider $\dom\energy$ as a Banach space equipped with the norm $\normiii{}_1$ (see \cref{sec:setup}).
    If $\{T(t)\}_{t\ge0}$ is an energy-limited dynamical semigroup, then each $T(t)$ is a bounded operator on $\dom\energy$ with operator norm scaling as $e^{\lambda t}$ for some $\lambda>0$.
    To the best of the author's knowledge, the equivalent statements in \cref{thm:main} do not imply that $T(t)$ is strongly continuous for the $\normiii{}_1$-norm.
    Instead, it seems that $\normiii{}_1$-strong continuity is an additional property, which can be restated in several equivalent forms.
    Indeed, the following are equivalent if $\{T(t)\}_{t\ge0}$ satisfies the equivalent properties of \cref{thm:main}:
    \begin{enumerate}[(i)]
        \item $T(t)\upharpoonright\dom\energy$ is strongly continuous for the norm $\normiii{}_1$.\footnote{In the language of \cite[Sec.~4.5]{pazy1983}, this means that $(\dom\energy,\normiii{}_1)\hookrightarrow (\T(\H),\norm\placeholder_1)$ is an admissible subspace.}
        \item $\dom(\L\upharpoonright\dom\energy)=\R(\lambda)\dom\energy$ is a $\normiii{}_1$-dense subspace of $\dom\energy$.
        \item Consider the isometric isomorphism $W=\sqrt{G+1}(\placeholder)\sqrt{G+1}:(\dom\energy,\normiii{}_1)\to\T(\H)$ (see \cref{lem:base}). The operator $\L' = W\L W^{-1}$ with domain
            $\dom\L' = %W\dom(\L\upharpoonright\dom\energy) =
            \{\rho\in\T(\H) : W^{-1}\rho\in\dom\L,\ W\L W^{-1}\rho\in\T(\H)\}$
        generates a strongly continuous one-parameter semigroup on $\T(\H)$.
    \end{enumerate}
\end{remark}

\begin{remark}
    We are not aware of a connection between energy-limitedness and differentiability of the dynamical semigroup with respect to the ECD norm, studied in \cite{shirokov2019}.
\end{remark}

\subsection{Standard generators}\label{sec:standard}
As Lindblad famously proved \cite{lindblad}, generators of uniformly continuous quantum dynamical semigroups have the form:
\begin{equation}\label{eq:std_gen}
    \L\rho = K\rho + \rho K^* +\sum_\alpha L_\alpha\rho L_\alpha^*,\qquad \sum_\alpha L_\alpha^*L_\alpha \le -(K+K^*),
\end{equation}
where $K$ is a bounded dissipative operator and $L_\alpha$ are bounded operators.
Uniformly continuous dynamics are conservative if and only if the infinitesimal conservativity condition $\sum_\alpha L_\alpha^*L=-(K^*+K)$  is satisfied.
It helps to think about \eqref{eq:std_gen} as a perturbation of the generator $\L_0\rho = K\rho+\rho K^*$ by the cp map $\P\rho=\sum_\alpha L_\alpha\rho L_\alpha^*$.
The unperturbed dynamics is $T_0(t)\rho = C(t)\rho C(t)^*$ where $C(t)=e^{tK}$ is the contraction semigroup generated by $K$. 
% Because of their role in arrival time measurements, dynamical semigroup $\{T_0(t)\}_{t\ge0}$ that are of this form for some contraction semigroup $\{C(t)\}_{t\ge0}$ on $\H$ are called \emph{no-event} semigroups \cite{siemon_unbounded_2017,werner_arrival_1987}.
The operator inequality in \eqref{eq:std_gen}, which can be restated as $\tr\L\rho =\tr\L_0\rho+\tr\P\rho \le 0$ for all states $\rho$, ensures that the perturbed dynamics is trace-nonincreasing. 

Quantum dynamical semigroups are rarely uniformly continuous in infinite-dimensional Hilbert spaces, so the story does not end here.
In \cite{siemon_unbounded_2017}, \emph{standard generators} are defined by generalizing \eqref{eq:std_gen}:
% to the unbounded case through the so-called minimal solution procedure.%
% \footnote{If $\L_0$ is the generator of a quantum dynamical semigroup with loss, e.g., a no-event semigroup $T_0(t) = C(t)(\placeholder)C(t)^*$, and $\P:\dom\L_0\to\T(\H)$ is a cp perturbation such that $\tr\L_0\rho+\tr\P\rho\le0$ for all $0\le \rho\in\dom\L_0$, then, among all extensions $\L\supset \L_0+\P$ that generated quantum dynamics semigroups, there is a unique solution $\Lmin$ which is minimal in the sense that $e^{t\Lmin}-e^{t\L}$ is cp for all $t\ge0$ and all other extensions $\L$.}
A standard generator $\L$ is determined by a pair $(K,\{L_\alpha\})$ of a maximally dissipative operator $K$ on $\H$ and a collection of operators $L_\alpha:\dom K\to\H$ satisfying 
\begin{equation}\label{eq:L_summation}
    \sum_\alpha\norm{L_\alpha\psi}^2 \le -2\Re\ip\psi{K\psi},\qquad \psi\in\dom K
\end{equation}
Roughly speaking, it is defined as the so-called minimal solution to the problem of perturbing the generator $\L_0\rho = K \rho+\rho K^*$ of the semigroup $T_0(t)\rho = e^{tK}\rho(e^{tK})^*$ by the cp map with Kraus operators $\{L_\alpha\}$, see \cite{siemon_unbounded_2017} for details.
This definition guarantees that $\dom\L$ contains the ketbra domain $(\dom K)^{\kettbra{}} = \lin\{\ketbra\psi\phi : \psi,\phi\in\dom K\}$ on which it acts via
\begin{equation}
    \L\ketbra\psi\phi = \ketbra{K\psi}\phi +\ketbra{\psi}{K\phi}+\sum_\alpha \ketbra{L_\alpha\psi}{L_\alpha\phi},\qquad \psi,\phi\in\dom K.
\end{equation}
The standard generator $\L$ is called \emph{formally conservative} if equality holds in \eqref{eq:L_summation}, i.e., if
\begin{equation}\label{eq:L_summation_eq}
    \sum_\alpha\norm{L_\alpha\psi}=-2\Re\ip\psi{K\psi},\qquad \psi\in\dom K.
\end{equation}
Unlike the uniformly continuous case, formally conservative generators do not necessarily generate conservative dynamics; see \cite{davies_quantum_1977, fagnola1999, inken_thesis}.
This phenomenon also occurs in classical systems and can often be regarded as an escape to infinity in finite time of certain parts of the system.
In \cref{sec:birth}, we will consider an example of a formally conservative generator that generates nonconservative dynamics.
Davies showed that a formally conservative standard generator generates conservative dynamics if and only if the ketbra domain $(\dom K)^{\kettbra{}}$ (see \cref{lem:base}) is a core for $\L$ (see \cite[Prop.~3.32]{fagnola1999}).
We note the following slight generalization:

\begin{lemma}\label{lem:davies}
    If  $\L$ is formally conservative, i.e., \eqref{eq:L_summation_eq} holds, then the following are equivalent:
    \begin{enumerate}[(a)]
        \item\label{it:davies1} $\{e^{t\L}\}_{t\ge0}$ is conservative,
        \item\label{it:davies2} $(\dom K)^{\kettbra{}}$ is a core for $\L$,
        \item\label{it:davies3} For every core $\D\subset \dom K$ for $K$, the ketbra domain $\D^{\kettbra{}}$ is a core for $\L$.
    \end{enumerate}
\end{lemma}

\begin{proof}
    \ref{it:davies1} $\Leftrightarrow$ \ref{it:davies2} is shown in \cite[Prop.~3.32]{fagnola1999} and \cite[Prop.~4.4.2]{siemon_unbounded_2017}.
    \ref{it:davies3} $\Rightarrow$ \ref{it:davies2} is clear. 
    \ref{it:davies1} \& \ref{it:davies2} $\Rightarrow$ \ref{it:davies3}:
    Let $\L_0$ be the generator of the semigroup $T_0(t)=e^{tK}(\placeholder)e^{tK^*}$.
    Since $\L_0\ketbra\psi\phi =\ketbra{K\psi}\phi+\ketbra\psi{K\phi}$ for $\ketbra\psi\phi\in (\dom K)^{\kettbra{}}\subset \dom\L_0$, elements of $(\dom K)^{\kettbra{}}$ can approximated in $\L_0$-graph norm by elements in $\D^{\kettbra{}}$ for a core $\D$ of $K$.%
    \footnote{Indeed, given $\psi,\phi\in \dom K$ and $\psi',\phi'\in\D$, consider $\norm{\L_0\ketbra\psi\phi-\L_0\ketbra{\psi'}{\phi'}}_1 \le \norm{\L_0\ketbra{\psi-\psi'}\phi}_1 + \norm{\L_0\ketbra{\psi'}{\phi-\phi'}}_1 \le \norm{K(\psi-\psi')} \norm\phi + \norm{\psi-\psi'}\norm{K\phi} + \norm{K\psi'}\norm{\phi-\phi'}+\norm{\psi'}\norm{K(\phi-\phi')}$, which can be made arbitrarily small by choice of $\psi',\phi'$.}
    Since $(\dom K)^{\kettbra{}}$ is a core for $\L_0$, this implies that $\D^{\kettbra{}}$ is a core for $\L_0$ as well. It remains to show that it is also a core for $\L$.
    This follows from the argument in the proof of \cite[Prop.~4.4.2]{inken_thesis}, which only needs that $(\dom K)^{\kettbra{}}$ is a core for $\L_0$ to show that it is a core for $\L$.\footnote{The proof of \cite[Prop.~4.4.2]{inken_thesis} shows more than they state: The minimal solution $\Lmin$ to a cp peturbation theorem $\L_0+\P$ is conservative if and only if $\dom\L_0$ is a core for $\Lmin$ if and only if every core for $\L_0$ is a core for $\Lmin$.}
\end{proof}

If $\L$ is the standard generator determined by $K$ and $\{L_\alpha\}$ as above, then \cref{thm:main} asserts that energy-limitedness with stability constants $\omega,E_0$ is \emph{formally} equivalent to the operator inequality $ K^*G+GK+\sum_\alpha L_\alpha^*GL_\alpha \le \omega(G+E_0)$.
Since it is hard to characterize the full domain of a standard generator, the condition in \cref{thm:main} is hard to verify.
The following result provides sufficient conditions that can be checked in practice:

\begin{theorem}\label{thm:standard}
    Let $K$ be maximally dissipative, let $L_\alpha:\dom K\to\H$ be operators satisfying \eqref{eq:L_summation_eq}, let $\L$ be the standard generator formally given by \eqref{eq:std_gen}.
    Assume that the semigroup $\{T(t)\}_{t\ge0}$ generated by $\L$ is conservative.
    If $\D\subseteq \dom K\cap\dom G$ is a core for $G$ such that $K$ is $G$-bounded on $\D$ such that 
    \begin{equation}\label{eq:standard}
        2\Re \ip{G\psi}{K\psi} +\sum_\alpha \norm{\sqrt GL_\alpha\psi}^2 \le \omega(\norm{\sqrt G\psi}^2+E_0\norm\psi^2),\qquad \psi\in\D
    \end{equation}
    for constants $\omega,E_0\ge0$.
    Then $\{T(t)\}_{t\ge0}$ is energy-limited with stability constants $\omega,E_0$.
\end{theorem}

In \eqref{eq:standard}, we use the convention that $\norm{\sqrt G\phi}=\oo$ if $\phi\not\in\dom\sqrt G$. Thus, for \eqref{eq:standard} to hold it is necessary that $L_\alpha\D\subset\dom\sqrt G$.
We start with the following preparatory result:

\begin{lemma}\label{lem:bounded}
    Let $\{T(t)\}_{t\ge0}$ be a uniformly continuous dynamical semigroup with generator $\L$ and let $K,L_\alpha$ be bounded operators such that \eqref{eq:std_gen} holds. 
    We set $L = \sum_\alpha L_\alpha\ox \ket\alpha \in\B(\H,\H\ox\ell^2)$.
    Assume that
    \begin{itemize}
        \item $L\dom\sqrt G\subset \dom(\sqrt G\ox1)$, and $K\dom\sqrt G\subset \dom\sqrt G$
        \item $\tilde K=ZK Z^{-1}\in\B(\H)$ and $\tilde L = (Z\ox1) L Z^{-1}\in\B(\H)$, where $Z=\sqrt{G+1}$,
    \end{itemize}
    If $\omega,E_0\ge0$ are such that $2\Re \ip{\sqrt G\psi}{\sqrt GK\psi}+ \norm{(\sqrt G\ox1)L\psi}^2 \le \omega(\norm{\sqrt G\psi}^2+E_0\norm\psi^2$ for all $\psi\in\dom\sqrt G$, then $\{T(t)\}_{t\ge0}$ is energy-limited with stability constants $\omega,E_0$. 
\end{lemma}
\begin{proof}
    We will check condition \ref{it:main3} of \cref{thm:main}. 
    Recall from \cref{lem:base} that $W = Z(\placeholder)Z:(\dom\energy,\normiii{}_1)\to \T(\H)$ is an isometric isomorphism.
    Since $\tilde K$ and $\tilde L$ are bounded, $\tilde\L\rho = \tilde K\rho+\rho \tilde K^*+\tr_{\ell^2}\tilde L\rho \tilde L^*$ defines a bounded operator on $\T(\H)$.
    By construction, the $\normiii{}_1$-bounded operator $W^{-1}\tilde \L W$ on $\dom\energy$ is precisely $\L\restriction\dom\energy$.
    In particular, this implies $\dom(\L\upharpoonright\dom\energy)=\dom\energy$.
    The resolvents of $\L\restriction\dom\energy$ are given by $W^{-1}(\lambda-\tilde\L)^{-1}W$ and hence satisfy $\R(\lambda)\dom\energy = W(\lambda-\tilde\L)^{-1} \T(\H) = W\T(\H) = \dom\energy = \dom(\L\restriction\dom\energy)$ for sufficiently large $\lambda>0$.
    We define bounded operators $\tilde G= Z^{-1}G Z^{-1}=G/(1+G)\in\B(\H)$ and $\tilde 1 = (1+G)^{-1}=Z^{-2}$ and note that $\energy[\rho] = \tr[\tilde G W(\rho)]$ and $\tr[\rho]= \tr[\tilde G W(\rho)]$ for $\rho\in\dom\energy$.
    The operator inequality $\tilde \L^*(\tilde G) \le \omega(\tilde G+E_0\tilde1)$ follows from
    \begin{align*}
        \ip\psi{\tilde\L^*(\tilde G)\psi} 
        &= \ip\psi{(\tilde G\tilde K+\tilde K^*\tilde G+\tilde L^*(\tilde G\ox1)\tilde L)\psi}\\
        &= 2\Re \ip{\sqrt GZ^{-1}\psi}{\sqrt GKZ^{-1}\psi} + \norm{(\sqrt G\ox1)Z^{-1}\psi}^2 \\
        &\le \omega(\norm{\sqrt GZ^{-1}\psi}^2+\norm{E_0Z^{-1}\psi}^2 )
        = \ip \psi{\omega(\tilde G+E_0\tilde 1)\psi}
    \end{align*}
    for arbitrary $\psi\in\H$.
    Therefore, it holds that
    \begin{equation}
        \energy[\L\rho] = \tr[\tilde G W\L\rho] 
        = \tr[\tilde \L^*(\tilde G)W\rho]
        \le \tr[\omega(\tilde G+E_0\tilde 1)W\rho] = \omega(\energy[\rho]+E_0\tr[\rho])
    \end{equation}
    for all $0\le \rho\in\dom\energy=\dom(\L\restriction\dom\energy)$.
    Thus, \cref{thm:main} implies the claim.
\end{proof}

\begin{proof}[Proof of \cref{thm:standard}]
    \emph{Step 1.} 
    In the first step, we show that we may assume $\D=\dom K$.
    Since $K$ is $G$-bounded on a core $\D$ for $G$, it follows that $\dom G\subseteq\dom K$ and that $K$ is $G$-bounded on $\dom G$.
    Let $\D\ni\psi_n\to \psi\in \dom G$ converge in $G$-graph norm.
    Instead of working with the family of Kraus operators $L_\alpha$, let us use the operator $L= \sum_\alpha L_\alpha \ox \ket\alpha:\dom K\to\H\ox\ell^2$.
    Note the following 
    \begin{equation*}
        \sum_\alpha \norm{L_\alpha\psi}^2 = \norm{L\psi}^2,\qquad 
        \sum_\alpha\norm{\sqrt GL_\alpha \psi}^2=\norm{(\sqrt G\ox1)L\psi}^2.
    \end{equation*}
    It follows from \eqref{eq:L_summation} that $L$ is $K$-bounded and, since $K$ is $G$-bounded, $L$ is also $G$-bounded.
    Consequently $K\psi_n\to K\psi$ and $L\psi_n\to L\psi$.
    From this and lower semicontinuity of $\norm{\sqrt G(\placeholder)}$, we get
    \begin{align*}
        2\Re \ip{G\psi}{K\psi} +\norm{(\sqrt G\ox1)L\psi}^2 
        &\le  \liminf_n\,(2\Re\ip{G\psi_n}{K\psi_n} + \norm{(\sqrt G\ox1)L\psi_n}^2)\\
        &\le \liminf_n \omega(\norm{\sqrt G\psi_n}^2+E_0\norm\psi^2)
        = \omega(\norm{\sqrt G\psi}+E_0\norm\psi^2).
    \end{align*}
    In particular, since the right-hand side and $2\Re\ip{G\psi}{K\psi}$ are finite, we obtain $L\psi\in \dom(\sqrt G\ox1)$, and hence $L_\alpha\psi\in\dom \sqrt G$.
    This shows that we may assume $\D=\dom G$ without loss of generality.

    \emph{Step 2.} In this step, we use \cref{lem:bounded} to establish the claim for a regularized version of the generator.
    As in the proof of \cref{thm:dissipative}, we use the contractions $R_\eps = (1+\eps G)^{-1}$ for $\eps>0$ to define $K_\eps = R_\eps K R_\eps$. We also introduce $L_\eps = LR_\eps$.
    These make sense because our assumptions imply $R_\eps\H = \dom G \subseteq \dom K \subseteq \dom L$.
    $G$-boundedness of $K$ and $L$ implies that $K_\eps$, $L_\eps$ and $KR_\eps$ are bounded operators.
    Furthermore, it holds that
    \begin{equation}\label{eq:kodak}
        L_\eps^*L_\eps = -(K_\eps^*+K_\eps) \qandq GK_\eps + K_\eps^*G + L_\eps^*(G\ox1)L_\eps \le \omega (G+E_0),
    \end{equation}
    where $K_\eps^*G$ denotes the adjoint $(GK_\eps)^*\in\B(\H)$.
    We set $Z=\sqrt{1+G}$ and claim that $ZK_\eps $ and $(Z\ox1)L_\eps$ are bounded operators.
    Indeed, $ZK_\eps  = (ZR_\eps)(KR_\eps)$ is a product of bounded operators on $\H$, and \eqref{eq:kodak} shows
    \begin{align*}
        \norm{(Z\ox1)L_\eps\psi}^2 
        &= \norm{(\sqrt G\ox1)LR_\eps\psi}^2 + \norm{L_\eps\psi}^2 \nonumber\\
        &\le - 2\Re\underbrace{\ip{R_\eps\psi}{GKR_\eps\psi}}_{=\ip{\psi}{GK_\eps\psi}} + \omega(\norm{\sqrt GR_\eps\psi}^2 +E_0\norm{R_\eps\psi}^2) -2\Re \ip\psi{K_\eps\psi} \nonumber\\
        % &= - 2\Re\ip{\psi}{GK_\eps\psi} + \omega(\norm{\sqrt GR_\eps\psi}^2 +E_0\norm{R_\eps\psi}^2) \nonumber\\
        &\le (2\norm{(G+1)K_\eps} + \omega \norm{\sqrt GR_\eps}^2 + E_0\omega) \norm\psi^2,
    \end{align*}
    (where we used $\norm{R_\eps}\le1$).
    In particular, it follows that $ZK_\eps Z^{-1}$ and $(Z\ox1)L_\eps Z^{-1}$ are bounded. In combination with \eqref{eq:kodak} this implies that $\L_\eps\rho = K_\eps\rho +\rho K_\eps^* + \tr_{\ell^2} L_\eps \rho L_\eps^*$ generates a uniformly continuous energy-limited quantum dynamical semigroup with stability constants $\omega,E_0$.

    \emph{Step 3.} 
    We remove the regularization by taking the limit $\eps\to0$.
    As shown in the proof of \cref{thm:dissipative}, $K_\eps$ converges to $K$ strongly on $\dom G$.
    A similar argument shows that $L_\eps$ converges strongly to $L$ on $\dom G$. 
    Indeed, writing $Y = L(1+G)^{-1}\in\B(\H)$, we have $L_\eps = YR_\eps(1+G)$ which converges strongly to $Y(1+G)=L$ on $\dom G$.
    Therefore, $\L_\eps$ converges strongly to $\L$ on $(\dom K)^{\kettbra{}}$.
    The generation theorem in \cref{sec:generation} shows that $\dom G$ is a core for $K$, which by \cref{lem:davies} implies that $(\dom G)^{\kettbra{}}$ is a core for $\L$.
    Thus, the Trotter-Kato approximation theorem \cite[Thm.~II.4.8]{EngelNagel} implies that $T_\eps(t)$ converges strongly to $T(t)$.
    The claim then follows from item \ref{it:energy_lim3} of \cref{lem:energy_lim}.
\end{proof}

\section{Examples of energy-limited dynamics}\label{sec:examples}

\subsection{Gaussian channels and Markov dynamics on bosonic systems}\label{sec:gaussian}

We consider bosonic systems the number operator as the reference Hamiltonian.
Shirokov showed that Gaussian channels are energy-limited in \cite{shirokov2020extension}.
Here, we further establish the energy-limitedness of Gaussian quantum Markov dynamics.
Let us start by fixing some notation.
We set $\H=L^2(\RR^n)$, and we denote the vector of canonical operators by $R = (Q_1,\ldots,Q_n,\allowbreak P_1,\ldots,P_n)$.
We will freely let matrices act on vectors of operators and write $\frac12R^2:=\sum_j R_j^2$ for the Harmonic oscillator.
The number operator is given by $N= \frac12R^2 - \frac n2$.
If defined, the displacement $d$ and the covariance matrix $\gamma$ of a state $\rho\in\states(\H)$ are the vector $\beta\in\RR^{2n}$, and the matrix $\gamma\in M_{2n}(\RR)$ with
\begin{equation}\label{eq:moments}
    \beta_j = \tr[\rho R_j], \qquad \gamma_{jk}= 2\Re\tr[\rho(R_j-\beta_j)(R_k-\beta_k)].
\end{equation}
Positivity of the state $\rho$ requires that the covariance matrix satisfies the semi-definite constraint 
\begin{equation}\label{eq:gaussian_constraint}
    \gamma+i\sigma\ge0, \qquad \sigma = \mat{ 0&-\1_n \\ \1_n &0}.
\end{equation}
The characteristic function of a quantum state $\rho\in \states(\H)$ is the function $\chi_\rho(\alpha) = \tr[\rho D_\alpha]$, $\alpha\in\RR^{2n}$, where $D_\alpha = e^{i \alpha\trp \sigma R}$ denotes the family of Weyl (or displacement) operators.
The Wigner function $W_\rho$ of a state $\rho$ is the (symplectic) Fourier transform of the characteristic function $W_\rho(\alpha) = (2\pi)^{-\frac n2}\int_{\RR^{2n}} e^{i\alpha\trp \sigma \beta} \chi_\rho(\beta)\,d\beta$.
A state is called \emph{Gaussian} if its Wigner function is Gaussian \cite{eisert2007gaussian}:
\begin{equation}\label{eq:gauss}
    W_\rho(\alpha) = \Big(\tfrac2\pi\Big)^n (\det \gamma)^{-\frac12}e^{-(\alpha-\beta) \trp \gamma^{-1} (\alpha-\beta)}.
\end{equation}
In \eqref{eq:gauss}, the parametrization is consistent with the definition of the covariance matrix and displacement vector above.
In particular, every vector $\beta\in\RR^{2n}$ and every symmetric real matrix $\gamma=\gamma\trp\in M_{2n}(\RR)$ such that \eqref{eq:gaussian_constraint} holds, determines a Gaussian state $\rho_{\gamma,\beta}$, e.g., $\rho_{\1,\beta} = \kettbra \beta$.

The energy expectation value of Gaussian states with respect to the number operator can be calculated from its covariance and displacement:
\begin{align}
    \energy[\rho_{\gamma,\beta}]
    % = \frac12\sum_j \tr[\rho R_j^2] - \frac n2
    % &= \frac12 \sum_j \Big(\tr[\rho (R_j-\beta_j)^2] - \beta_j^2 + 2\beta_j \tr[\rho R_j]-d\Big)\nonumber\\
    % &= \frac12 \sum_j ( \tfrac12\gamma_{jj} + \beta_j^2 -n)
    = \frac14\tr \gamma + \frac12 \beta^2 - \frac n2
\end{align}
A quantum channel $T:\T(\H_A)\to\T(\H_B)$ between two bosonic systems with $\H_j = L^2(\RR^{n_j})$ is said to be Gaussian if it takes Gaussian states to Gaussian states.
A Gaussian channel necessarily transforms the covariance matrix and displacement vector linearly:
If $T$ is Gaussian there exist a linear map $X:\RR^{2n_B}\to\RR^{2n_A}$, $Y=Y\trp\in M_{2n_B}(\RR)$ and $\alpha\in \RR^{2n_B}$ such that the covariance matrix $\gamma'$ and the displacement $\beta'$ of $\rho'=T\rho_{\gamma,\beta}$ are given by
\begin{equation}
    \gamma' = X\trp \gamma X+Y, \qquad \beta' = X\trp \beta+\alpha.
\end{equation}
Complete positivity of $T$ enforces the positivity condition
\begin{equation}\label{eq:gaussian_cp}
    Y + i\sigma_B -iX\trp \sigma_A X \ge0,
\end{equation}
and, conversely, any triple $(X,Y,\alpha)$ that satisfies \eqref{eq:gaussian_cp} determines a Gaussian channel, denoted $T_{X,Y,\alpha}$, in this way \cite{eisert2007gaussian}.
The matrix $X$ describes the linear transformation on phase space which the channel implements. The matrix $Y$ is the noise introduced by the channel, and $\alpha$ is an additional displacement.
By \eqref{eq:gaussian_cp}, every linear map $X$ may be implemented by a Gaussian quantum channel with sufficient noise.
We can factor every Gaussian channel into pure displacement and a nondisplacing channel: 
\begin{equation}
     T_{X,Y,\alpha} = T_{0,0,\alpha} T_{X,Y,0}.
\end{equation}
Pure displacement channels are implemented by Weyl operators $T_{0,0,\alpha} = D_\alpha(\placeholder)D_\alpha^*$.
To understand the energy change caused by a Gaussian channel, one needs to estimate the action of the dual channel $T_{X,Y,\alpha}^*$ on the number operator.
This is readily derived from the formula
\begin{equation}
    T_{X,Y,\alpha}^*(R\trp\! AR) = (XR+\alpha)\trp  A(XR+\alpha) + \frac12\tr[AY],
\end{equation}  
valid for any symmetric matrix $A=A\trp\in M_{2n}(\RR)$.
As observed by Shirokov, applying this to the number operator $\frac12 R\trp R - \frac n2$ immediately gives:

\begin{lemma}[{\cite[Sec.~5]{shirokov2020extension}}]\label{thm:gaussian_channel}
    Gaussian channels are energy-limited with respect to the number operators on the input and output systems.
    The output energy of a nondisplacing Gaussian channel is bounded as
    \begin{equation}\label{eq:fT_gaussian}
        f_{T_{X,Y,0}}(E) \le \norm X_\oo^2 E+ \frac14 \tr[Y]+\norm X_\oo^2\frac{n_A}2-\frac{n_B}2,
    \end{equation}
    where $n_A$ and $n_B$ are the number of input and output modes, respectively.
\end{lemma}

We now consider Gaussian quantum Markov dynamics. For simplicity, we restrict to the case without displacement.
We use the following structure theorem from {\cite[Sec.~5.1]{lars_thesis}}:

\begin{proposition}\label{thm:gaussian_qds}
    Let $\{T(t)\}_{t\ge0}$ be a quantum dynamical semigroup.
    If each $T(t)$ is a Gaussian quantum channel without displacement, there exist matrices $\dot X,\dot Y=\dot Y\trp\!\in M_{2n}(\RR)$ satisfying $\dot Y+\frac i2(\dot X\trp\sigma+\sigma\dot X)\ge0$, such that
    \begin{equation}
        T(t) = T_{X(t),Y(t),0},\qquad X(t)=e^{t\dot X}, \qquad Y(t)=\int_0^t X(s)\trp\dot YX(s)\,ds.
    \end{equation}
    The generator of such a semigroup is standard and given by
    \begin{equation}\label{eq:gaussian_generator}
        \L\rho = \frac12 \sum_{jk} \Big( m_{jk} \big(R_j [\rho,R_k] +[R_j,\rho]R_k \big) + h_{jk} [R_j R_k,\rho] \Big),
    \end{equation}
    with matrices $0\le m\in M_{2n}(\CC)$, $h=h\trp\in M_{2n}(\RR)$, given by $m = \sigma \dot Y \sigma + \frac i2(\sigma \dot X+\dot X\trp\sigma)$ and $h=\tfrac12 (\sigma \dot X\trp-\dot X\sigma)$.
    Furthermore, every quantum dynamical semigroup of Gaussian channels arises this way.
\end{proposition}

% To properly define the generator $\L$ in \eqref{eq:gaussian_generator} as a standard generator in the sense of \cite{siemon_unbounded_2017}, we need to specify a maximally dissipative operator $K$ and operators $L_j$ with $K+K^*+\sum_j L_j^*L_j$ such that $\L$ is the standard generator formally given by $\L\rho = K\rho+\rho K^*+\sum_j L_j^*\rho L_j$.
% This is achieved by defining $K=-\frac12R\trp(ih-m)R$ and $L_j = (\sqrt mR)_j$.

\begin{proposition}\label{thm:gaussian_el}
    Let $\{T(t)\}_{t\ge0}$ be a Gaussian quantum dynamical semigroup and let $\dot X$, $\dot Y$ be the matrices from \cref{thm:gaussian_qds}.
    Then $\{T(t)\}_{t\ge0}$ is energy-limited with stability constants $\omega=2 \norm{\dot X}_\oo$, $E_0=\frac n2+\norm{\dot Y}_\oo/8\norm{\dot X}_\oo$.
\end{proposition}

\begin{proof}
    Since $T(t)=T_{X(t),Y(t),0}$ with the notation from \cref{thm:gaussian_qds}, equation \eqref{eq:fT_gaussian} implies
    \begin{align*}
        f_{T(t)}(E)&\le \norm{X(t)}_\oo^2 (E +\tfrac n2)+ \tfrac14\tr[Y(t)]- \tfrac n2\\ 
        &\le e^{2t\norm{\dot X}_\oo} (E+\tfrac n2) + \tfrac14  \int_0^t \norm{X(s)\dot YX(s)}_\oo\,ds -\tfrac n2\\
        &\le e^{2t\norm{\dot X}_\oo} (E +\tfrac n2) + \tfrac{\norm{\dot Y}_\oo}4  \int_0^t e^{2s\norm{\dot X}_\oo} \,ds -\tfrac n2\\
        &\le e^{2t\norm{\dot X}_\oo} E + (e^{2t\norm{\dot X}_\oo}-1)(\tfrac n2 +\tfrac{\norm{\dot Y}_\oo}{8\norm{\dot X}_\oo}).\qedhere
    \end{align*}
\end{proof}

\subsection{Coherent state quantization}\label{sec:quantization}

We consider the Hilbert space $\H=\K\ox L^2(\RR^n) =  L^2(\RR^n;\K)$ of a quantum system with $n$ canonical degrees of freedom coupled to a system with Hilbert space $\K$.
We continue to use the notation from \cref{sec:gaussian} for operators on $L^2(\RR^n)$.
We denote by $\ket 0\in L^2(\RR^n)$ the ground state of the number operator $N= \sum_{i=1}^n a_i^\dagger a_i$ and
by $\ket \alpha= D_\alpha\ket0$, $\alpha\in\RR^{2n}$, the family of coherent states.
As the reference Hamiltonian, we take $G=1\ox N$.
The coherent state quantization of an hermitian operator-valued function $h \in L^\oo(\RR^{2n};\B(\K))$ is the operator
\begin{equation}\label{eq:coherent_state_quant}
    H = (2\pi)^{-n}\!\int_{\RR^{2n}} h(\alpha) \ox \kettbra \alpha\, d\alpha\in \B(L^2(\RR^n;\K).
\end{equation}
The map $h\mapsto H$ defines a normal unital completely positive map $L^\oo(\RR^{2n})\to\B(\H)$, where unitality follows from the overcompleteness relation $\int_{\RR^{2n}} \kettbra \alpha\, d\alpha = (2\pi)^{n}1$.
It also makes sense to consider the coherent state quantization of unbounded functions $h$. 
If $h$ is a measurable and polynomially bounded then $H$ is naturally defined as an operator on the domain of vector-valued Schwartz functions $\Sch(\RR^n,\K)$ \cite{ESA}.

\begin{proposition}\label{prop:coherent_state_quant}
    Let $h:\RR^{2n}\to\B(\K)$ be a hermitian operator-valued $C^1$-function whose gradient $\nabla h\in C(\RR^{2n},\RR^{2n}\times \B(\K))$ is globally Lipschitz continuous.
    Then:
    \begin{enumerate}[(1)]
        \item 
            The coherent state quantization $H$ of $h$ is essentially self-adjoint.
        \item 
            The unitary dynamics generated by $\bar H$ is energy-limited with respect to $G=1\ox N$.% and with respect to $G^2$.
        \item 
            Let $\omega,E_0>0$ be such that $\| \alpha\trp \sigma \nabla h(\alpha)\|\le \omega(\tfrac12\abs{\alpha}^2+E_0-n)$, then $\omega,E_0$ are stability constants for the unitary dynamics generated by $\bar H$. 
    \end{enumerate}
\end{proposition}

In particular, the gradient $\nabla h$ is Lipschitz continuous if $h\in C^2$ with uniformly bounded second derivatives.

\begin{proof}
    The coherent state quantization $H$ is unitarily equivalent to the Berezin-Toeplitz operator $T_h$ with operator-valued symbol $h:\CC^n\equiv\RR^{2n}\to\B(\K)$ on the vector-valued Segal-Bargmann space, a certain $L^2$-space of complex analytic functions on $\CC^n$, via the Bargman transform \cite{bargmann,folland} (see \cite{ESA} for the vector-valued case).
    The main Theorem of \cite{ESA} states that our assumptions imply essential self-adjointness of the Berezin-Toeplitz operator with symbol $h$ on the domain that corresponds to the Schwartz functions under the Bargmann transform.
    This is proved by checking the assumptions of Nelson's commutator theorem, which, by \cref{thm:unitary2}, also imply energy-limitedness.
    % The proof of this result goes by establishing $G$-boundedness and then checking the assumptions of Nelson's commutator theorem, which by \cref{thm:unitary2} implies that the generated dynamics is indeed energy-limited.
    The stability constants are obtained from the requirement $\pm i [H,G]\le \omega(G+E_0)$.
    It is shown in \cite{ESA} that the commutator $-i[H,G]$ is equal to the coherent state quantization of the symbol $\partial_\theta h(\alpha) :=\alpha\trp \sigma \nabla h(\alpha)$. 
    Since the coherent state quantization is monotone and takes $\frac12\abs\placeholder^2$ to $G+n$, the assumptions imply $\pm i[H,G]\le \omega((G+n)+E_0-n) = \omega(G+E_0)$ \cite{ESA}.
    % It remains to show that the dynamics of $H$ are also energy-limited with respect to $G^2$.
    % By \cref{lem:frohlich}, it suffices to show that the commutator $[H,G]$ is $G$-bounded.
    % In the Segal-Bargmann representation, the commutator is the Berezin-Toeplitz operator $T_f$ with symbol $f(\alpha)=\alpha\trp\sigma\nabla h$. 
    % As in \cite{ESA}, the inequality $(T_f)^2\le T_{f^2}$, which holds in the sense of quadratic forms on the Schwartz space, implies $T_f^2 \le c_1 T_{(\abs\placeholder^2+1)^2} \le c_2 ( T_{\abs\placeholder^2})^2 \le c_3 (G^2+1)$ for suitable constants $c_1,c_2,c_3>0$.
\end{proof}

Similar results can be shown for the Weyl quantization at the price of additional regularity assumptions on $h$.
For instance, a similar proof applies to the Weyl quantization if $h\in C^{2d+3}$ with uniformly bounded derivatives of second and higher order (see \cite{ESA_Weyl}).
Using the generation theorem in \cref{sec:generation}, one can also cover contraction semigroups generated by coherent state quantizations of dissipative operator-valued functions $h:\RR^{2n}\to\B(\K)$.

Let us consider consider a single mode, i.e., $n=1$, system coupled to a qubit $\K=\CC^2$.
If we take the function
\begin{equation}
    h(q,p) = \Omega\bigg(\frac{q^2+p^2}2-1\bigg) +\sqrt2 g q \sigma_x  + \nu \sigma_z 
\end{equation}
for constants $\Omega>0$, $\nu,g\in\RR$, the coherent state quantization yields the quantum Rabi Hamiltonian  
\begin{align}
    H =\Omega a^\dagger a + g \sigma_x(a+a^\dagger) +  \nu\sigma_z,
\end{align}
where we suppressed the tensor product symbol.
Therefore, the quantum Rabi model is energy-limited.
The same is true for all Hamiltonians with interaction linear in $Q,P$ or $a$ and $a^\dagger$.

\subsection{Quantum birth process}\label{sec:birth}

In this section, we consider a class of standard quantum dynamical semigroups introduced in \cite{siemon_unbounded_2017}.
What is interesting about this class is that it contains nonconservative dynamics even though the infinitesimal conservativity condition $K^*+K=L^*L$ holds.
In the nonconservative case, one can perturb the generators to make them actually conservative, and it was proved in \cite{siemon_unbounded_2017} that this results in a nonstandard generator.

Following \cite{siemon_unbounded_2017}, we consider the Hilbert space $\H=\ell^2(\NN_0)$.
We denote by $\ket n$, $n\in\NN_0$, the canonical basis and set $\psi_n := \braket n\psi$ for $\psi\in\ell^2(\NN_0)$.
Let $\mu_0,\mu_1,\ldots>0$ be a sequence of positive numbers.
To define the process, we introduce operators
\begin{align}
    &&&&K\ket n &= -\frac{\mu_n}2\ket n, &&& \dom K &= \{\psi\in\ell^2(\NN_0) : \sum_n \mu_n^2 \abs{\psi_n}^2 <\oo\},&& \\
    &&&&L\ket n&=\sqrt{\mu_n}\ket{n+1},&&& \dom L &= \{\psi\in\ell^2(\NN_0) : \sum_n \mu_n \abs{\psi_{n}}^2<\oo\}.
\end{align}
Let us now consider the standard generator $\L$ determined by $K$ and $L$ (see \cref{sec:standard} or \cite{inken_thesis}).
Heuristically speaking, the dynamics generated by $\L$ may be described as follows:
The states $\ket n$ may transition to states $\ket{n+1}$, and the probability of this is distributed exponentially with parameter $\mu_n$. Thus, the transition $\ket n \to \ket{n+1}$ takes a time of $\mu_n^{-1}$ on average.
Therefore, the expected time for the first state $\ket0$ to escape to infinity is
\begin{equation}\label{eq:tau}
    \tau:=\sum_{n=0}^\oo\frac 1{\mu_n} \in (0,\oo].
\end{equation}
If $\mu_n\to\oo$, transitions happen at faster and faster rates and if they grow sufficiently fast, e.g., $\mu_n=n^2$, we have $\tau<\oo$ meaning that particles escape to infinity in finite time.
While the relation $K^*+K+L^*L=0$ guarantees that $\L$ is infinitesimally conservative on the ketbra domain, i.e., $\tr\L\rho=0$ for $\rho\in(\dom K)^{\kettbra{}}$, it does not imply that $\L$ is infinitesimally conservative on its full domain $\dom\L$ (see \cref{lem:nonstandard} below).%
\footnote{Escape in finite time is not special to quantum systems. E.g., it occurs in the classical birth process \cite{siemon_unbounded_2017}.}
Indeed, the dynamics of $\L$ is conservative if and only if $\tau=\oo$.

\begin{lemma}[{\cite{siemon_unbounded_2017,inken_thesis}}]\label{lem:nonstandard}
    The dynamics generated by $\L$ is conservative if and only if $\tau=\oo$.
    If $\tau<\oo$, then $\sigma = \frac1\tau \sum_n \frac1{\mu_n} \kettbra n\in\dom\L$ and $\tr\L\sigma = -\frac1\tau<0$.
\end{lemma}

What is interesting about the quantum birth process is that it can be used to construct a nonstandard generator by perturbing $\L$ to restore conservativity:

\begin{lemma}[{\cite{siemon_unbounded_2017,inken_thesis}}]\label{lem:nonstandard gen}
    Assume $\tau<\oo$.
    If $\chi\in\states(\H)$ is a density operator and $\L'$ is defined as 
    \begin{equation}
        \L'\rho := \L\rho -\tr[\L\rho]\,\chi, \qquad \rho\in\dom\L':=\dom\L.
    \end{equation}
    Then, $\L'$ is a nonstandard generator of a conservative quantum dynamical semigroup.
\end{lemma}

The natural reference Hamiltonian in this setting is of the form
\begin{align}\label{eq:reference_qjump}
    G= \sum_n \epsilon_n \,\kettbra n, \qquad \dom G = \Big\{\psi\in\ell^2(\NN_0) : \sum \epsilon_n^2\abs{\psi_n}^2<\oo\Big\}
\end{align}
with eigenvalues $\epsilon_0=0\le \epsilon_1\le\epsilon_2\le \cdots$. 
Since energy-limitedness with respect to a bounded reference is trivial, we assume $\lim_n\epsilon_n=\oo$.% Note that $G$ has compact resolvents.
By \cref{thm:main}, energy-limitedness requires that 
\begin{equation}
    K^*G+GK + L^*GL = \sum_{n=0}^\oo \mu_n (\epsilon_{n+1}-\epsilon_n) \kettbra n.
\end{equation}
is dominated by $\omega(G+E_0)$ for some $\omega,E_0\ge0$, which is equivalent to
\begin{equation}\label{eq:birth_EL}
    \mu_n(\epsilon_{n+1}-\epsilon_n) \le \omega(\epsilon_n +E_0), \qquad n\in\NN_0,
\end{equation}
If escape to infinity and energy-limitedness could be true simultaneously, then \cref{lem:nonstandard gen} with, say, $\chi= \kettbra0$, would yield an energy-limited dynamical semigroup with a nonstandard generator(see \cref{sec:problems} for further discussion).
This is not the case:

\begin{proposition}\label{prop:escape}
    The following are equivalent:
    \begin{enumerate}[(a)]
        \item\label{it:escape1} Conservativity of the dynamics or, equivalently, no escape in finite time: $\tau=\oo$.
        \item\label{it:escape2} There exists an increasing sequence $(\epsilon_n)$ with $\lim_n\epsilon_n=\oo$ such that \eqref{eq:birth_EL} holds for some constants $\omega,E_0\ge0$.
    \end{enumerate}
\end{proposition}

% This shows that, if escape to infinity is possible in finite time, then there exists no reference Hamiltonian of the form \eqref{eq:reference_qjump} (with $\lim_n \epsilon_n=\oo$) with respect to which the dynamics is energy-limited.

\begin{proof}
    \ref{it:escape1} $\Rightarrow$ \ref{it:escape2}:
    We define the sequence recursively via $\epsilon_{n+1}= (1+\frac{1}{\mu_n})\epsilon_n$ for $n\ge1$ and $\epsilon_0=0$, $\epsilon_1=1$. 
    The sequence diverges since
    \begin{equation}
        \lim_n \epsilon_n = \prod_{n=1}^\oo \Big(1+\frac1{\mu_n}\Big) \ge 1+ \sum_{n=1}^\oo \frac1{\mu_n} = 1+ \tau- \frac1{\mu_0} =\oo,
    \end{equation}
    where we used \cref{lem:nonstandard}.
    We have $\mu_n(\epsilon_{n+1}-\epsilon_n) = \mu_0 \delta_{0,n} + \epsilon_n \le \omega(\epsilon_n+E_0)$ with $\omega=1$ and $E_0=1/\mu_0$.

    \ref{it:escape2} $\Rightarrow$ \ref{it:escape1}:
    For simplicity, we assume $\epsilon_1>0$. The general case follows similarly.
    By appropriate choice of the offset $E_0>0$, we see that there exists $\omega>0$ such that $ \mu_n(\epsilon_{n+1}-\epsilon_n) \le \omega \epsilon_n$ for all $n\in\NN$ (excluding $n=0$).
    By rescaling the $\mu_n$ with a constant, we can further assume $\omega=1$.
    We can now rearrange the resulting inequality $\mu_n(\epsilon_{n+1}-\epsilon_n)\le \epsilon_n$ to give $\frac{\eps_{n+1}}{\eps_n} \le 1+\frac1{\mu_n}$.
    Then 
    \begin{equation*}
        e^{\tau} = \prod_{n=0}^\oo e^{1/\mu_n} \ge \prod_{n=0}^\oo \Big(1+\frac1{\mu_n}\Big) \ge \prod_{n=1}^\oo \Big(1+\frac1{\mu_n}\Big) \ge \prod_{n=1}^\oo \frac{\epsilon_{n+1}}{\epsilon_n} = \epsilon_1 \cdot \lim_{n\to\oo} \epsilon_n =\oo
        % \lim_{n\to\oo} \epsilon_n =  \prod_{n=1}^{\oo} \frac{\epsilon_{n+1}}{\epsilon_n} \le \prod_{n=1}^\oo \Big(1+\frac1{\mu_n}\Big) \le \prod_{n=1}^\oo e^{\frac1{\mu_n}} \le \epsilon_1
    \end{equation*}
    (the limits make sense because $(\epsilon_n)$ is an increasing sequence, and the infinite products make sense because each factor is $\ge 1$).
    Therefore, $\lim_n \epsilon_n= \oo$ implies $\tau=\oo$ and $\tau<\oo$ implies $\lim_n\epsilon_n<\oo$
\end{proof}

\subsection{Representations of Lie groups}\label{sec:groups}
\renewcommand\grp{\mathcal G}

In this subsection, we show that every unitary representation of a connected Lie group is energy-limited relative to a natural reference Hamiltonian, the Nelson Laplacian.
The results presented here build on \cite{lie_group_error}, where state-dependent quantum speed limits for Lie group representations (cp.\ \cref{sec:speed}).

Let $\grp$ be a Lie group with Lie algebra $\lie$ and let $U:\grp\ni g\mapsto U_g\in\U(\H)$ be a continuous unitary representation on a Hilbert space $\H$, where the unitary group $\U(\H)$ is equipped with the strong operator topology.
We equip the Lie algebra with an inner product $\ip\placeholder\placeholder_\lie$, and we denote the induced norm by $\norm X_\lie=\sqrt{\ip XX_\lie}$.
We may freely choose the inner product. 
Typically, $\lie$ is a Lie algebra of skew-symmetric real (or skew-hermitian complex) matrices and it makes sense to pick the Frobenius inner product $\ip XY_{\lie} = \tr X\trp Y$ (or $\tr X^*Y$).
Let us denote by $A:\lie\ni X\mapsto A(X)$ the induced Lie algebra representation in terms of self-adjoint operators, which is uniquely defined by
\begin{align}
    &&U_{e^{tX}}&=e^{-itA(X)},&&t\in\RR,\ X\in\lie.\\
\intertext{The dense subspace $C^\oo(U)$ of $U$-smooth vectors{\footnotemark}  
is invariant, i.e., $U_gC^\oo(U)=C^\oo(U)$ for all $g\in\grp$. Furthermore, $C^\oo(U)$ is an invariant core for all $A(X)$, $X\in\mf g$, on which the commutator relations }
    &&[A(X), A(Y)]&= i A([X,Y]),&&X,Y\in\lie,\\
    &&U_gA(X) U_g^* &= A(\Ad_gX),&& X\in\lie,\ g\in\grp,
\end{align}
hold; see \cite{nelson1959} for details.
\footnotetext{A vector $\psi\in\H$ is smooth with respect to the continuous representation $U$ of $\grp$ on $\H$ if $\grp\ni g\mapsto U_g\in\H$ is smooth with respect to the strong operator topology.}
The natural reference Hamiltonian is the Nelson Laplacian \cite{lie_group_error,nelson1959}.
To define it, we pick an orthonormal basis $\{X_i\}\subset \lie$ and set
\begin{equation}\label{eq:NL}
    \Delta= \sum A(X_i)^2.
\end{equation}
This expression makes sense on the dense subspace $C^\oo(U)$ of $U$-smooth vectors and defines an essentially self-adjoint operator \cite{nelson1959}.
The Nelson Laplacian only depends on the choice of inner product but not on the chosen basis.
It is shown in \cite{lie_group_error} that any other inner product $\ip\placeholder\placeholder_{\mf g}'$ yields an equivalent Nelson Laplacian
\begin{equation}\label{eq:Delta'}
    c \Delta' \le \Delta \le C\Delta',
\end{equation}
% or, equivalently, $C^{-1}\Delta\le \Delta'\le c^{-1}\Delta$,
where $c,C>0$ are constants such that $c\ip XX_{\mf g}\le\ip XX_{\mf g}'\le C \ip XX_{\mf g}$ for all $X\in \mf g$.

\begin{lemma}\label{lem:Delta}
    Let $\alpha\in[0,1]$. The following estimates hold
    \begin{equation}
        \norm{\Ad_{g}}_{op}^{-2}\Delta \le  U_g^*\Delta U_g  \le \norm{\Ad_{g^{-1}}}_{op}^{2} \Delta,
    \end{equation}
    where $\norm{}_{op}$ denotes the operator norm with respect to $\norm{}_\lie$.
\end{lemma}
\begin{proof}
    $\ip XY'_{\mf g} := \ip{\Ad_g X}{\Ad_g Y}_{\mf g}$ defines an inner product whose corresponding Nelson Laplacian is $\Delta'=U_g^*\Delta U_g$.
    % The optimal constants $c,C>0$ so that $c\ip XX_{\mf g}\le\ip{\Ad_g X}{\Ad_g X}_{\mf g}\le C \ip XX_{\mf g}$ for all $X\in\mf g$ are $c=\norm{\Ad_{g^{-1}}}_{op}^{-2}$ and $C=\norm{\Ad_g}_{op}^2$. 
    Thus, the claim follows from \eqref{eq:Delta'}.
\end{proof}

It follows that the whole group representation is energy-limited:

\begin{proposition}\label{prop:EL_grp}
    Let $G=\Delta^\alpha-E_0^\alpha$ be the system's reference Hamiltonian, where $E_0 =\inf \Sp\Delta$ and $0\le \alpha\le1$.
    Then 
    \begin{equation}
        f_{U_g}(E) \le \norm{\Ad_{g^{-1}}}_{\op}^{2\alpha}(E+E_0^\alpha)-E_0^\alpha
    \end{equation}
    In particular, if $X\in\lie$, the unitary group $\{e^{-itA(X)}\}_{t\in\RR}$ is energy-limited with stability constants $2\alpha \norm{\mathrm{ad}_X}_{op}$, $E_0^\alpha$, i.e.,
    \begin{equation}
        f_{e^{-itA(X)}}(E) \le e^{2\abs t\alpha\norm{\ad_x}_{\op}} (E+E_0^\alpha)-E_0^\alpha.
    \end{equation}
\end{proposition}
\begin{proof}
    The first claim is straightforward from \cref{lem:Delta,thm:dual_prob_Enorm}.
    The second claim follows from the first one and the estimate $\norm{\Ad_{e^{tX}}}_{\op}=\norm{e^{t\ad_X}}_{\op}\le e^{\abs t\norm{\ad_X}_{\op}}$.
\end{proof}

\begin{example}[Metaplectic representation]
    Let $G=\textrm{Mp}(2n,\RR)$ be the metaplectic group, i.e., the two-fold cover of the symplectic group $\mathrm{Sp}(2n,\RR)$.
    The metaplectic group has a natural continuous representation $U$ on $\H=L^2(\RR^n)$ such that
    \begin{equation}
        U_{e^X} = e^{-\frac i2 R\trp \sigma X R}, \qquad X\in\lie=\mathfrak{sp}(2m,\RR),
    \end{equation}
    where $\sigma$ denotes the symplectic matrix and $R$ is the vector of canonical operators (see \cref{sec:gaussian}).
    In \cite{lie_group_error} the Nelson Laplacian of this representation is shown to be the squared Harmonic oscillator (plus a constant).
    Therefore, the metaplectic group is energy-limited with respect to $G=N^2$ and with respect to $G=N$, where $N$ denotes the number operator.
\end{example}

We can explicitly estimate the ECO norm of the infinitesimal generators $A(X)$:

\begin{lemma}\label{lem:enorm_grp}
    Let the reference Hamiltonian be the grounded Nelson Laplacian $G=\Delta-E_0$, then the ECO norm of $A(X)$, $X\in\lie$, is given by
    \begin{equation}
        \norm{A(X)}_{\op,E} \le \norm{X}_\lie \sqrt{E+E_0}.
    \end{equation}
\end{lemma}

\begin{proof}
    It is proved in \cite[Lem.~4]{lie_group_error} that  $A(X)^2\le \norm{X}_\lie^2 \Delta$, which by \cref{thm:dual_prob_Enorm} implies the claim.
\end{proof}

\section{Applications}\label{sec:applications}

In this section, we will show that the combination of energy-limited dynamics, energy-constrained norms and the submultiplicativity estimates allows one to prove state-dependent continuity bounds in infinite-dimensions by paralleling arguments from the finite-dimensional case.

\subsection{Quantum speed limits}\label{sec:speed}

Here, we present a simple application of the submultiplicativity estimate from \cref{thm:submultiplicativity} to quantum speed limits.
Let us start with the case of unitary dynamics:

\begin{proposition}\label{prop:qsl1}
    Let $H_1$ and $H_2$ be self-adjoint operators generating energy-limited unitary groups $U_1(t)$ and $U_2(t)$, respectively. 
    Let $\D$ be a $U_2(t)$-invariant core for $\sqrt G$ with $\D\subset \dom H_1,\dom H_2$.
    Let $\omega,E_0$ be stability constants for $U_2(t)$.
    Then, for a state vector $\psi\in\H$ with energy $E=\energy[\psi]$, we have 
    \begin{align}\label{eq:qsl1}
        \norm{U_1(t)\psi-U_2(t)\psi} 
        &\le \abs t\,\norm{H_1-H_2}_{\op,f_t(E)}, 
        % &\le \Big(\abs t + \tfrac{e^{\omega\abs t} -1-\omega \abs t}{2\omega} (1+\tfrac{E_0}E)\Big) \, \norm{H_1-H_2}_{\op,E}\nonumber\\
        % &= \big(\abs t+\Order(\abs t^2)\big) \norm{H_1-H_2}_{\op,E}.
    \end{align}
    where $f_t(E)=E+(e^{\omega\abs t}-1)(E+E_0)$.
\end{proposition}

The ECO norm appearing in \eqref{eq:qsl1} is defined as in item \ref{it:norm_dense} of \cref{lem:ECD} by optimizing the distance over energy-constrained state vectors in $\D$. It is finite if and only if $H_1-H_2$ is $\sqrt G$-bounded on $\D$.
Thus, if one wants to apply this to, say, quadratic bosonic Hamiltonians, the reference Hamiltonian needs to be something like the squared number operator.
\Cref{prop:qsl1} follows directly from the following Lemma:
\begin{lemma}\label{lem:qsl_lemma}
    Under the assumption of \cref{prop:qsl1}, it holds that 
    \begin{equation}
        \norm{U_1(t)-U_2(t)}_{\op,E} \le \int_0^{\abs t} \norm{H_1-H_2}_{\op,f_s(E)}\,ds.
    \end{equation}
\end{lemma}
\begin{proof}
    Without loss of generality we assume $t>0$.
    Let $\psi,\phi\in\D$ be unit vectors and put $E=\energy[\psi]$.
    By assumption, $U_1(t)\phi$ and $U_2(t)\psi$ are differentiable in $t$. We find
    \begin{align*}
        \abs{\ip\phi{(U_1(t)-U_2(t))\psi}}
        &=\bigg|\int_0^t \frac{d}{ds} \ip{U_1(s-t)\phi}{U_2(s)\psi} \,ds\bigg|\\
        &=\bigg|\int_0^t \big(\ip{H_1U_1(s-t)\phi}{U_2(s)\psi}-\ip{U_1(s-t)\phi}{H_2U_2(s)\psi}\big)\,ds\bigg| \\
        &=\bigg|\int_0^t \ip\phi{U_1(t-s)(H_1-H_2)U_2(s)\psi}\,ds\bigg| \\
        &\le \int_0^t \norm{(H_1-H_2)U_2(s)\psi}\,ds \le \int_0^{t} \norm{(H_1-H_2)}_{\op,f_s(E)}\, ds
    \end{align*}
    where we used \cref{thm:submultiplicativity}.
    By \eqref{eq:Enorm'} this gives the desired bound on the ECO norm.
\end{proof}

If the generators come from a continuous Lie group representation, we can use the Nelson Laplacian (see \cref{sec:groups}) and \cref{prop:EL_grp} to make the quantum speed limit explicit:

\begin{corollary}\label{cor:qsl_grp}
    Let $g\mapsto U_g$ be a continuous representation of a connected Lie group $\grp$ and let $X\mapsto A(X)$ be the induced Lie algebra representation.
    Pick some inner product on the Lie algebra $\lie$ and let $\Delta$ be the corresponding Nelson Laplacian (see \cref{sec:groups}).
    Then
    \begin{equation}\label{eq:qsl_grp}
        \norm{e^{-iA(X)}\psi - e^{-iA(Y)}\psi} \le  
        \tfrac{e^{\omega}-1}\omega \norm{X-Y}_\lie \norm{\sqrt\Delta\psi}, \qquad \psi\in\dom\sqrt\Delta,
    \end{equation}
    where $\omega=\min\{\norm{\ad_X}_{op},\norm{\ad_Y}_{op}\}$ and $\tfrac{e^\omega-1}\omega=:1$ if $\omega=0$.
\end{corollary}

The operator norm of $\ad_X = [X,\placeholder]$ is taken relative to the chosen inner product on $\lie$.

\begin{proof}
    Let $E_0$ be the ground state energy of $\Delta$ and take $G=\Delta-E_0$ as the reference Hamiltonian. Set $f_t(E)=E+(e^{\omega t}-1)(E+E_0)$.
    By \cref{prop:EL_grp,lem:enorm_grp,lem:qsl_lemma}, we have
    \begin{align*}
        \norm{e^{-iA(X)}-e^{-iA(Y)}}_{\op,E} \le \int_0^1 \norm{A(X)-A(Y)}_{\op,f_s(E)}\,ds \le \int_0^1 \norm{X-Y}_\lie \sqrt{e^{2s\omega}(E+E_0)}\,ds.
    \end{align*}
    The right hand side equals $\norm{X-Y}_\lie \sqrt{E+E_0}$ times $\int_0^1 e^{s\omega}ds = \frac{e^{\omega}-1}\omega$.
    The claim follows because $\norm{\sqrt{\Delta}\psi}= \sqrt{\energy[\psi]+E_0}$ for unit vectors $\psi\in\dom\sqrt G$.
\end{proof}

The case $Y=0$ in \cref{cor:qsl_grp} yields $\norm{e^{-iA(X)}\psi-\psi}\le \norm X_\lie \norm{\sqrt\Delta\psi}$ which is precisely the estimate used in \cite{lie_group_error} to derive the bound $\norm{U_g\psi-U_h\psi}\le d(g,h)\norm{\sqrt\Delta\psi}$ for general group elements $g,h\in\grp$, where $d$ is a left-invariant metric on $\grp$.
However, the metric $d$ is rather hard to estimate and, in applications, one relies on the upper bound $d(g,h)\le \norm{\log(g^{-1}h)}_\lie$ \cite{lie_group_error}, which requires one to find a logarithm of $g^{-1}h$. 
The estimate \eqref{eq:qsl_grp}, which involves only infinitesimal objects, seems better suited for treating quantum speed limits with Hamiltonians coming from a Lie algebra representation.

A similar technique works for open systems and gives:
\begin{proposition}\label{prop:qsl2}
    Let $\L_1$ and $\L_2$ be generators of energy-limited dynamical semigroups $\{T_i(t)\}_{t\ge0}$. Let $\omega,E_0$ be stability constants for $\{T_2(t)\}_{t\ge0}$ and set $f_t(E)=E+(e^{\omega t}-1)(E+E_0)$.
    Let $\D\subset \dom\L_1\cap \dom\L_2$ be a $T_2(t)$-invariant $\normiii\placeholder_1$-dense subspace of $\dom\energy$.
    Then
    \begin{align}\label{eq:qsl2}
        \norm{T_1(t)-T_2(t)}_{\diamond,E} 
        \le t\norm{\L_1-\L_2}_{\diamond,f_t(E)}.
        % &\le \underbrace{\big( t + \tfrac{e^{\omega t}-1}\omega (1+\tfrac{E_0}E) \big)}_{=t+\Order(t^2)} \,\norm{\L_1-\L_2}_{\diamond,E}.
    \end{align}
\end{proposition}

\begin{proof}
    Let $\rho\in\states\cap\D$ and let $A\in\dom\L_1^*$ ($\L_1^*$ is the generator of the dual semigroup $T^*(t)$ which is strongly continuous for the $\sigma$-weak operator topology).
    Since $t\mapsto T_1^*(t)(A)$ is $C^1$ for the $\sigma$-weak operator topology and $t\mapsto T_2(t)\rho$ is $C^1$ for the trace norm topology, we know that $(t,s)\mapsto \tr[T_1^*(t)(A)T_2(s)\rho] = \tr[AT_1(t)T_2(s)\rho]$ is $C^1$.
    Therefore:
    \begin{align*}
        \abs{\tr[ A(T_1(t)-T_2(t))\rho]} 
        &= \bigg|\int_0^t \frac{d}{ds} \tr[AT_1(t-s)T_2(s)\rho]\,ds\bigg| \\
        &=\bigg| \int_0^t \tr[\L_1^*(A)T_1(t-s)T_2(s)\rho] - \tr[AT_1(t-s)T_2(s)\L_2\rho]\,ds\bigg|\\
        &=\bigg| \int_0^t \tr[AT_1(t-s)(\L_1-\L_2)T_2(s)\rho]\,ds\bigg|\\
        &\le \int_0^{t} \norm{(\L_1-\L_2)T_2(s)}_{\diamond,E}\,ds
        % &\le \int_0^{t}\tfrac{f_{T_2(s)}(E)}E\, ds\,\norm{\L_1-\L_2}_{\diamond,E} \\
        % &\le \int_0^{t} \Big( e^{\omega s} + (e^{\omega s}-1)\tfrac{E_0}{E} \Big) \norm{\L_1-\L_2}_{\diamond,E}\, ds \\
        % &= \int_0^{t} \Big(1+(e^{\omega s}-1)(1+\tfrac{E_0}{E}) \Big) \norm{\L_1-\L_2}_{\diamond,E}\, ds \\
        % &= \big( t + \tfrac{e^{\omega t}-1}\omega (1+\tfrac{E_0}E) \big) \norm{\L_1-\L_2}_{\diamond,E}.
        \le \int_0^{t} \norm{\L_1-\L_2}_{\diamond,f_s(E)}\,ds
        \le t \norm{\L_1-\L_2}_{\diamond,f_t(E)}.
    \end{align*}
    If we optimize over operators $A\in\dom\L_1^*$ with norm $\le 1$, we obtain $\norm{T_1(t)\rho-T_2(t)\rho}_1 \le t\norm{\L_1-\L_2}_{\diamond,f_t(E)}$.
    % \begin{equation*}
    %     \norm{T_1(t)\rho-T_2(t)\rho}_1 \le \big( t + \tfrac{e^{\omega t}-1}\omega (1+\tfrac{E_0}E) \big) \norm{\L_1-\L_2}_{\diamond,E}.
    % \end{equation*}
    The same reasoning applies to the semigroups $T_1(t)\ox\id$ and $T_2(t)\ox\id$ and states $\rho\in\D\odot\T(\H_R)$, where $\H_R$ is another Hilbert space.
    Therefore, the claimed bound follows.
\end{proof}

\subsection{Trotter product formula in open systems}\label{sec:trotter}

Here, we use the submultiplicativity estimate \eqref{eq:submultiplicativity2} to lift the proof of operator norm convergence rates of Trotter convergence from the finite-dimensional setting to the infinite-dimensional setting.
The idea for this was developed for \cite{trotter_maths}, where unitary dynamics are treated.\footnote{This application was the author's original motivation for investigating energy-limited dynamics.
Strong error bounds for the Trotter product formula in dimension have recently been studied in \cite{mobus_strong_2024,burgarth_state-dependent_2023,lie_group_error,burgarth_strong_2024,hahn_lower_2024}.
}

\begin{proposition}
    Let $\L_1$ and $\L_2$ be generators of energy-limited dynamical semigroups $\{T_j(t)\}_{t\ge0}$, $j=1,2$ with joint stability constants $\omega,E_0\ge0$.
    Let $\D\subset\dom\energy$ be a $\normiii{}_1$-dense $T_1(t)$- and $T_2(t)$-invariant subspace with the property that $(t,s)\mapsto T_1(t)T_2(s)\rho$ and $(t,s)\mapsto T_2(t)T_1(s)\rho$ are $C^2$ functions for all $\rho\in\D$.

    Assume that the commutator $[\L_1,\L_2]:\D\to\T(\H)$ has finite ECD norm.
    If there exists an extension
    $\L\supseteq \L_1+\L_2$ that generates a quantum dynamical semigroup $\{T(t)\}_{t\ge0}$, then there is a unique generating extension.
    In this case, the Trotter product formula converges with convergence rates bounded as
    \begin{equation}\label{eq:trotter_estimate}
        \norm{( T_1(t/n)T_2(t/n))^n - T(t)}_{\diamond,E} \le \frac{t^2}{2n} \norm{[\L_1,\L_2]}_{\diamond,f_{2t}(E)},
    \end{equation}
    where $f_t(E) = E+ (e^{\omega t}-1)(E+E_0)$.
\end{proposition}

Note that, by \eqref{eq:concavity_ineq_ECD}, the right-hand side of \eqref{eq:trotter_estimate} is bounded by $ \frac{t^2}{2n} \norm{[\L_1,\L_2]}_{\diamond,E} \cdot \big(1+ (e^{2\omega t}-1)(1+\tfrac{E_0}E)\big)$.
\begin{proof}
    We adapt the argument for the unitary case from \cite{trotter_maths}.
    Let us begin by noting that the assumptions guarantee that $\L_1\L_2$ and $\L_2\L_1$ are defined on $\D$ since they arise as second-order derivatives of $T_1(t)T_2(s)$ and $T_2(t)T_1(s)$ at $(t,s)=(0,0)$.
    Therefore, the commutator makes sense as an operator on $\D$ and, by item \ref{it:ECDnorm_dense} of \cref{lem:ECD}, it canonically extends to an operator $\dom\energy\to\T(\H)$ (with the same ECD norm).
    Furthermore, the assumptions guarantee that $\D\subset\dom\L$ for all generating extensions $\L$.
    
    We begin with the usual telescoping sum trick.
    Set $V(t)=T_1(t)T_2(t)$ and $f_t(E)=e^{\omega t}(E+E_0)-E_0$. Then
    \begin{align} 
        \norm{V(t/n)^n -T(t)}_{\diamond,E} 
        & = \bigg\|\sum_{j=1}^n T\big(t(j+1)/n\big) \big(V(t/n)-T(t/n)\big) V(t/n)^{n-j} \bigg\|_{\diamond,E} \nonumber\\
        & \le \sum_{j=1}^n \big\| \big(V(t/n)-T(t/n)\big)V(t/n)^{n-j}\big\|_{\diamond,E} \nonumber\\
        & \le \sum_{j=1}^n \|V(t/n)-T(t/n)\|_{\diamond,f_{2t(n-j)/n}(E)} \nonumber\\
        & \le n\|V(t/n)-T(t/n)\|_{\diamond,f_{2t-2t/n}(E)}.\label{eq:brot} 
    \end{align}
    This reduces the problem to estimating $\norm{V(t)-T(t)}_{\diamond,E}$ for small times $t>0$.
    Next, we show the identity
    \begin{equation}\label{eq:leaf}
        [\L_2,T_1(s)]\rho =\int_0^s T_1(s-u) [\L_2,\L_1] T_1(u)\rho\,du,\qquad \rho\in\D.
    \end{equation}
    Note that the integral makes sense since the integrand is continuous by assumption.
    Formally the integrand is simply $(d/du)T_1(s-u)\L_2T_1(u)\rho$. However, we are not guaranteed that this function is differentiable.
    Since $\L_2T_1(u)\rho$ is differentiable, this is solved by taking a dual pairing with an operator $A\in \dom\L_1^*$ ($\L_1^*$ is the generator of the dual semigroup $T_1^*(t)$ which is strongly continuous for the $\sigma$-weak operator topology):
    \begin{align*}
        \tr\!\big[A[\L_2,T_1(s)]\rho\big] 
        &= \int_0^s \frac d{du} \tr\!\Big[T_1^*(s-u)(A)\,\L_2T_1(u)\rho\Big] du \\
        &=\int_0^s\tr\!\Big[\L_1^*T_1^*(s-u)(A)\,\L_2T_1(u)\rho-T_1^*(s-u)(A)\,\L_2\L_1T_1(u)\rho\Big] du\\
        &=\int_0^s \tr\!\Big[A\,T_1(s-u)[\L_1,\L_2]T_1(u)\rho\Big ]du.
    \end{align*}
    Since $\dom\L_1^*$ is $\sigma$-weakly dense in $\B(\H)$, this shows that \eqref{eq:leaf} holds.
    We apply the same trick to $V(t)-T(t)$.
    If $A\in\dom\L^*$ and $\rho\in\D$, we find
    \begin{align*}
        \tr\!\big[A\,(V(t)\rho-T(t)\rho)\big] 
        &= \int_0^t \frac{d}{ds} \tr\!\Big[T^*(t-s)(A)\,T_1(s)T_2(s)\rho\Big]ds\\
        &= \int_0^t \tr\!\Big[T^*(t-s)(A)\,T_1(s)(\L_1+\L_2)T_2(s)\rho\Big] -\tr\!\Big[\L^*T^*(t-s)(A)\,T_1(s)T_2(s)\rho\Big] ds\\
        &= \int_0^t \tr\!\Big[A\,T(t-s)T_1(s)(\L_1+\L_2)T_2(s)\rho\Big]  -\tr\!\Big[A\,T(t-s)(\L_1+\L_2)T_1(s)T_2(s)\rho\Big] ds\\
        % &= \int_0^t  \tr\!\big[A\,T(t-s)\big(\L T_1(s)-T_1(s)(\L_1+\L_2)\big)T_2(s)\rho\big]ds\\
        % &= \int_0^t \tr\!\big[A\,T(t-s)[\L,T_1(s)]T_2(s)\rho\big]ds\\
        &= \int_0^t \tr\!\Big[A\,T(t-s)[T_1(s),\L_2]T_2(s)\rho\Big] ds\\
        &=\int_0^t\int_0^s \tr\!\Big[A\,T(t-s)T_1(s-u)[\L_1,\L_2] T_1(u)T_2(s)\rho\Big] du\,ds,
    \end{align*}
    where we used \eqref{eq:leaf} in the last step.
    Now assume $\rho\in\states_E\cap\D$.
    Since $\dom\L^*$ is $\sigma$-weakly dense, optimizing over $A\in\dom\L^*$ with $\norm A\le 1$ gives
    \begin{align*}
        \norm{V(t)\rho-T(t)\rho}_1 
        &\le \int_0^t\int_0^s \norm{  T(t-s)T_1(s-u)[\L_2,\L_1] T_1(u)T_2(s)\rho}_1\,du\,ds\\
        &\le \int_0^t\int_0^s \norm{[\L_2,\L_1] T_1(u)T_2(s)\rho}_1\,du\,ds\\
        &\le \int_0^t\int_0^s \norm{[\L_1,\L_2]}_{\diamond,f_{s+u}(E)}\,du\,ds
        \le \frac{t^2}2 \norm{[\L_1,\L_2]}_{\diamond,f_{2t}(E)}
        ,
    \end{align*}
    where, in the last step, we used that $\omega,E_0$ are joint stability constants.
    The same argument applies to $\H\ox\H_R$, $T_j(t)\ox \id$, $j=1,2$, and $\rho\in \D\odot \T(\H_R)$, where $\H_R$ is another Hilbert space and "$\odot$" denotes the algebraic tensor product.
    The $\normiii{}_1$-density assumption guarantees that $\D\odot \T(\H_R)\subset\dom\tilde\energy$ is similarly dense for the corresponding norm induced by the reference Hamiltonian $\tilde G=G\ox 1$ on $\H\ox\H_R$.
    Since $\D$ is $\normiii{}_1$-dense, the above establishes the estimate
        $\norm{V(t)-T(t)}_{\diamond,E} \le \frac{t^2}2 \norm{[\L_1,\L_2]}_{\diamond,f_{2t}(E)}$.
    If we insert this in \eqref{eq:brot}, we get
    \begin{align*}
        \norm{V(t/n)^n-T(t)}_{\diamond,E}
        \le n \frac{t^2}{2n^2} \norm{[\L_1,\L_2]}_{\diamond,f_{2t/n}(f_{2t-2t/n}(E))}
        % &\le n \int_0^{t/n}\int_0^s \norm{[\L_1,\L_2]}_{\diamond,\underbrace{f_{2t-2t/n+s+u}(E)}_{\le f_{2t}(E)}}\,du\,ds 
        % &\le n\int_0^{t/n}\int_0^s\norm{[\L_1,\L_2]}_{\diamond,f_{2t}(E)}\,ds\,du 
        = \frac{t^2}{2n} \norm{[\L_1,\L_2]}_{\diamond,f_{2t}(E)}.
    \end{align*}
    Eq.~\eqref{eq:trotter_estimate} now follows from \eqref{eq:concavity_ineq_ECD}.
    Recall that convergence in ECD norm implies strong convergence. Since \eqref{eq:trotter_estimate} holds for all extensions $\L\supset \L_1+\L_2$ that generate dynamical semigroups, all such extensions generate the same dynamical semigroup and, hence, coincide.
\end{proof}

\section{Open problems}\label{sec:problems}

In the following, we discuss open questions, possible generalizations and ideas for future work.

\paragraph{Limited energy loss.}
By definition, energy-limited quantum channels are channels with controlled energy increase.
Let us consider channels $T$ from system $A$ to $B$ with controlled energy loss.
While energy-limitedness is equivalent to the input energy bounding the output energy, limited energy loss asks for the reverse inequality, i.e., the output bounds the input energy.
To quantify this, one can introduce the function
\begin{equation}
    g_T(E) = \inf\Big\{ \energy[T\rho] : \rho\in\states(\H_A),\ \energy[\rho]\ge E\Big\}.
\end{equation}
This is a convex nondecreasing function. 
Let us say that a quantum channel $T$ has \emph{limited energy loss} if $g_T(E)\to\oo$ as $E\to\oo$.
By convexity, this is indeed equivalent to the existence of an affine lower bound $g_T(E)\ge \lambda E-E_0$ for constants $\lambda>0$, $E_0\ge0$, which is equivalent to the operator inequality
\begin{equation}
    T^*(G_B) \ge \lambda G_A -E_0.
\end{equation}
On physical grounds, requiring finite energy loss might not sound as compelling as requiring a finite energy gain.
For instance, channels $\rho \mapsto \omega_0$, which reset the state of the system to, say, the ground state state $\omega_0$, clearly take an unbounded amount of energy away.
However, ground state preparation is extremely hard to perform in practice.
We expect that a theory of dynamics with limited energy loss can be done in parallel to the limited energy increase that we considered in the main text.
Furthermore, we expect dynamics that have both limited energy gain and loss to be particularly well-behaved.

\paragraph{Escape to infinity and (non)standard generators.}
In infinite dimensions, generators of quantum dynamical semigroups are still not fully understood.
To the best of the author's knowledge, all Markov semigroups used in actual models of open quantum systems have standard generators.%\footnote{I would be grateful for comments on this.}
However, we cannot conclude that nonstandard generators are unphysical since it may be our ignorance that keeps us from using them in models.
Here, we consider whether the physically meaningful property of energy-limitedness might be related to the generator's standardness.
In the special case of the quantum birth process, we saw that this is indeed the case (see \cref{sec:birth}).

Clearly, every dynamical semigroup is energy-limited with respect to every bounded reference Hamiltonian, e.g., $G=1$.
Even for unbounded reference Hamiltonians, energy-limitedness might hold trivially, e.g., if $G$ is only unbounded on a subsystem on which the dynamics is trivial.
To avoid such artificial cases, let us assume that the reference Hamiltonian is of the form \eqref{eq:discrete}, i.e., has compact resolvent.
We consider the following problem, suggested to the author by Andreas Winter:

\begin{problem*}\label{prob1}
    Are energy-limited quantum dynamical semigroups necessarily generated by standard generators?
\end{problem*}

In view of the previous paragraph, it might be necessary to assume additionally limited energy loss.
% The reason for restricting the question to conservative dynamics will become apparent shortly.
% Recall that formally conservative standard generators, i.e., where $-(K^*+K)=\sum_\alpha L_\alpha^*L_\alpha$ holds in the sense of \eqref{eq:L_summation_eq}, need not generate conservative dynamics.
% We may interpret such a failure of conservativity as particles escaping to infinity in finite time.
% Our intuition is that since the reference energy diverges at infinity, the energy of an escaping particle has to become infinite at a finite time as well.
% Thus, we expect that energy-limited dynamics cannot allow for an escape in finite time.
% 
A method of constructing nonstandard generators is to take a formally conservative standard generator with nonconservative dynamics and to reset the system to a state $\chi$ whenever an escape occurs \cite{siemon_unbounded_2017,inken_thesis}. On the infinitesimal level this is a perturbation $\L'=\L - \tr[\L(\placeholder)]\chi$ (see \cref{sec:birth} and \cite{siemon_unbounded_2017,inken_thesis}).\footnote{The proof for $\L'$ being nonstandard in \cite{siemon_unbounded_2017} is valid for all strongly standard generators (see \cite[Def.~4.4.3]{inken_thesis}).}
In this construction, we may choose $\chi$ freely. 
Thus, if we pick a finite-energy state, \cref{thm:main} implies that $\L'$ is energy-limited if and only if $\L$. 
Thus, the existence of an energy-limited dynamical semigroup with escape to infinity leads to a nonstandard energy-limited semigroup.
Therefore, an affirmative answer to the Problem above would imply that energy-limitedness prohibits escape to infinity -- at least for "strongly standard" generators where $K$ and $L_\alpha$ satisfy a closability assumption \cite{siemon_unbounded_2017}.
Perhaps surprisingly, this has been studied in a paper by Chebotarev and Fagnola \cite{chebotarev1998sufficient} (see also \cite[Sec.~3.6]{fagnola1999}).
They show that a formally conservative standard generator that satisfies the infinitesimal energy-limitedness inequality with respect to some reference Hamiltonian admits no escape in finite time.%
\footnote{The infinitesimal version of energy-limitedness only appears as a sufficient mathematical condition in their work. It is not studied in its own right and it is not interpreted in the context of energies.}
In addition to infinitesimal energy-limitedness, they require certain assumptions. One of these is that $F=-(K+K^*)$ is a self-adjoint operator dominated by the reference Hamiltonian, which is an infinitesimal version of limited energy loss.

In the special case of the quantum birth process (see \cref{sec:birth}): Energy-limitedness and the impossibility of escape to infinity are equivalent.
It is an interesting question whether this holds in general.

\paragraph{Energy scales on von Neumann algebras.}
Energy-limitedness and energy-constrained norms make sense for classical systems where the energy scale is determined by a reference Hamiltonian function.
In fact, we can go far beyond this:
If $\M$ is a von Neumann algebra, then a reference Hamiltonian is a positive self-adjoint operator affiliated with $\M$ or, equivalently, an element $\energy\in \overline\M^+$ of the extended positive cone (see \cref{sec:epc}).
Relative to a fixed reference energy scale, we can then define an ECO norm for elements of $\M$ and an ECD norm for $^*$-preserving maps $\M\to\N$ for some other von Neumann algebra.
This includes "ordinary" quantum systems $\M=\B(\H)$ as well as classical systems $\M=L^\oo(X,\mu)$ and hybrid systems.
However, it also covers the more exotic observable algebras appearing in quantum field theory and quantum statistical mechanics.
The question is, of course, whether such a generalization is useful for anything. 

\appendix

\section{Operator inequalities for Heisenberg-picture channels applied to unbounded operators}\label{sec:epc}

In this appendix, we review different ways of describing positive self-adjoint operators and explain how the action of a quantum channel in the Heisenberg picture can be extended to them.
We then explain how results from \cite{dinh2019}  allow us to extend certain operator inequalities to the case of unbounded operators.

We start by recalling the definition of quadratic forms (see \cite[Sec.~VIII.6]{ReedSimon1} for details):

\begin{definition}[Quadratic forms]
    A \emph{quadratic form} on a Hilbert space $\H$ is a sesquilinear map $a:Q(a)\times Q(a)\to\CC$ where $Q(a)\subset\H$ is a subspace called the form domain.
    If not explicitly said otherwise, we assume $Q(a)$ to be dense.
    A quadratic form $a$ is \emph{positive} if $a(\psi,\psi)\ge0$ for all $\psi\in Q(a)$ and a positive quadratic form is \emph{closed} if $Q(a)$ is complete under the norm $\norm{\psi}_{Q(a)}=\sqrt{\norm\psi^2+a(\psi,\psi)}$.
\end{definition}

In the following, we only consider positive quadratic forms.
By polarization, a quadratic form $a$ is uniquely defined by the numbers $a(\psi,\psi)$, $\psi\in Q(a)$.
The norm $\norm{}_{Q(a)}$ on $Q(a)$ is the norm induced by the inner product $(\psi,\phi)\mapsto \ip\psi\phi+a(\psi,\phi)$.
Thus, a positive quadratic form $a$ is closed if and only if this inner product turns $Q(a)$ into a Hilbert space.
A subspace $\D\subset Q(a)$ is called a \emph{form core} if $\D$ is dense in $Q(a)$ with respect to the norm $\norm{}_{Q(a)}$.
A positive quadratic form $a$ is \emph{closable} if it admits a closed extension. In this case, it admits a smallest closed extension, called the \emph{closure} and denoted by $\bar a$.

\begin{theorem}[{\cite[Sec.~VIII.6]{ReedSimon1}}]\label{thm:forms}
    If $A=A^*\ge0$ is a positive self-adjoint operator, then it defines a closed quadratic form $a$ by 
    \begin{equation}\label{eq:a_from_A}
        Q(a)=\dom\sqrt A,\qquad a(\psi,\phi)=\ip{\sqrt A\psi}{\sqrt A\phi}.
    \end{equation}
    Every closed positive quadratic form arises from a unique positive self-adjoint operator in this way.
\end{theorem}

That \eqref{eq:a_from_A} is indeed closed is easy to see.
We briefly explain how to construct the positive self-adjoint operator $A$ inducing a given closed quadratic form $a$: 
On the domain $\dom A = \{\psi\in Q(a): \exists_{\tilde\psi\in Q(a)}  \forall_{\phi\in Q(a)} \ a(\phi,\psi)=\ip{\phi}{\tilde \psi}\}$ the operator $A$ is now defined on $\dom A$ by $A\psi=\tilde \psi$.
Clearly, $A$ is symmetric, and it is not too hard to check explicitly that $\dom A^*=\dom A$. 

For positive quadratic forms $a_1$ and $a_2$, the order relation $a_1\le a_2$ is defined by
\begin{equation}\label{eq:order_forms}
    Q(a_1)\supseteq Q(a_2)\qandq a_1(\psi,\psi)\le a_2(\psi,\psi), \ \psi\in Q(a_2).
\end{equation}
Formulated in terms of the corresponding positive self-adjoint operators $A_1$ and $A_2$, this is precisely the definition of $A_1\le A_2$ used in the main text (see \cref{eq:order})

Next, we consider a concept from von Neumann algebra theory (see \cite[Sec.~X.4]{takesaki2} for details):

\begin{definition}
    The \emph{extended positive cone} $\overline{\B(\H)}{}^+$ of $\B(\H)$ is the set of lower semicontinuous maps $m:\TC^+\to\bar\RR^+$ such that $m(\lambda\rho)=\lambda m(\rho)$ and $m(\rho+\sigma)=m(\rho)+m(\sigma)$ for all $\lambda\ge0$, $\rho,\sigma\in\TC^+$. 
    An element $m\in\overline{\B(\H)}{}^+$ is called semifinite if $\{\rho\in\TC^+: m(\rho)<\oo\}$ is dense in $\T(\H)^+$.
\end{definition}

Every bounded positive operator $A\in\B(\H)^+$ corresponds to an element $m$ of the extended positive cone via $m(\rho)=\tr A\rho$.
From the duality $\B(\H)=\TC^*$ it follows that $\B(\H)^+\hookrightarrow\overline{\B(H)}^+$ contains precisely the finite elements, i.e.\ those $m\in\overline{\B(H)}^+$ which never evaluate to infinity.
The extended positive cone $\overline{\B(\H)}{}^+$ for the one-dimensional Hilbert space $\H=\CC$ can be identified with $\bar\RR^+$.

\begin{theorem}[{\cite[Sec.~X.4]{takesaki2}}]
    There is a bijection between the following objects
    \begin{enumerate}[(i)]
        \item elements of the extended positive cone $m\in\overline{\B(\H)}{}^+$,
        \item pairs $(A,\K)$ of a closed subspace $\K\subseteq \H$ and a positive self-adjoint operator $A:\K\supseteq\dom A\to\K$,
        \item projection-valued Borel measures $P$ on $\bar\RR^+$,
    \end{enumerate} 
    given by the following:
    $P|_{\RR^+}$ is the spectral measure of $A$ and $P(\{\oo\})$ is the projection onto $\K^\perp$.
    Conversely, $\K$ is the orthogonal complement of $P(\{\oo\})\H$ and $A = \int_0^\oo x\, dP(x)$.
    $m$ can be obtained from $P$ via
    \begin{equation}
        m(\rho)=\int_0^\oo \lambda \tr[\rho\,dP(\lambda)] + \tr[(1-P)\rho]\cdot\oo.
    \end{equation}
    The subspace $\K$ is related to $m$ via $\K=\bar{\{\psi\in\H : m(\kettbra\psi)<\oo\}}$.
    On $\K$ a closed quadratic form with form domain $Q = \{\psi\in\H: m(\kettbra\psi)<\oo\}\subseteq\K$ is defined by polarization from $m$ and $A$ is the unique positive self-adjoint operator corresponding to it.
    Furthermore, semifiniteness is equivalently characterized by
    \begin{equation}
        m\ \text{is semifinite} \iff \K=\H \iff P(\{\oo\})=0.
    \end{equation}
\end{theorem}

Abusing notation, the correspondence between $m$ and $(A,\K)$ can be summarized by $m = A \oplus  \oo$ relative to $\H = \K\oplus \K^\perp$.
The energy functional $\energy[\placeholder]$ induced by a positive reference Hamiltonian $G$ used in the main text is the semifinite element of the extended positive cone corresponding to the pair $(G,\H)$.
For completeness, we mention two further equivalent characterizations of elements of the extended positive cone:%
\footnote{From (iv), one obtains a pair $(A,\K)$ via $\K=\bar{Q(a)}$, where the closure is taken with respect to the norm topology on $\H$, and $A$ is the positive self-adjoint operator on $\K$ inducing the closed densely defined quadratic form $a$ (now viewed as a form on $\K$). The characterization (v) is connected directly to elements $m$ in the extended positive cone via $m(\rho) = \tr\rho \cdot h(\rho /\tr\rho)$ if $\rho\ne0$ and $m(0)=0$.}
\begin{enumerate}[resume*]\it
    \item[(iv)]\label{it:eps4} closed positive quadratic forms $a$ on $\H$ which are not-necessarily densely defined.
    \item[(v)]\label{it:eps5} affine lower semicontinuous functionals $h:\states(\H)\to\bar\RR^+$ on the state space.
\end{enumerate} 

The main advantage of the extended positive cone is that it makes sense to define the sum and the semidefinite ordering on all pairs of elements of the full extended cone:
For $m_1,m_2\in \overline{\B(\H)}{}^+$ and $\lambda\ge0$, the element $m_1+\lambda m_2\in \overline{\B(\H)}{}^+$ is defined by $(m_1+\lambda m_2)(\rho)=m_1(\rho)+\lambda m_2(\rho)$, and the order relation $m_1\le m_2$ is defined via $m(\rho)\le n(\rho)$ for all $\rho\in\TC^+$.
In contrast, linear combinations and order relations can only be defined in the realm of positive self-adjoint operators if certain domain assumptions are met.
If $m_1,m_2$ are semifinite and correspond to positive self-adjoint operators $A_1,A_2$, respectively, then $m_1\le m_2$ if and only if $A_1\le A_2$ (see \eqref{eq:order}).
Furthermore, if $m_1+m_2$ is semifinite, it corresponds to the form sum $A_1\dot+A_2$ \cite[Ap.~A.9]{takesaki2}.
We also need the following notions:
\begin{itemize}
    \item For $K\in\B(\H,\K)$ and $m\in \overline{\B(\K)}^+$, $K^*m K\in \overline{\B(\H)}{}^+$ is defined by $(K^*mK)(\rho)=m(K\rho K^*)$.
    \item For $m_i\in\overline{\B(\H_i)}^+$, $i=1,2$, we define $m_1\ox m_2\in\overline{\B(\H_{1}\!\ox\!\H_{2})}^+$ via the corresponding pairs $(A_i,\K_i)$ as the element corresponding to the pair $(A_1\ox A_2,\K_1\ox\K_2)$
    \item If $T:\T(\H_A)\to\T(\H_B)$ is a (bounded) positive linear map and $m\in\overline{\B(\H_B)}^+$, define $T^*m\in\overline{\B(\H_A)}^+$ via $T^*m(\rho) = m(T\rho)$.
    \item For $m\in\overline{\B(\H)}{}^+$ and a Borel function $f:\bar\RR^+\to\bar\RR^+$ and $m\in\overline{\B(\H)}{}^+$, define $f(m)\in\overline{\B(\H)}{}^+$ via the associated Borel measure $P$ of $m$ as the element whose associated Borel measure is the push-forward measure $f_*P$, i.e., $$f(m)(\rho) =\int_0^\oo f(\lambda) \tr[\rho dP(\lambda)]+ \tr[(1-P)\rho]\cdot f(\oo).$$ 
        Therefore, $f(m)$ is semifinite if and only if $f^{-1}(\{\oo\})$ is an $P$-null set. In particular, $f(m)$ is semifinite if no element is mapped to infinity.
\end{itemize}

If $f:\RR^+\to\RR^+$ is a monotone function, we define an extension $f:\bar\RR^+\to\bar\RR^+$ by setting $f(\oo)=\sup f$. This extension is a Borel function.

\begin{lemma}[{\cite{dinh2019}}]\label{thm:op_monotone}
    Let $f:\RR^+\to\RR^+$ be an operator-monotone function and let $m,n\in\overline{\B(\H)}{}^+$. Then
    \begin{equation}
        m\le n \implies f(m)\le f(n).
    \end{equation}
\end{lemma}

\begin{lemma}[{\cite{dinh2019}}]\label{thm:Hansen_ineq}
    Let $f:\RR^+\to\RR^+$ be an operator-monotone function and let $K:\H\to\K$ be a linear contraction.
    Then 
    \begin{equation}
        K^*f(m) K \le f(K^*mK),\qquad m\in\overline{\B(\K)}^+.
    \end{equation}
\end{lemma}
\begin{proof}
    This result is proved in \cite{dinh2019} for the case that $\K=\H$, but the same proof works in the general case.
\end{proof}

\begin{corollary}\label{thm:Stinespring_lemma}
    Let $T:\T(\H_A)\to\T(\H_B)$ be completely positive and let $(V,\K)$ be a Stinespring dilation, i.e.\ $V\in\B(\H_A,\H_B\ox\K)$ is such that $T\rho=\tr_\K V\rho V^*$ for all $\rho\in\T(\H_A)$.
    Then
    \begin{equation}
        T^*m = V^* (m\ox 1) V,\qquad m\in\overline{\B(\H_B)}^+.
    \end{equation}
\end{corollary}

\begin{corollary}\label{thm:op_monotone_quantum_op}
    Let $f:\RR^+\to\RR^+$ be an operator-monotone function and let $T:\T(\H_A)\to\T(\H_B)$ be a completely positive trace-nonincreasing map.
    Then
    \begin{equation}
        T^*f(m) \le f(T^*m),\qquad m\in\overline{\B(\K)}^+.
    \end{equation}
\end{corollary}
\begin{proof}
    This follows from combining \cref{thm:Hansen_ineq,thm:Stinespring_lemma}.
\end{proof}

\section{A generation theorem for dissipative operators on Hilbert spaces}
\label{sec:generation}

In this appendix, we present a generation theorem for dissipative operators on Hilbert spaces.
The core idea is to use infinitesimal energy-limitedness to verify the assumptions of the Lumer-Phillips generation theorem.
The class of dissipative generators satisfying the assumptions of this theorem is closed under summation. By restricting to skew-hermitian operators, we obtain Nelson's commutator theorem as a special case.

Recall that an operator $K:\H\supseteq \dom K\to\H$ is called dissipative if 
\begin{equation*}
    \Re \ip\psi{K\psi} \le 0, \qquad \psi\in \dom K.
\end{equation*}
Dissipativity captures infinitesimally that an operator generates a contraction semigroup, i.e., a strongly continuous one-parameter semigroup of linear contractions on $\H$.
However, not all dissipative operators are generators (like not all symmetric operators are self-adjoint).
Those that are generators are precisely the maximally dissipative operators, i.e., dissipative operators that admit no proper dissipative extensions \cite{arendt_extensions_2023}.

If a given dissipative operator is not a generator one must one has to find a generating extension.\footnote{Every dissipative operator admits a maximally dissipative extension \cite{arendt_extensions_2023}. This is in contrast to symmetric operators that only admit self-adjoint extensions if their defect indices are equal.}
In good cases, there is a unique generating extension, namely the closure $\bar K$ (this corresponds to essentially self-adjoint operators).
The following theorem provides sufficient conditions for this:

\begin{theorem}\label{thm:generation}
    Let $\H$ be a Hilbert space and let $N\ge0$ a self-adjoint operator with core $\D$.
    Let $K: \D\to\H$ be a dissipative $N$-bounded operator and let $\omega>0$ such that
    \begin{equation}\label{eq:generation1}
        \ip{K\psi}{N\psi}+\ip{N\psi}{K\psi} \le \omega \ip\psi{N\psi},\qquad \psi\in \D,
    \end{equation}
    Then $\bar K$ generates a contraction semigroup, $\dom\bar K\supseteq \dom N$, and every core for $N$ is a core for $\bar K$. 
\end{theorem}

\cref{thm:dissipative} in the main text shows that the assumptions furthermore imply 
\begin{equation}\label{eq:generation2}
    \norm{N^{\frac12}e^{t\bar K}\psi} \le e^{\omega t/2} \norm{N^{\frac12}\psi}, \qquad \psi\in \dom N.
\end{equation}
In the case of a skew-symmetric operator $K$, our theorem implies Nelson's Commutator Theorem \cite{nelson_time-ordered_1972}:

\begin{corollary}[Nelson's Commutator Theorem]
    Let $\H$ be a Hilbert space and let $N\ge0$ a self-adjoint operator with core $\D$.
    Let $H: \D\to\H$ be a symmetric $N$-bounded operator such that
    \begin{equation}
        \pm i\big(\ip{H\psi}{N\psi}-\ip{N\psi}{H\psi}\big) \le \omega \ip\psi{N\psi}, \qquad \psi\in\D,
    \end{equation}
    for some $\omega>0$. Then $H$ is essentially self-adjoint, $\dom\bar H\supseteq \dom N$, and every core for $N$ is a core for $\bar H$.
\end{corollary}

\begin{proof}
    Since $H$ is symmetric, $K=\pm i H$ is dissipative.
    Applying \cref{thm:generation} to both operators shows that $i \bar H$ and $-i\bar H$ generate contraction semigroups. 
    Since these semigroups are adjoints of each other, both are unitary, and $\bar H$ is self-adjoint.
\end{proof}

We note an immediate consequences of our result:
The class of dissipative operators that satisfy the assumptions or \cref{thm:generation} for a core $\D$ for $N$ is closed under positive linear combinations.
Thus, if we can decompose a given operator $K$ into a real and an imaginary part, it suffices to check the conditions for these parts separately:

\begin{corollary}\label{corollary}
    Let $\H$, $N$ and $\D$ be as in \cref{thm:generation}. 
    Let $H,P:\D\to\H$ be symmetric $N$-bounded operators such that $\ip\psi{P\psi}\ge0$ for all $\psi\in\D$.
    If 
    \begin{equation}\label{eq:corollary1}
        -i(\ip{H\psi}{N\psi}-\ip{N\psi}{H\psi}) \le \omega \ip\psi{N\psi}, \qquad \psi\in\D,
    \end{equation}
    and 
    \begin{equation}\label{eq:corollary2}
        \ip{P\psi}{N\psi} + \ip{P\psi}{K\psi} \le \omega \ip\psi{N\psi},\qquad\psi\in\D,
    \end{equation}
    then the closure of $K=iH-P$ generates a contraction semigroup.
    In fact, the same holds for $(-i\alpha H-\beta P)$ for all $\alpha,\beta>0$.
\end{corollary}

As a consequence of \cref{corollary} and the Chernoff product formula \cite[Thm.~III.5.2]{EngelNagel}, we get the following:

\begin{corollary}
    Let $\H$, $N$ be as in \cref{thm:generation} and let $K_1,\ldots K_m$ be operators satisfying the assumptions of \cref{thm:generation} and set $K=\sum_i \bar K_i$. Then
    \begin{equation}
        \big\|\big(e^{t\bar K_1/n}\cdots e^{t \bar K_m/n}\big)^n \psi - e^{t\bar K}\psi\big\| \to 0, \qquad \psi\in\H.
    \end{equation}
\end{corollary}

We now come to the proof of \cref{thm:generation}. The proof is based on Nelson's original argument to check the conditions of the Lumer-Phillips Theorem.

\begin{proof}[Proof of \cref{thm:generation}]
    \emph{Step 1.} 
    Since $K$ is dissipative, it is closable and the closure $\bar K$ is dissipative as well \cite[Thm.~4.5]{pazy1983}.
    Since $K$ is $N$-bounded and since $\D$ is a core for $N$, we have $\dom\bar K\supseteq \dom N$.
    Another consequence of $N$-boundedness is that \eqref{eq:generation1} remains true if $K$ is replaced by $\bar K$ and $\D$ is replaced by $\dom N$. To see this let $(\psi_n)$ be an $N$-graph norm Cauchy sequence in $\D$ with limit $\psi\in \dom N\subseteq \dom\bar K$, and note that $K\psi_n$ is a Cauchy sequence in $\H$ because $\norm{K\psi_n-K\psi_m}$ is bounded by a multiple of $\norm{N(\psi_n-\psi_m)}$. Therefore,
    \begin{align*}
        \ip{\bar K\psi}{N\psi}+\ip{N\psi}{\bar K\psi} &= \lim_n \big(\ip{K\psi_n}{N\psi_n} + \ip{N\psi_n}{K\psi_n}\big) \nonumber\le \lim_n \omega \ip{\psi_n}{N\psi_n} = \omega \ip\psi{N\psi}
    \end{align*}
    
    \emph{Step 2.}
    So far, we have shown that the restriction of $\bar K$ to $\dom N$ satisfies the same assumptions as $K$ with the core $\D$ given by $\dom N$.
    Since the closure of this restriction is $\bar K$, we may simply assume that $\D=\dom N$ in the following.
    By the Lumer-Phillips Theorem \cite[Thm.~3.15]{EngelNagel}, $\bar K$ generates a contraction semigroup if and only if $(\lambda-K)\dom N = (\lambda-K)\dom N\subseteq\H$ is dense for some/all $\lambda>0$.
    Assume that $\phi\in \H$ is orthogonal to $[(\lambda-K)\dom N]$ and let $\psi = (1+N)^{-1}\phi \in \dom N$. Then $\ip\phi{(\lambda-K)\psi} =0$ or, equivalently,
        $\ip\phi{K\psi}= \lambda\ip\phi\psi$.
    Therefore, we have
    \begin{align*}
        0\le \lambda \ip\psi{(1+N)\psi} 
        = \lambda \Re\ip\phi{\psi} 
        &= \Re\ip\phi{K\psi} \nonumber\\
        &= \Re\ip{(1+N)\psi}{K \psi} \nonumber\\
        &= \Re \ip{N\psi}{K \psi} + \Re \ip\psi{K\psi} \nonumber\\
        &\le \frac12 \big(\ip{K\psi}{N\psi}+\ip{N\psi}{K\psi}\big)+0\nonumber\\
        &\le \frac\omega2 \ip\psi{N\psi} < \frac\omega2 \ip\psi{(1+N)\psi}.
    \end{align*}
    For $\lambda>\frac\omega2$, this implies $\ip\psi{(1+N)\psi}=0$ and hence $\psi=0$. 
    Therefore, $(\lambda-K)\dom N$ must be dense.

    \emph{Step 3.} It remains to show that every core for $N$ is a core for $\bar K$.
    If $\D'$ is another core for $N$, we can run through the above arguments to show that the closure of $P:=K\upharpoonright\D'$ generates a contraction semigroup. 
    Since $\dom \bar P\supseteq \dom N\supseteq \dom P=\D'$, $\dom N$ is a core for $\bar P$ as well, and since $\bar K$ and $\bar P$ agree on a common core, we have $\bar P=\bar K$.
    Thus $\D'$ is a core for $\bar K$.
\end{proof}

We close with a comment on similar results in the literature.
Derek Robinson's work on commutator theorems \cite{robinson_commutator_1987,robinson1988,batty,robinson1989} contains generalizations of the commutator theorem that also cover dissipative operators.
However, these generalizations are different from our version in spirit.
While we replace the assumption $\pm i[H,N]\le \omega N$ in Nelson's commutator theorem by $K^*N+NK\le \omega N$, Robinson keeps the commutator by considering assumptions of the form $\abs{\ip\psi{[K,N]\phi}}\le \omega \norm{N^{\frac12}\psi}\norm{N^{\frac12}\phi}$.
The reason is that, in our case, $K^*N+NK$ measures the "infinitesimal energy gain" in the sense that (formally) $(d/dt) (e^{tK})^*Ne^{tK} |_{t=0}=K^*N+NK$ while the commutator shows up in Robison's work because it measures the noncommutativity of $e^{itN}$ and $e^{tK}$.

\printbibliography

\end{document}